\documentclass[oribibl]{llncs}

\usepackage{times}
\usepackage[dvipsnames, table]{xcolor}
\usepackage{amsmath}
\usepackage{amsthm}
\usepackage{amssymb}
\usepackage{mathrsfs}  
\usepackage[all]{xy}
\usepackage{multirow}
\usepackage{caption}
\usepackage{float}
\usepackage{multicol}
\usepackage{soul}
\usepackage{enumerate}
\usepackage{hhline}
\usepackage{bbm}
\usepackage{mathtools}
\usepackage{bm}
\usepackage{ dsfont }
\usepackage{pifont}
\usepackage{enumitem}
\usepackage{tikz}
\usepackage{array}
\usepackage{pdflscape}
\usepackage{mdframed}
\usepackage{thm-restate}
\usepackage{tikz-cd}
\usepackage{pgfplots}
\usepackage[T1]{fontenc}
\usepackage[ttdefault=true]{AnonymousPro}
\usepackage{adjustbox}
\usepackage[backend=biber]{biblatex}
\usepackage{lineno}
\usepackage{cleveref}

\usepackage[pass]{geometry}

\newcolumntype{P}[1]{>{\centering\arraybackslash}p{#1}}

\newcommand{\A}{\ensuremath{\mathcal{A}}}
\newcommand{\B}{\ensuremath{\mathcal{B}}}
\newcommand{\C}{\ensuremath{\mathcal{C}}}
\newcommand{\D}{\ensuremath{\mathcal{D}}}

\newcommand{\F}{\ensuremath{\mathcal{F}}}
\newcommand{\T}{\ensuremath{\mathcal{T}}}
\newcommand{\s}{\ensuremath{\sigma}}
\renewcommand{\S}{\ensuremath{\mathcal{S}}}
\renewcommand{\P}{\ensuremath{\mathcal{P}}}
\newcommand{\vars}{\textit{vars}}
\newcommand{\ax}{\textit{Ax}}
\newcommand{\wit}{\textit{wit}}
\newcommand{\overarrow}{\overrightarrow}

\newcommand{\dash}{\vdash}
\newcommand{\tdash}{\vdash_{\T}}
\newcommand{\tmodels}{\models_{\T}}
\newcommand{\smooth}{\textbf{SM}}
\newcommand{\stainf}{\textbf{SI}}
\newcommand{\convex}{\textbf{CV}}
\newcommand{\finwit}{\textbf{FW}}
\newcommand{\strfinwit}{\textbf{SW}}
\newcommand{\mc}{\textit{mincard}}

\newcommand{\TSM}{\T^{s}_{f}}
\newcommand{\Tgeqn}{\T_{\geq n}}

\newcommand{\TsM}{\T_{f}}
\newcommand{\Tinfty}{\T_{\infty}}

\newcommand{\Teven}{\T_{\mathit{even}}^{\infty}}

\newcommand{\Tninfty}{\T_{n, \infty}}

\newcommand{\Tone}{\T_{\leq 1}}
\newcommand{\Tleqn}{\T_{\leq n}}
\newcommand{\Tneqodd}{\T^{\neq}_{\mathit{odd}}}
\newcommand{\Tmn}{\T_{m,n}}
\newcommand{\Tneqoneinfty}{\T^{\neq}_{1,\infty}}
\newcommand{\Tneqtwoinfty}{\T^{\neq}_{2,\infty}}
\newcommand{\psiv}{\psi_{\vee}}
\newcommand{\addf}[1]{(#1)_{s}}
\newcommand{\adds}[1]{(#1)^{2}}
\newcommand{\addnc}[1]{(#1)_{\vee}}
\newcommand{\Taddf}{\addf{\T}}
\newcommand{\Tadds}{\adds{\T}}
\newcommand{\Taddnc}{\addnc{\T}}
\newcommand{\Ttwo}{\T_{2,3}}

\newcommand{\Tonetwo}{\T_{1}^{\infty}}
\newcommand{\Ttwotwo}{\T_{2}^{\infty}}
\newcommand{\Toddtwo}{\T_{1}^{odd}}

\newcommand{\subs}[1]{#1_{1}}
\newcommand{\subsf}[2]{#1_{#2}}
\newcommand{\plusf}[1]{s(#1)}
\newcommand{\subf}[1]{\texttt{0}(#1)}
\newcommand{\subff}[1]{#1_{\texttt{0}}}
\newcommand{\subncf}[2]{#1^{\dag}_{#2}}
\newcommand{\dagg}[1]{#1^{\dag}}

\usepackage{scalerel}
\newcommand\eq{\scaleobj{0.7}{=}}
\newcommand\diff{\scaleobj{0.7}{\neq}}

\newcommand{\set}[1]{\left\{{#1}\right\}}

\newcommand{\Ap}{\ensuremath{\mathcal{A}^{\prime}}}
\newcommand{\Bp}{\ensuremath{\mathcal{B}^{\prime}}}

\newcommand{\Exists}[1]{\exists\,#1.\:}
\newcommand{\Forall}[1]{\forall\,#1.\:}

\newcommand{\stafin}{\textbf{SF}}
\newcommand{\finmodpro}{\textbf{FM}}
\newcommand{\TBB}{\T_{\varsigma}}
\newcommand{\bb}{\varsigma}
\newcommand{\TBBtwo}{\T_{\varsigma}^{\infty}}
\newcommand{\TBBn}{\T_{n}^{\varsigma}}
\newcommand{\Tmninfty}{\T_{m,n}^{\infty}}
\newcommand{\Tupinfty}{\T^{\infty}}
\newcommand{\TsBB}{\T_{\bb}^{s}}
\newcommand{\qq}{\bb^{-1}}
\newcommand{\TtwoBB}{\T^{=}_{\bb}}
\newcommand{\TmnBB}{\T_{m,n}^{\bb}}
\newcommand{\TneqBB}{\T_{\bb}^{\neq}}
\newcommand{\TBBtwotwo}{\T_{\bb}^{2}}
\renewcommand{\t}{\tau}
\newcommand{\Sigmastwo}{\Sigmas^{2}}
\newcommand{\Sigmatwo}{\Sigma_{2}}
\newcommand{\Sigmas}{\Sigma_{s}}
\newcommand{\Sigmaone}{\Sigma_{1}}
\newcommand{\TssBB}{\T_{\bb}^{\vee}}
\newcommand{\TstwoBB}{\T_{\bb\vee}^{=}}
\newcommand{\TBBone}{\T_{1}^{\bb}}
\newcommand{\TneqBBone}{\T_{\bb,1}^{\neq}}
\newcommand{\Ttwocube}{\Ttwo^{3}}
\newcommand{\TBBtwocube}{\TBB^{\infty,3}}
\newcommand{\Tsupinfty}{\T^{\infty}_{\neq}}
\newcommand{\TsBBtwo}{\T_{\varsigma\neq}^{\infty}}
\newcommand{\Tneg}{\T_{\neg}}

\usepackage{authblk}

\title{Combining Finite Combination Properties: \\ Finite Models and Busy Beavers}

\author{Guilherme V. Toledo$^{\text{1}}$
\and Yoni Zohar$^{\text{1}}$
\and Clark Barrett$^{\text{2}}$
}\institute{
$^{\text{1}}$Bar-Ilan University
\enskip
$^{\text{2}}$Stanford University
}

\begin{document}


\setcounter{page}{1}     

\maketitle

\begin{abstract}
This work is a part of an ongoing effort to understand the relationships
between properties used in theory combination.
%
We here focus on including two properties that are related to shiny theories: 
the finite model property and stable finiteness.
For any combination of properties, 
we consider the question of whether there exists a theory that exhibits it.
When there is, we provide an example with the simplest possible signature.
One particular class of interest includes theories with  the finite model property that are not
finitely witnessable.
To construct such theories, we utilize the Busy Beaver function.%
%
%
\footnote{Funded in part by NSF-BSF grant numbers 2110397 (NSF) and 2020704 (BSF) and ISF grant number 619/21. 
}
\end{abstract}

\section{Introduction}
The story of this paper begins with~\cite{DBLP:journals/jar/ShengZRLFB22}, 
where it was shown that the theory of algebraic datatypes,
useful for modeling data structures like lists and trees,
can be combined with any other theory, using the polite combination method~\cite{ranise:inria-00000570}.
This combination method
offers a way to combine decisions procedures of two theories
into a decision procedure for the combined theory, with different assumptions than those
of the earlier Nelson-Oppen approach~\cite{NelsonOppen}.
In particular,
it was proven that the theory admits a technical property concerning cardinalities
of models, called
{\em strong politeness}~\cite{JB10-TR}.
It was noted in~\cite{DBLP:journals/jar/ShengZRLFB22} that proving strong politeness 
for this theory
seemed much harder than proving {\em politeness}, a similar but simpler property.
Therefore, the proof was split into three steps:
$(i)$~a class of theories was identified in which politeness and strong politeness coincide;
$(ii)$~the theory of algebraic datatypes was shown to be in this class; and
$(iii)$~this theory was proven to be polite.
This proof technique raised the following question:
    {{\bf does politeness imply strong politeness?}} 
An affirmative answer to this question 
would simplify strong politeness proofs that follow such steps, 
as only the last step would be needed.
Unfortunately, the answer to this question was shown in \cite{SZRRBT-21} to be negative, in its most general form.
However, an affirmative answer was given for theories over
one-sorted 
empty signatures, where politeness and strong politeness do coincide.

Seeing that relationships between model-theoretic properties of theories (like politeness and strong politeness)
are non-trivial, and can have a big impact on proofs in the field of theory combination,
we have recently initiated a more general research plan:
to systematically determine the relationships between
 model-theoretic properties that relate to theory combination.
An analysis of such properties can, for example, simplify proofs, in cases where
a property follows from a combination of other properties.

In the first stage of this plan~\cite{BarTolZoh},
we studied the relationships between 
all properties that relate to either polite or Nelson-Oppen combination, namely:
stable infiniteness, smoothness, finite witnessability, strong finite witnessability,
and convexity.
The first two properties relate to the ability to enlarge  cardinalities of models,
while the next two require a computable {\em witness} function that  restricts the models of a formula based on its variables.
The last property relies on the ability to deduce an equality from a disjunction of equalities.
The result of \cite{BarTolZoh} was a comprehensive table: nearly every combination of these properties (e.g., theories that are smooth and stably infinite but do not admit the other properties) was either proved to be infeasible, or an example for it was given.

In this paper we continue with this plan
by adding two properties:
the finite model property and stable finiteness, both
related to shiny theories~\cite{TinZar-RR-04}.
The former requires finite models for satisfiable formulas, and the latter enforces bounds on them.

Of course, the theories from \cite{BarTolZoh} can be reused.
For these, one only needs to determine if they admit the finite model property and/or stable finiteness.
The results and examples from \cite{BarTolZoh} are, however, not enough.
Given that the number of considered combinations is doubled with the addition of each property,
 new theories need to be introduced in order to exemplify the new possibilities,
 and new impossible combinations can be found.
Hence, in this paper we provide several impossibility results for the aforementioned properties,
as well as examples of theories for possible combinations.
The overall result is a new table which extends that of \cite{BarTolZoh} with 
two new columns corresponding to the finite model property 
and stable finiteness.\footnote{While we use several results from \cite{BarTolZoh}, we do not assume here any familiarity with that paper. All required results are mentioned here explicitly.}


%
The most interesting combinations that we study are theories
that admit the finite model property but not finite witnessability. %
While both properties deal with finite models, the latter has a computable element to it,
namely the witness function.
In separating these properties, we found it useful to define theories that are based on the {\em Busy Beaver} function,
a well known function from computability theory, that is not only non-computable,
but also grows eventually faster than any computable function.

{\bf Outline:} 
\Cref{Preliminary notions} reviews many-sorted logics and theory combination properties. \Cref{Finitemodelpropertyand stablefiniteness} identifies combinations that are contradictory; 
\Cref{sec:posresults} constructs the extended table of combinations, and describes the newly introduced theories.
\Cref{conclusion} gives final remarks and future directions this work can take.


\section{Preliminary Notions}\label{Preliminary notions}
\subsection{Many-sorted Logic}
A {\em many-sorted signature} $\Sigma$ is a triple $(\S_{\Sigma}, \F_{\Sigma}, \P_{\Sigma})$ where: $\S_{\Sigma}$ is a countable set of {\em sorts}; $\F_{\Sigma}$ is a countable set of function symbols; and $\P_{\Sigma}$ is a countable set of predicate symbols containing, for each $\sigma\in \S_{\Sigma}$, an equality $=_{\sigma}$.
When $\sigma$ is clear from the context, we write $=$.
Every function symbol has an {\em arity} of the form $\sigma_{1}\times\cdots\times\sigma_{n}\rightarrow \sigma$, and every predicate symbol one of the form $\sigma_{1}\times\cdots\times \sigma_{n}$, where $\sigma_{1}, \ldots, \sigma_{n}, \sigma\in \S_{\Sigma}$; 
equalities $=_{\s}$ have arity $\sigma\times\sigma$.

A signature that has no functions and only the equalities as predicates is called {\em empty}. 
Many-sorted signatures $\Sigma$ where $\S_{\Sigma}$ has only one element are called {\em one-sorted}. 

For any sort in $\S_{\Sigma}$ we assume a countably infinite set of variables, and distinct sorts have disjoint sets of variables; we then define first-order terms, formulas, 
and literals in the usual way. 
The set of free variables of sort $\s$ in a formula $\varphi$ is denoted by $\vars_{\s}(\varphi)$, while $\vars(\varphi)$ will denote $\bigcup_{\s\in\S_{\Sigma}}\vars_{\s}(\varphi)$.

$\Sigma$-Structures $\mathbb{A}$ are defined as usual, by interpreting sorts (denoted by $\s^{\mathbb{A}}$), functions ($f^{\mathbb{A}}$) and predicate symbols ($P^{\mathbb{A}}$),
with the restrictions that equality symbols are interpreted as identities. 
A $\Sigma$-interpretation $\A$ is an extension of a $\Sigma$-structure $\mathbb{A}$ with interpretations to variables.
If $\mathbb{A}$ is the underlying $\Sigma$-structure of a $\Sigma$-interpretation $\A$,
we say that $\A$ is an interpretation on $\mathbb{A}$. For simplicity, and because the use of structures is sparse in this paper, we will usually denote both structures and interpretations by using the same font, $\A$, $\B$ and so on.
$\alpha^{\A}$ is the value taken by a $\Sigma$-term $\alpha$ in a $\Sigma$-interpretation $\A$, and if $\Gamma$ is a set of terms, we simply write $\Gamma^{\A}$ for $\{\alpha^{\A} : \alpha\in \Gamma\}$. 

We write $\A\vDash\varphi$ if the $\Sigma$-interpretation $\A$ 
satisfies the $\Sigma$-formula $\varphi$; $\varphi$ is then said to be {\em satisfiable} if it is 
satisfied by some interpretation $\A$. 
The formulas found in \Cref{card-formulas} will be useful in the sequel. A $\Sigma$-interpretation $\A$: satisfies 
$\psi^{\sigma}_{\geq n}$ iff $|\sigma^{\A}|\geq n$; satisfies $\psi^{\sigma}_{\leq n}$ iff $|\sigma^{\A}|\leq n$; and satisfies $\psi^{\sigma}_{= n}$ iff $|\sigma^{\A}|=n$. For simplicity, when dealing with one-sorted signatures, we may drop the sort $\s$ from the cardinality formulas.

\begin{figure}[htp]
\begin{mdframed}
\[\psi^{\sigma}_{\geq n}=\Exists{\overarrow{x}} \bigwedge_{1\leq i<j\leq n}\neg(x_{i}=x_{j})
\quad\psi^{\sigma}_{\leq n}=\Exists{\overarrow{x}}\Forall{y} \bigvee_{i=1}^{n}y=x_{i}
\quad\psi^{\sigma}_{= n}=\psi^{\sigma}_{\geq n}\wedge \psi^{\sigma}_{\leq n}
\]
\end{mdframed}
\caption{Cardinality Formulas. $\overarrow{x}$ stands for $x_{1},\ldots,x_{n}$, all variables of sort $\s$.}
\label{card-formulas}
\end{figure}

A $\Sigma$-\emph{theory} $\T$ is a class of all $\Sigma$-interpretations (called $\T$-interpretations) that satisfy some set $\ax(\T)$ of closed formulas called the \emph{axiomatization} of $\T$; the structures underlying these interpretations will be called the \emph{models} of $\T$.

A formula is $\T$\emph{-satisfiable} if it is satisfied by some $\T$-interpretation and, analogously, a set of formulas is $\T$-satisfiable if there is a $\T$-interpretation that satisfies all of them simultaneously. Two formulas are $\T$\emph{-equivalent} when a $\T$-interpretation satisfies the first iff it satisfies the second. 
We write $\tmodels\varphi$, and say that $\varphi$ is \emph{$\T$-valid} if $\A\vDash\varphi$ for all $\T$-interpretations $\A$.

\subsection{Theory Combination Properties}
Let $\Sigma$ be a signature,
$\T$ a $\Sigma$-theory
and $\S\subseteq \S_{\Sigma}$.
We define several properties
$\T$ may have with respect to $S$.

\smallskip
\noindent
{\bf Convexity, Stable Infiniteness, and Smoothness }
$\T$ is \emph{convex} with respect to $S$ if 
for any conjunction of $\Sigma$-literals $\phi$ and any finite set of variables $\{u_{1}, v_{1},\ldots , u_{n}, v_{n}\}$ of sorts in $S$ with $\tmodels\phi\rightarrow \bigvee_{i=1}^{n}u_{i}=v_{i}$, one has $\tmodels\phi\rightarrow u_{i}=v_{i}$
for some $i$.
 $\T$ is \emph{stably infinite} with respect to $S$ if for every $\T$-satisfiable quantifier-free $\Sigma$-formula there is a $\T$-interpretation $\A$ satisfying it such that $|\sigma^{\A}|$ is infinite for each $\sigma\in S$.
 $\T$ is \emph{smooth} with respect to $S$ if for every quantifier-free formula, $\T$-interpretation $\A$ that satisfies it, and function $\kappa$ from $S$ to the class of cardinals such that 
$\kappa(\s)\geq|\s^{\A}|$ for each $\s\in S$, there is a $\T$-interpretation $\B$ that satisfies it with $|\sigma^{\B}|=\kappa(\sigma)$ for each $\sigma\in S$.

\smallskip
\noindent
{\bf (Strong) Finite witnessability}
For finite sets of variables $V_{\sigma}$ of sort $\sigma$ for each $\s\in S$, and equivalence relations $E_{\sigma}$ on $V_{\sigma}$, the arrangement on $V=\bigcup_{\s\in S}V_{\sigma}$ induced by $E=\bigcup_{\s\in S}E_{\sigma}$, denoted by $\delta_{V}$ or $\delta_{V}^{E}$, is the formula
$\delta_{V}=\bigwedge_{\sigma\in S}\big[\bigwedge_{xE_{\sigma}y}(x=y)\wedge\bigwedge_{x\overline{E_{\sigma}}y}\neg(x=y)\big]$,
where $\overline{E_{\sigma}}$ denotes the complement of the equivalence relation $E_{\sigma}$.

$\T$ is \emph{finitely witnessable} with respect to $S$ when there exists a computable function $\wit$, 
called a \emph{witness}, from the quantifier-free $\Sigma$-formulas to themselves that satisfies, 
for every $\phi$:
$(i)$~$\phi$ and $\exists\, \overarrow{w}.\:\wit(\phi)$ are $\T$-equivalent, for $\overarrow{w}=\vars(\wit(\phi))\setminus\vars(\phi)$; and
$(ii)$~if $\wit(\phi)$ is $\T$-satisfiable, there exists a $\T$-interpretation $\A$ satisfying $\wit(\phi)$ such that 
$\sigma^{\A}=\vars_{\sigma}(\wit(\phi))^{\A}$
for each $\sigma\in S$.

\emph{Strong finite witnessability} is defined similarly to finite witnessability,
replacing $(ii)$ by:
$(ii)'$~given a finite set of variables $V$ and an arrangement $\delta_{V}$ on $V$, if $\wit(\phi)\wedge\delta_{V}$ is $\T$-satisfiable, there exists a $\T$-interpretation $\A$ that satisfies  $\wit(\phi)\wedge\delta_{V}$ with
$\sigma^{\A}=\vars_{\sigma}(\wit(\phi)\wedge\delta_{V}\big)^{\A}$
for all $\sigma\in S$.
If $\T$ is smooth and (strongly) finitely witnessable with respect to $\S$,
then it is \emph{(strongly) polite} with respect to $\S$. 


\smallskip
\noindent
{\bf Finite Model Property and Stable Finiteness}
    $\T$ has the \emph{finite model property} with respect to $S$ if for every quantifier-free $\T$-satisfiable $\Sigma$-formula, there exists a $\T$-interpretation $\A$ that satisfies it with $|\s^{\A}|$ finite for each $\s\in S$.
$\T$ is \emph{stably finite} with respect to $S$ if, for every quantifier-free $\Sigma$-formula and $\T$-interpretation $\A$ that satisfies it, there exists a $\T$-interpretation $\B$ that satisfies it with: $|\s^{\B}|$ finite for each $\s\in S$; and $|\s^{\B}|\leq |\s^{\A}|$ for each $\s\in S$.
%
%
%
%
%
%
%
%
Clearly, stable finiteness implies the finite model property:
\begin{theorem}
\label{SMimpliesSI}
\label{SFWimpliesFW}
\label{FWimpliesSfinite}
\label{SF=>FMP}
    If  $\T$ is stably finite w.r.t. $S$, then it  has the finite model property w.r.t. $S$.
    \end{theorem}

We shall write
$\stainf$ for stably infinite;
$\smooth$ for smooth;
$\finwit$ ($\strfinwit$) for (strong) finitely witnessable;
$\convex$ for convex;
$\finmodpro$ for the finite model property;
and $\stafin$ for stably finite.

\section{Relationships between model-theoretic properties}\label{Finitemodelpropertyand stablefiniteness}
In this section we study the connections between  finiteness properties related to theory combination:
the finite model property, stable finiteness, finite witnessability, and strong finite witnessability.
We show how these properties are related to one another.
In \Cref{sec:gensig}, we provide general results that hold for all signatures.
Then, in
\Cref{sec:emptysig}, we focus on empty signatures, in which we are able to find more connections.

\subsection{General Signatures}
Finite witnessability, as well as its strong variant, were introduced in the context of polite theory combination.
In contrast, the study of shiny theories utilizes the notions of the finite model property, as well as stable finiteness.
It was shown in \cite{Casal2018} that for theories with a decidable quantifier-free satisfiability problem,
shiny theories and strongly polite theories are one and the same. This already showed some connections between the aforementioned
finiteness properties. 
However, that analysis also relied on smoothness, the decidability of the quantifier-free satisfiability problem of the studied theories,
as well as the computability of the $\mc$ function, the function that computes the minimal sizes of domains in models of a given formula in these theories.

Here we focus purely on the finiteness properties, and show that even without any other assumptions, they are closely related.
Considering finite witnessability and the finite model property, notice that any witness ensures that some formulas
always have finite models.
Using the equivalence of the existential closure of such formulas to the formulas that are given to the witness,
one gets the following result, according to which finite witnessability implies the finite model property.

\label{sec:gensig}
\begin{restatable}{theorem}{FWimpliesFMP}\label{FW=>FMP}
    Any $\Sigma$-theory $\T$ finitely witnessable with respect to $S\subseteq \S_{\Sigma}$ also has the finite model property with respect to $S$.
\end{restatable}

Strong finite witnessability is a stronger property than finite witnessability, obtained by requiring finite models in the presence of arrangements.
This requirement allows one to conclude stable finiteness for it, as the finer control on cardinalities that is required
for stable finiteness can be achieved with the aid of arrangements. 
The following result is proved in Lemma 3.6 of \cite{Casal2018}, although under the assumption that the theory is smooth, something that is not actually used in their proof.

\begin{restatable}{theorem}{SFWimpliesSF}\label{SFW=>SF}
    Any $\Sigma$-theory $\T$ strongly finitely witnessable with respect to $S\subseteq \S_{\Sigma}$ is also stably finite with respect to $S$.
\end{restatable}

Clearly, stable finiteness implies the finite model property (\Cref{SF=>FMP}).
The converse does not generally hold, as we will see in \Cref{sec:positive}.
However, when these properties are considered with respect to a single sort,
they actually coincide:

\begin{restatable}{theorem}{OSplusFMPimplySF}\label{OS+FMP=>SF}
    If a $\Sigma$-theory $\T$ has the finite model property with respect to a set of sorts $S$ with $|S|=1$, then $\T$ is also stably finite with respect to $S$.
\end{restatable}

\Cref{FW=>FMP,SFW=>SF} are visualized in the Venn diagram of \Cref{venn-all},
where, for example, theories that are strongly finitely witnessable are clearly inside the intersection of finitely witnessable theories
and stably finite theories. 

When only one sort is considered, the picture is much simpler, and is described in \Cref{venn-one}.
There, the finite model property and stable finiteness populate the same region, as ensured by \Cref{OS+FMP=>SF}.
Notice that the results depicted in \Cref{venn-one} hold for one-sorted and many-sorted signatures. The key thing is that the properties are all w.r.t. one of the sorts.

\begin{figure}[t]\label{venndiagrams}
\centering
\begin{minipage}{0.49\textwidth}
\centering
    \begin{tikzpicture}[scale=.7]
\def\firstcircle{(0,0) coordinate (a) circle (2.5cm)}
\def\secondcircle{(0,0) coordinate (b)  circle (0.7cm)}
\def\thirdcircle{(-0.5,0) coordinate (c) circle (1.5cm)}
\def\fourthcircle{(0.5,0) coordinate (d)  circle (1.5cm)}
\draw \firstcircle;
\draw \secondcircle;
\draw \thirdcircle;
\draw \fourthcircle;
\node[label={$\strfinwit$}] (B) at (0,-0.3) {};
\node[label={$\finwit$}] (B) at (1.5,-0.3) {};
\node[label={$\stafin$}] (B) at (-1.5,-0.3) {};
\node[label={$\finmodpro$}] (B) at (0,1.5) {};
    \end{tikzpicture}
    \caption{Finiteness properties:  general case.}\label{venn-all}
    \end{minipage}   
    ~
    \begin{minipage}{0.49\textwidth}
\centering
    \begin{tikzpicture}[scale=.7]
\def\firstcircle{(0,0) coordinate (a) circle (2.5cm)}
\def\secondcircle{(0,0) coordinate (b)  circle (0.7cm)}
\def\fourthcircle{(0,0) coordinate (d)  circle (1.5cm)}
\draw \firstcircle;
\draw \secondcircle;
\draw \fourthcircle;
\node[label={$\strfinwit$}] (B) at (0,-0.3) {};
\node[label={$\finwit$}] (B) at (0,0.6) {};
\node[label={$\finmodpro$ \& $\stafin$}] (B) at (0,1.5) {};
    \end{tikzpicture}
    \caption{Finiteness properties w.r.t. one sort.}\label{venn-one}
    \end{minipage}

\end{figure}

\subsection{Empty Signatures}
\label{sec:emptysig}

\Cref{venn-all,venn-one} show a complete picture of the relationships between the properties studied in this section, for  arbitrary signatures.
However, when this generality is relaxed, several other connections appear.
For this section, we require that the signatures are empty, and that they have a finite set of sorts.
We further require that the properties in question hold for the entire set of sorts, not for any subset of it.

\Cref{tab:signatures} defines the 5 signatures that will be used in the examples found in \Cref{sec:posresults}, and that will also appear in some of the results shown below:
the empty signatures $\Sigmaone$, $\Sigmatwo$ and $\Sigma_{3}$, with sets of sorts $\{\s\}$, $\{\s, \s_{2}\}$ and $\{\s, \s_{2}, \s_{3}\}$, respectively; and the signatures $\Sigmas$ and $\Sigmastwo$ with one function $s$ of arity $\s\rightarrow\s$, and sets of sorts $\{\s\}$ and $\{\s, \s_{2}\}$, respectively. Notice these are the simplest possible signatures when we order those by establishing: first, that the signature with fewer sorts is simpler; and second, that if two signatures have the same number of sorts, the one with fewer function symbols is simpler. We are free not to consider predicates, as they are at least as expressive as functions themselves; furthermore, we do not consider the problem of defining which of two signatures with the same numbers of sorts and function symbols is simpler, choosing rather to add only functions from a sort to itself.

\begin{table}[htp]
    \centering
    \renewcommand{\arraystretch}{1.4}
    \begin{tabular}{c|c|c}
    Signature & Sorts & Function Symbols \\\hline
         $\Sigma_1$ & $\set{\sigma}$ & $\emptyset$ \\
         $\Sigma_2$ & $\set{\sigma,\sigma_2}$ & $\emptyset$ \\
         $\Sigma_3$ & $\set{\sigma,\sigma_2, \s_{3}}$ & $\emptyset$ \\
         $\Sigma_s$ & $\set{\sigma}$ & $\set{s: \sigma\rightarrow\sigma}$ \\
         $\Sigma_s^2$ & $\set{\sigma,\sigma_2}$ & $\set{s: \sigma\rightarrow\sigma}$ 
    \end{tabular}
    \renewcommand{\arraystretch}{1}
    \vspace{2em}
    \caption{Signatures that will be used throughout the paper.}
    \label{tab:signatures}
\end{table}

First, in such a setting, we have that the finite model property implies finite witnessability, in the presence of smoothness.

\begin{restatable}{theorem}{OSESSMFMPimplyFW}\label{OS+ES+SM+FMP=>FW}
    If $\Sigma$ is an empty signature with a finite set of sorts $\S_{\Sigma}$, and the $\Sigma$-theory $\T$ has the finite model property and is smooth with respect to $\S_{\Sigma}$, then $\T$ is also finitely witnessable with respect to $\S_{\Sigma}$.
\end{restatable}

Next, we show that stable finiteness and smoothness together, imply strong finite witnessability.

\begin{restatable}{theorem}{OSESSMSFimpliesSFW}\label{OS+ES+SM+SF=>SFW}
    If $\Sigma$ is an empty signature with a finite set of sorts $\S_{\Sigma}$, 
    and the $\Sigma$-theory $\T$ is stably finite and smooth with respect to $\S_{\Sigma}$, then $\T$ is also strongly finitely witnessable with respect to $\S_{\Sigma}$.
\end{restatable}

\begin{figure}[t]
    \centering
    \begin{tikzpicture}[scale=1]
\def\firstrectangle{(-3.8,-1.4) rectangle (0.5, 0.5)}
\def\secondrectangle{(-2.5,-0.5) rectangle (0.5, 0.5)}
\def\thirdrectangle{(-0.4,0) rectangle (1.7, 1)}
\draw \firstrectangle;
\draw \secondrectangle;
\draw \thirdrectangle;
\node[label={$\smooth$}] (B) at (1.1,0) {};
\node[label={$\finwit$ ($\strfinwit$)}] (B) at (-1.4,-0.7) {};
\node[label={$\finmodpro$ ($\stafin$)}] (B) at (-2.3,-1.5) {};
    \end{tikzpicture}
    \caption{Interplay between $\smooth$, $\finwit$ ($\strfinwit$) and $\finmodpro$ ($\stafin$) w.r.t. $\S_{\Sigma}$ in an empty signature.}
    \label{diagramforsmoothplusfiniteproperty}
\end{figure}

While \Cref{FW=>FMP} and \Cref{SFW=>SF} establish certain unconditional relations between finite witnessability and the finite model property, and strong finite witnessability and stable finiteness, the converses shown to hold in \Cref{OS+ES+SM+FMP=>FW} and \Cref{OS+ES+SM+SF=>SFW} demand smoothness and that the properties hold with respect to the entire set of sorts. In that case, the situation can be represented by the diagram found in \Cref{diagramforsmoothplusfiniteproperty}, showing clearly that a smooth theory that also has the finite model property (respectively, is stably finite), cannot not be finitely witnessable (strongly finitely witnessable).

Lastly, regarding the empty signatures $\Sigmaone$, $\Sigmatwo$ and $\Sigma_{3}$, the following theorem shows that
$\Sigma_{3}$ is sometimes necessary.

\begin{restatable}{theorem}{threesorted}\label{threesorted}
There are no $\Sigmaone$ or $\Sigmatwo$-theories $\T$ that are, simultaneously, neither stably infinite nor stably finite, but are convex and have the finite model property, with respect to the entire set of their sorts.
\end{restatable}

Hence, to exhibit such theories, one has to consider three-sorted theories.

\section{A taxonomy of examples}
\label{sec:posresults}

\begin{table}[htp]
\renewcommand{\arraystretch}{1.2}
\centering
\begin{tabular}{|P{0.5cm}|P{0.5cm}|P{0.5cm}|P{0.5cm}|P{0.5cm}|P{0.5cm}|P{0.5cm}||P{1.8cm}P{1.8cm}P{1.8cm}P{1.8cm}|P{0.5cm}|}
\hline
\multicolumn{7}{|c|}{} & \multicolumn{2}{c|}{Empty} & \multicolumn{2}{c|}{Non-empty} & \\
\hline
$\stainf$ & $\smooth$ & $\finwit$ & $\strfinwit$ & $\convex$ & $\finmodpro$ & $\stafin$ & \multicolumn{1}{c|}{One-sorted} & \multicolumn{1}{c|}{Many-sorted} &\multicolumn{1}{c|}{One-sorted} & \multicolumn{1}{c|}{Many-sorted}& $N^{\underline{o}}$\\\hline

\multirow{24}{*}{$T$ } & \multirow{12}{*}{$T$ }&\multirow{6}{*}{$T$ }&\multirow{2}{*}{$T$ }&\multirow{1}{*}{$T$}& \multirow{1}{*}{$T$} & $T$ &\multicolumn{1}{c|}{$\Tgeqn$}&\multicolumn{1}{c|}{$\adds{\Tgeqn}$}& \multicolumn{1}{c|}{$\addf{\Tgeqn}$}&\multicolumn{1}{c|}{$\addf{\adds{\Tgeqn}}$}& 1\\\hhline{~~~~--------}
&&&&\multirow{1}{*}{$F$ }& \multirow{1}{*}{$T$} & $T$ &\multicolumn{2}{c|}{\cite{BarTolZoh}\cellcolor{red!15}}& \multicolumn{1}{c|}{$\addnc{\Tgeqn}$} &\multicolumn{1}{c|}{$\addnc{\adds{\Tgeqn}}$}& 2\\\hhline{~~~---------}
&&&\multirow{4}{*}{$F$ }&\multirow{2}{*}{$T$}&\multirow{2}{*}{$T$}& $T$&\multicolumn{2}{c|}{\cellcolor{red!15}\Cref{OS+ES+SM+SF=>SFW}}&\multicolumn{1}{c|}{$\TsM$}&\multicolumn{1}{c|}{$\addf{\TsM}$}& 3\\\hhline{~~~~~~------}
&&&&&&$F$&\multicolumn{1}{c|}{\cellcolor{red!15}\Cref{OS+FMP=>SF}}&\multicolumn{1}{c|}{$\Ttwo$}&\multicolumn{1}{c|}{\cellcolor{red!15}\Cref{OS+FMP=>SF}}&\multicolumn{1}{c|}{$\addf{\Ttwo}$}&4\\\hhline{~~~~--------}

&&&&\multirow{2}{*}{$F$ }&\multirow{2}{*}{$T$}&$T$&\multicolumn{2}{c|}{\cellcolor{red!15}}& \multicolumn{1}{c|}{$\TSM$}&\multicolumn{1}{c|}{$\adds{\TSM}$}& 5\\\hhline{~~~~~~-*{1}{>{\arrayrulecolor{red!15}}|--}*{1}{>{\arrayrulecolor{black}}|-}--}
&&&&&&$F$&\multicolumn{2}{c|}{\multirow{-2}{*}{\cite{BarTolZoh}\cellcolor{red!15}}}&\multicolumn{1}{c|}{\cellcolor{red!15}\Cref{OS+FMP=>SF}}&\multicolumn{1}{c|}{$\addnc{\Ttwo}$}&6\\\hhline{~~----------}

&&\multirow{6}{*}{$F$ }&\multirow{6}{*}{$F$ }&\multirow{3}{*}{$T$ }& \multirow{2}{*}{$T$}&$T$&\multicolumn{2}{c|}{\cellcolor{red!15}}&\multicolumn{1}{c|}{$\TsBB$}&\multicolumn{1}{c|}{$\TtwoBB$}&7\\\hhline{~~~~~~-*{1}{>{\arrayrulecolor{red!15}}|--}*{1}{>{\arrayrulecolor{black}}|-}--}
&&&&&&$F$&\multicolumn{2}{c|}{\multirow{-2}{*}{\Cref{OS+ES+SM+FMP=>FW}\cellcolor{red!15}}}&\multicolumn{1}{c|}{\cellcolor{red!15}\Cref{OS+FMP=>SF}}&\multicolumn{1}{c|}{$\TBBtwotwo$}&8\\\hhline{~~~~~-------}
&&&&&$F$&$F$&\multicolumn{1}{c|}{$\Tinfty$}&\multicolumn{1}{c|}{$\adds{\Tinfty}$}&\multicolumn{1}{c|}{$\addf{\Tinfty}$}&\multicolumn{1}{c|}{$\addf{\adds{\Tinfty}}$}&9\\\hhline{~~~~--------}

&&&&\multirow{3}{*}{$F$}&\multirow{2}{*}{$T$}&$T$&\multicolumn{2}{c|}{\cellcolor{red!15}}&\multicolumn{1}{c|}{$\TssBB$}&\multicolumn{1}{c|}{$\adds{\TssBB}$}& 10\\\hhline{~~~~~~-*{1}{>{\arrayrulecolor{red!15}}|--}*{1}{>{\arrayrulecolor{black}}|-}--}
&&&&&&$F$&\multicolumn{2}{c|}{\cellcolor{red!15}}&\multicolumn{1}{c|}{\cellcolor{red!15}\Cref{OS+FMP=>SF}}&\multicolumn{1}{c|}{$\TstwoBB$}&11\\\hhline{~~~~~--*{1}{>{\arrayrulecolor{red!15}}|--}*{1}{>{\arrayrulecolor{black}}|-}--}
&&&&&$F$&$F$&\multicolumn{2}{c|}{\multirow{-3}{*}{\cite{BarTolZoh}\cellcolor{red!15}}}&\multicolumn{1}{c|}{$\addnc{\Tinfty}$}&\multicolumn{1}{c|}{$\addnc{\adds{\Tinfty}}$}&12\\\hhline{~-----------}

&\multirow{12}{*}{$F$ }&\multirow{6}{*}{$T$ }&\multirow{2}{*}{$T$ }&\multirow{1}{*}{$T$ }&\multirow{1}{*}{$T$ }&$T$&&&&& 13\\\hhline{~~~~---~~~~-}

&&&&\multirow{1}{*}{$F$ }&\multirow{1}{*}{$T$}&$T$&\multicolumn{4}{c|}{\multirow{-2}{*}{\textcolor{red}{Unicorn}}}&14\\\hhline{~~~---------}

&&&\multirow{4}{*}{$F$ }&\multirow{2}{*}{$T$ }&\multirow{2}{*}{$T$}&$T$&\multicolumn{1}{c|}{$\Teven$}&\multicolumn{1}{c|}{$\adds{\Teven}$}&\multicolumn{1}{c|}{$\addf{\Teven}$}&\multicolumn{1}{c|}{$\addf{\adds{\Teven}}$}& 15\\\hhline{~~~~~~------}
&&&&&&$F$&\multicolumn{1}{c|}{\cellcolor{red!15}\Cref{OS+FMP=>SF}}&\multicolumn{1}{c|}{$\Tupinfty$}&\multicolumn{1}{c|}{\cellcolor{red!15}\Cref{OS+FMP=>SF}}&\multicolumn{1}{c|}{$\addf{\Tupinfty}$}&16\\\hhline{~~~~--------}

&&&&\multirow{2}{*}{$F$ }&\multirow{2}{*}{$T$}&$T$&\multicolumn{2}{c|}{\cellcolor{red!15}}&\multicolumn{1}{c|}{$\addnc{\Teven}$}&\multicolumn{1}{c|}{$\addnc{\adds{\Teven}}$}& 17\\\hhline{~~~~~~-*{1}{>{\arrayrulecolor{red!15}}|--}*{1}{>{\arrayrulecolor{black}}|-}--}
&&&&&&$F$&\multicolumn{2}{c|}{\multirow{-2}{*}{\cite{BarTolZoh}\cellcolor{red!15}}}&\multicolumn{1}{c|}{\cellcolor{red!15}\Cref{OS+FMP=>SF}}&\multicolumn{1}{c|}{$\addnc{\Tupinfty}$}&18\\\hhline{~~----------}

&&\multirow{6}{*}{$F$ }&\multirow{6}{*}{$F$ }&\multirow{3}{*}{$T$ }&\multirow{2}{*}{$T$}&$T$&\multicolumn{1}{c|}{$\TBB$}&\multicolumn{1}{c|}{$\adds{\TBB}$}&\multicolumn{1}{c|}{$\addf{\TBB}$}&\multicolumn{1}{c|}{$\addf{\adds{\TBB}}$}&19\\\hhline{~~~~~~------}
&&&&&&$F$&\multicolumn{1}{c|}{\cellcolor{red!15}\Cref{OS+FMP=>SF}}&\multicolumn{1}{c|}{$\TBBtwo$}&\multicolumn{1}{c|}{\cellcolor{red!15}\Cref{OS+FMP=>SF}}&\multicolumn{1}{c|}{$\addf{\TBBtwo}$}&20\\\hhline{~~~~~-------}
&&&&&$F$&$F$&\multicolumn{1}{c|}{$\Tninfty$}&\multicolumn{1}{c|}{$\adds{\Tninfty}$}&\multicolumn{1}{c|}{$\addf{\Tninfty}$}&\multicolumn{1}{c|}{$\addf{\adds{\Tninfty}}$}&21\\\hhline{~~~~--------}

&&&&\multirow{3}{*}{$F$ }&\multirow{2}{*}{$T$}&$T$&\multicolumn{2}{c|}{\cellcolor{red!15}}&\multicolumn{1}{c|}{$\addnc{\TBB}$}&\multicolumn{1}{c|}{$\addnc{\adds{\TBB}}$}& 22\\\hhline{~~~~~~-*{1}{>{\arrayrulecolor{red!15}}|--}*{1}{>{\arrayrulecolor{black}}|-}--} 
&&&&&&$F$&\multicolumn{2}{c|}{\cellcolor{red!15}}&\multicolumn{1}{c|}{\cellcolor{red!15}\Cref{OS+FMP=>SF}}&\multicolumn{1}{c|}{$\addnc{\TBBtwo}$}&23\\\hhline{~~~~~--*{1}{>{\arrayrulecolor{red!15}}|--}*{1}{>{\arrayrulecolor{black}}|-}--}
&&&&&$F$&$F$&\multicolumn{2}{c|}{\multirow{-3}{*}{\cite{BarTolZoh}\cellcolor{red!15}}}&\multicolumn{1}{c|}{$\addnc{\Tninfty}$}&\multicolumn{1}{c|}{$\addnc{\adds{\Tninfty}}$}&24\\\hline

\multirow{12}{*}{$F$ }&\multirow{12}{*}{$F$ }&\multirow{6}{*}{$T$ }&\multirow{2}{*}{$T$ }&\multirow{1}{*}{$T$ }&\multirow{1}{*}{$T$ }&$T$&\multicolumn{1}{c|}{$\Tone$}&\multicolumn{1}{c|}{$\adds{\Tone}$}&\multicolumn{1}{c|}{$\addf{\Tone}$}&\multicolumn{1}{c|}{$\addf{\adds{\Tone}}$}& 25\\\hhline{~~~~--------}

&&&&\multirow{1}{*}{$F$ }&\multirow{1}{*}{$T$ }&$T$&\multicolumn{1}{c|}{$\Tleqn$}&\multicolumn{1}{c|}{$\adds{\Tleqn}$}&\multicolumn{1}{c|}{$\addf{\Tleqn}$}&\multicolumn{1}{c|}{$\addf{\adds{\Tleqn}}$}& 26\\\hhline{~~~---------}

&&&\multirow{4}{*}{$F$ }&\multirow{2}{*}{$T$ }&\multirow{2}{*}{$T$ }&$T$&\multicolumn{1}{c|}{\cellcolor{red!15}\cite{BarTolZoh}}&\multicolumn{1}{c|}{$\Toddtwo$}&\multicolumn{1}{c|}{$\Tneqodd$}& \multicolumn{1}{c|}{$\addf{\Toddtwo}$}&27\\\hhline{~~~~~~------}

&&&&&&$F$&\multicolumn{1}{c|}{\cellcolor{red!15}\Cref{OS+FMP=>SF}}&\multicolumn{1}{c|}{$\Ttwocube$}&\multicolumn{1}{c|}{\cellcolor{red!15}\Cref{OS+FMP=>SF}}&\multicolumn{1}{c|}{$\Tsupinfty$}&28\\\hhline{~~~~--------}

&&&&\multirow{2}{*}{$F$ }&\multirow{2}{*}{$T$ }&$T$& \multicolumn{1}{c|}{$\Tmn$}&\multicolumn{1}{c|}{$\adds{\Tmn}$}&\multicolumn{1}{c|}{$\addf{\Tmn}$}&\multicolumn{1}{c|}{$\addf{\adds{\Tmn}}$}& 29\\\hhline{~~~~~~------}

&&&&&&$F$&\multicolumn{1}{c|}{\cellcolor{red!15}\Cref{OS+FMP=>SF}}&\multicolumn{1}{c|}{$\Tmninfty$}&\multicolumn{1}{c|}{\cellcolor{red!15}\Cref{OS+FMP=>SF}}&\multicolumn{1}{c|}{$\addf{\Tmninfty}$}&30\\\hhline{~~----------}

&&\multirow{6}{*}{$F$ }&\multirow{6}{*}{$F$ }&\multirow{3}{*}{$T$ }&\multirow{2}{*}{$T$ }&$T$&\multicolumn{1}{c|}{\cellcolor{red!15}}&\multicolumn{1}{c|}{$\TBBone$}&\multicolumn{1}{c|}{$\TneqBBone$}&\multicolumn{1}{c|}{$\adds{\TBBone}$}& 31\\\hhline{~~~~~~-*{1}{>{\arrayrulecolor{red!15}}|-}*{1}{>{\arrayrulecolor{black}}|-}---}

&&&&&&$F$&\multicolumn{1}{c|}{\cellcolor{red!15}}&\multicolumn{1}{c|}{$\TBBtwocube$}&\multicolumn{1}{c|}{\cellcolor{red!15}\Cref{OS+FMP=>SF}}&\multicolumn{1}{c|}{$\TsBBtwo$}&32\\\hhline{~~~~~--*{1}{>{\arrayrulecolor{red!15}}|-}*{1}{>{\arrayrulecolor{black}}|-}---}
&&&&&$F$&$F$&\multicolumn{1}{c|}{\cellcolor{red!15}}&\multicolumn{1}{c|}{$\Tonetwo$}&\multicolumn{1}{c|}{$\Tneqoneinfty$}&\multicolumn{1}{c|}{$\addf{\Tonetwo}$}&33\\\hhline{~~~~---*{1}{>{\arrayrulecolor{red!15}}|-}*{1}{>{\arrayrulecolor{black}}|-}---}

&&&&\multirow{3}{*}{$F$ }&\multirow{2}{*}{$T$ }&$T$&\multicolumn{1}{c|}{\multirow{1}{*}{\cellcolor{red!15}}}&\multicolumn{1}{c|}{$\TBBn$} &\multicolumn{1}{c|}{$\TneqBB$} &\multicolumn{1}{c|}{$\addf{\TBBn}$}& 34\\\hhline{~~~~~~-*{1}{>{\arrayrulecolor{red!15}}|-}*{1}{>{\arrayrulecolor{black}}|-}---}

&&&&&&$F$&\multicolumn{1}{c|}{\multirow{1}{*}{\cellcolor{red!15}}}&\multicolumn{1}{c|}{$\TmnBB$}&\multicolumn{1}{c|}{\multirow{1}{*}{\cellcolor{red!15}\Cref{OS+FMP=>SF}}}&\multicolumn{1}{c|}{$\addf{\TmnBB}$}&35\\\hhline{~~~~~--*{1}{>{\arrayrulecolor{red!15}}|-}*{1}{>{\arrayrulecolor{black}}|-}---}
&&&&&$F$&$F$&\multicolumn{1}{c|}{\multirow{-6}{*}{\cellcolor{red!15}\cite{BarTolZoh}}}&\multicolumn{1}{c|}{$\Ttwotwo$}&\multicolumn{1}{c|}{$\Tneqtwoinfty$}&\multicolumn{1}{c|}{$\addf{\Ttwotwo}$}&36\\\hline
\end{tabular}
\renewcommand{\arraystretch}{1}
\vspace{2em}
\caption[Caption for LOF]{Summary of all possible combinations of theory properties. 
Red cells represent impossible combinations. In lines $26$ and $34$, $n>1$; in lines $29$, $30$ and $35$, $m>1$, $n>1$ and $|m-n|>1$.}
\label{tab-summary}
\end{table}

\label{sec:positive}
In \cite{BarTolZoh}, we have created a table, in which for every possible combinations of properties from 
$\{$ $\stainf$, $\smooth$, $\finwit$, $\strfinwit$, $\convex$ $\}$ we either gave 
an example of a theory in this combination, or proved a theorem
that shows there is no such example,
with the exception of theories that are stably infinite and strongly finitely witnessable but not smooth.
Such theories, referred to in \cite{BarTolZoh} as {\em Unicorn Theories}
(due to our conjecture that they do not exist) were left for future work,
and are still left for future work, as the focus of the current paper
is the integration of finiteness properties,
namely $\finmodpro$ and $\stafin$ to the table.

And indeed, the goal of this section is to add two columns to the table 
from \cite{BarTolZoh}:
one for the finite model property and one for stable finiteness.
The extended table is \Cref{tab-summary}.
We do not assume familiarity with \cite{BarTolZoh}, and describe the entire resulting table
(though focusing on the new results).

This section is structured as follows:
In  \Cref{sec:thetable} we describe the structure of \Cref{tab-summary}.
In \Cref{sec:oldtheories,sec:simpletheories,sec:theoriesBB} we provide details about the axiomatizations of theories that populate it.
Finally, in \Cref{sec:theoryops}, we reuse operators from \cite{BarTolZoh}, prove that they
preserve the finite model property and stable finiteness, and show how they are used in order to generate more theories for \Cref{tab-summary}.

\subsection{The Table}
\label{sec:thetable}
The columns left to the vertical double-line of \Cref{tab-summary} correspond to possible combinations of properties.
In them, $T$ means that the property holds, while $F$ means that it does not.
The first $5$ columns correspond to properties already studied in \cite{BarTolZoh}, and the next two columns correspond to $\finmodpro$ and $\stafin$.
The columns right to the vertical double-line correspond to possible signatures:
empty or non-empty, and one-sorted or many-sorted.
White cells correspond to cases where a theory with the 
combination of properties induced by the row exists in a signature that is induced by the column.
In such a case, the name of the theory is written.
The theories themselves are defined in \Cref{tab-theories-tab-old,tab-theories-sigma-22,tab-theories-sigma-s}, axiomatically.
Red cells correspond to the cases where there is no such theory.
In such a case, the theorem that excludes this possibility is written.
If that theorem is from \cite{BarTolZoh}, we simply write \cite{BarTolZoh}.

\begin{example}
    Line $1$ of \Cref{tab-summary} corresponds to theories
    that admit all studied properties.
We see that there is such a theory in each of the studied types of signatures
(e.g., for the empty one-sorted signature,
the theory $\Tgeqn$ exhibits all properties).
In contrast, line $3$ corresponds to theories
that admit all properties but strong finite witnessability.
We see that such theories exist in non-empty signatures, but not in empty signatures.
This is thanks to \Cref{OS+ES+SM+SF=>SFW}.
\end{example}

\Cref{Finitemodelpropertyand stablefiniteness}, as well as results from \cite{BarTolZoh},
make some potential rows of \Cref{tab-summary} completely red.
To allow this table to fit a single page, we chose to erase such rows.
For example, by \Cref{SMimpliesSI}, there are no theories that are 
stably finite but do not have the finite model property, in any signature.
Thus, no rows that represent such theories appear in the table.

In the remainder of this section, we describe the various theories that populate the cells of the table.
Fortunately, all theories from \cite{BarTolZoh} can be reused
to exhibit also the new properties $\stafin$ and $\finmodpro$, or their negations.
These are described in \Cref{sec:oldtheories}.
However, the theories from \cite{BarTolZoh} alone are not enough.
Hence we introduce several new theories in
\Cref{sec:simpletheories,sec:theoriesBB}.
Some of them are relatively simple, and are described in \Cref{sec:simpletheories}.
Most of them, however, are more complex, and rely on the Busy Beaver function
from theoretical computer science. We discuss these theories in \Cref{sec:theoriesBB}.

\subsection{Theories from \cite{BarTolZoh}}
\label{sec:oldtheories}

\begin{figure}[t]
\centering
\begin{minipage}[t]{0.39\textwidth}
\renewcommand{\arraystretch}{1.4}
\centering
\begin{tabular}{c|c|c}
Name & Sig. & Axiomatization\\\hline
$\Tgeqn$ & $\Sigmaone$ & $\{\psi_{\geq n}\}$\\
$\Teven$ & $\Sigmaone$ & $\{\neg\psi_{=2k+1} :  k\in\mathbb{N}\}$\\
$\Tinfty$ & $\Sigmaone$ & $\{\psi_{\geq k} :  k\in\mathbb{N}\}$\\
$\Tninfty$ & $\Sigmaone$ & $\{\psi_{=n}\vee\psi_{\geq k} : k\in\mathbb{N}\}$\\
$\Tleqn$ &  $\Sigmaone$ & $\{\psi_{\leq n}\}$\\
$\Tmn$ & $\Sigmaone$ & $\{\psi_{=m}\vee\psi_{=n}\}$\\
\end{tabular} 
\end{minipage}
\begin{minipage}[t]{0.59\textwidth}
\renewcommand{\arraystretch}{1.4}
\centering
\begin{tabular}{c|c|c}
Name & Sig. & Axiomatization\\\hline

$\Ttwo$ & $\Sigmatwo$ & $\{(\psi^{\s}_{=2}\wedge\psi^{\s_{2}}_{\geq k})\vee(\psi^{\s}_{\geq 3}\wedge \psi^{\s_{2}}_{\geq 3}) : k\in\mathbb{N}\}$\\
$\Ttwotwo$ & $\Sigmatwo$ & $\{\psi^{\s}_{=2}\}\cup\{\psi^{\s_{2}}_{\geq k} : k\in\mathbb{N}\}$\\
$\Toddtwo$  & $\Sigmatwo$ & $\{\psi^{\s}_{=1}\}\cup\{\neg\psi^{\s_{2}}_{=2k} : k\in\mathbb{N}\}$\\
$\Tonetwo$  & $\Sigmatwo$ & $\{\psi^{\s}_{=1}\}\cup\{\psi^{\s_{2}}_{\geq k} : k\in\mathbb{N}\}$\\
\end{tabular}
\end{minipage}

\vspace{5mm}

\renewcommand{\arraystretch}{1.4}
\centering
\begin{tabular}{c|c|c}
Name & Sig. & Axiomatization\\\hline
$\TsM$ & $\Sigmas$ & $\{[\psi^{=}_{\geq f_{1}(k)}\wedge \psi^{\neq}_{\geq f_{0}(k)}]\vee\bigvee_{i=1}^{k}[\psi^{=}_{=f_{1}(i)}\wedge \psi^{\neq}_{=f_{0}(i)}]: k\in\mathbb{N}\setminus\{0\}\}$\\
$\TSM$ & $\Sigmas$ &  $\ax(\TsM)\cup\{\psiv\}$\\
$\Tneqtwoinfty$ & $\Sigmas$ & $\{[\psi_{=2}\wedge\Forall {x}p(x)]\vee [\psi_{\geq k}\wedge \Forall{x}\neg p(x)] : k\in\mathbb{N}\}$\\
$\Tneqodd$ & $\Sigmas$ & $\{\psi_{=1}\vee[\neg\psi_{=2k}\wedge\Forall{x}\neg p(x)] : k\in\mathbb{N}\}$\\
$\Tneqoneinfty$  & $\Sigmas$ & $\{ \psi_{=1}\vee[\psi_{\geq k}\wedge\Forall{x}\neg p(x)] : k\in\mathbb{N}\}$\\

\end{tabular}
\renewcommand{\arraystretch}{1}
\vspace{2em}
  \captionof{figure}{Theories for \Cref{tab-summary} that were studied in \cite{BarTolZoh}; $p(x)$ stands for $s(x)=x$.
In $\TsM$, $f$ is any non-computable function from the positive integers to $\{0,1\}$, such that for every $k\geq 0$,
$f$ maps half of the numbers between $0$ and $2^k$ to $1$, and the other half to $0$.
In \cite{BarTolZoh}, such a function was proven
to exist. 
}
  \label{tab-theories-tab-old}
  \label{tab-theories-venn-old}
\end{figure}

\newcommand{\distinct}[1]{\delta_{{#1}}}
\begin{figure}[t]
\begin{mdframed}
\[\small \psi^{\eq}_{\geq n}=\exists\, \overarrow{x}.\:\bigwedge_{i=1}^{n} p(x_i)\wedge \distinct{n} ~~~~~~~\small \psi^{\eq}_{=n}=\exists\, \overarrow{x}.\:[\bigwedge_{i=1}^{n}p(x_i)\wedge \distinct{n}\wedge\forall\, x.\:[p(x)\rightarrow\bigvee_{i=1}^{n}x=x_{i}]]\]
\[\small \psi^{\diff}_{\geq n}=\exists\, \overarrow{x}.\:\bigwedge_{i=1}^{n}\neg p(x_i)\wedge \distinct{n} ~~~~~~~
\small \psi^{\diff}_{=n}=\exists\, \overarrow{x}.\:[\bigwedge_{i=1}^{n}\neg p(x_i)\wedge \distinct{n}\wedge\forall\, x.\:[\neg p(x)\rightarrow\bigvee_{i=1}^{n}x=x_{i}]]\]
\[\psiv=\forall\, x.\:\big[(s(s(x))=x)\vee(s(s(x))=s(x))\big]\]
\end{mdframed}
\caption{Formulas for $\Sigmas$-theories.
$\overarrow{x}$ stands for $x_1,\ldots, x_n$.
$\distinct{n}$ stands for $\bigwedge_{1\leq i<j\leq n}\neg(x_{i}=x_{j})$, and
$p(x)$ stands for $s(x)=x$.
}
\label{fig-card-s}
\end{figure}

For completeness, we include in \Cref{tab-theories-tab-old} 
the axiomatizations of
all theories from
\cite{BarTolZoh} that are used in \Cref{tab-summary} 
(\Cref{fig-card-s} includes the definitions of formulas that are abbreviated in \Cref{tab-theories-tab-old},
such as $\psi^{\eq}_{\geq n}$ from the definition of $\TsM$).
For lack of space, however, we refrain from elaborating on these theories,
and refer the reader to their detailed description in \cite{BarTolZoh}.
For the theories of \Cref{tab-theories-tab-old}, whether they admit the properties
from $\{\stainf,\smooth,\finwit,\strfinwit,\convex\}$ or not was already established in \cite{BarTolZoh}.
For each of them, here, we also check and prove whether they admit the 
new properties $\finmodpro$ and $\stafin$.

For example,
for each $n$, $\Tgeqn$ consists of all $\Sigmaone$-structures that have
at least $n$ elements.
This theory was shown in \cite{BarTolZoh} to be strongly finitely witnessable, and so by \Cref{SFW=>SF} it is also
stably finite.
Then, by \Cref{SF=>FMP}, it also admits the finite model property.

It is worth mentioning that
$\Ttwo$ was first introduced in \cite{Casal2018}, in the context of shiny theories, where it was shown to have the finite model property, while
not being stably finite.
An alternative proof of this fact goes as follows:
it was proven in \cite{SZRRBT-21} that $\Ttwo$ is:
$(i)$~finitely witnessable;
$(ii)$~not strongly finitely witnessable; and
$(iii)$~smooth.
By \Cref{FW=>FMP} and $(i)$, it also has the finite model property.
But since it is over an empty signature, by $(ii)$, $(iii)$ and \Cref{OS+ES+SM+SF=>SFW}, 
we have that it cannot be
stably finite.

\subsection{New Theories: The Simple Cases}
\label{sec:simpletheories}

\newcommand{\diag}[3]{diag^{{{#1},{#2}}}{({#3})}}
\newcommand{\bbdiag}[3]{diag_{\bb}^{{{#1},{#2}}}{({#3})}}

\begin{figure}[t]
  \centering
  \renewcommand{\arraystretch}{1.4}
\centering
\begin{tabular}{c|c|c}
Name & Signature & Axiomatization\\
\hline
$\Tupinfty$ & $\Sigmatwo$ & $\{(\psi^{\s}_{=1}\wedge\psi^{\s_{2}}_{\geq k})\vee \diag{\s}{\s_{2}}{k+2}: k\in\mathbb{N}\}$\\
$\Tmninfty$ & $\Sigmatwo$ & $\{\psi^{\s}_{=\max\{m,n\}}\vee(\psi_{=\min\{m,n\}}^{\s}\wedge\psi^{\s_{2}}_{\geq k}) : k\in\mathbb{N}\}$\\\hline
$\Tsupinfty$& $\Sigmastwo$& $\{(\psi^{\s}_{=1}\wedge\psi^{\s_{2}}_{\geq k})\vee(\diag{\s}{\s_{2}}{k+2}\wedge\Forall{x}\neg p(x): k\in\mathbb{N}\}$\\\hline
$\Ttwocube$ &$\Sigma_{3}$ & $\{\psi^{\s_{3}}_{=1}\}\cup\{(\psi^{\s}_{=2}\wedge\psi^{\s_{2}}_{\geq k})\vee(\psi^{\s}_{\geq 3}\wedge \psi^{\s_{2}}_{\geq 3}) : k\in\mathbb{N}\}$\\
\end{tabular}
\renewcommand{\arraystretch}{1}
\vspace{2em}
  \captionof{figure}{Simple theories for \Cref{tab-summary}. $\diag{\s}{\s_{2}}{k+2}$, for any $k\in\mathbb{N}$, stands for the formula $(\psi^{\s}_{\geq k+2}\wedge\psi^{\s_{2}}_{\geq k+2})\vee\bigvee_{i=2}^{k+2}(\psi^{\s}_{=i}\wedge\psi^{\s_{2}}_{=i})$, and $p(x)$ stands for $s(x)=x$.}
  \label{tab-theories-sigma-22}
\end{figure}

While the theories from \Cref{tab-theories-tab-old}
suffice to populate many cells of \Cref{tab-summary},
they are not enough.
Hence we describe new theories, 
not taken from \cite{BarTolZoh}.
The simplest theories that we have added can be found in
\Cref{tab-theories-sigma-22}, and are described below.

$\Tupinfty$ is a theory with three distinct groups of models: 
its first group consists of models $\A$ that have $|\s^{\A}|=1$ and $\s_{2}^{\A}$ infinite; 
its second group, of models $\A$ 
where both $\s^{\A}$ and $\s_{2}^{\A}$ are infinite; 
and its third group, of models $\A$ where $|\s^{\A}|=|\s_{2}^{\A}|$ is any value $k\geq 2$. 
In its axiomatization, one finds the formula $\diag{\s}{\s_{2}}{k+2}$, equal to $(\psi^{\s}_{\geq k+2}\wedge\psi^{\s_{2}}_{\geq k+2})\vee\bigvee_{i=2}^{k+2}(\psi^{\s}_{=i}\wedge\psi^{\s_{2}}_{=i})$ for $k\in\mathbb{N}$: that formula characterizes the models $\A$ of $\Tupinfty$ that lie in the diagonal, that is, where $|\s^{\A}|=|\s_{2}^{\A}|$ (and this value is greater than $1$).

$\Tmninfty$ is a theory that depends on two distinct positive integers $m$ and $n$, and without loss of generality let us suppose $m>n$, when the theory has two types of models $\A$: in the first, $|\s^{\A}|$ equals $m$, while $\s_{2}^{\A}$ can be anything; in the second, $|\s^{\A}|$ equals $n$, and then $\s_{2}^{\A}$ must be infinite. 

The models $\A$ of the $\Sigmastwo$-theory $\Tsupinfty$ have either: $|\s^{\A}|=1$, $|\s_{2}^{\A}|\geq\omega$ and $s^{\A}$ the identity function; both $\s^{\A}$ and $\s_{2}^{\A}$ infinite, and $s^{\A}$ with no fixed points; or $|\s^{\A}|=|\s_{2}^{\A}|$ equal to any number in $\mathbb{N}\setminus\{0,1\}$, and again $s^{\A}$ with no fixed points.  

Finally,  $\Ttwocube$ is made up of just the models $\A$ of $\Ttwo$ (see \Cref{tab-theories-tab-old}) with an extra domain associated to the new sort $\s_{3}$ such that $|\s_{3}^{\A}|=1$.

~~~~

\subsection{New Theories: The Busy Beaver}
\label{sec:theoriesBB}

\begin{figure}[t]
\centering
\renewcommand{\arraystretch}{1.4}
\centering
\begin{tabular}{c|c|c}
Name & Signature & Axiomatization\\
\hline
$\TBB$ & $\Sigmaone$ & $\{\psi_{\geq \bb(k+2)}\vee\bigvee_{i=2}^{k+2}\psi_{=\bb(i)}: k\in\mathbb{N}\}$\\\hline
$\TBBtwo$ & $\Sigmatwo$ & $\{(\psi^{\s}_{=1}\wedge\psi^{\s_{2}}_{\geq k})\vee
\bbdiag{\s}{\s_{2}}{k+2} : k\in\mathbb{N}\}$\\
$\TBBn$ & $\Sigmatwo$ & $\{\psi^{\s}_{=n}\}\cup\{\psi^{\s_{2}}_{\geq \bb(k+2)}\vee\bigvee_{i=2}^{k+2}\psi^{\s_{2}}_{=\bb(i)}: k\in\mathbb{N}\}$\\
$\TmnBB$ & $\Sigmatwo$ & $\{(\psi^{\s}_{n}\wedge\psi^{\s_{2}}_{\geq k})\vee(\psi^{\s}_{m}\wedge\psi^{\s_{2}}_{\geq \bb(k+1)})\vee\bigvee_{i=1}^{k+1}(\psi^{\s}_{m}\wedge\psi^{\s_{2}}_{=\bb(i)}) : k\in\mathbb{N}\}$\\\hline
$\TsBB$ &$\Sigmas$& $\{(\psi_{\geq k+1}\wedge\psi^{=}_{\geq \qq(k+1)})\vee\bigvee_{i=1}^{k+1}(\psi_{=i}\wedge\psi^{=}_{=\qq(i)}) : k\in\mathbb{N}\}$\\
$\TneqBB$ &$\Sigmas$& $\{(\psi_{=2}\wedge\Forall{x}\neg p(x))\vee((\psi_{\geq \bb(k+2)}\vee\bigvee_{i=2}^{k+2}\psi_{=\bb(i)})\wedge\Forall{x}p(x)) : k\in\mathbb{N}\}$\\
$\TneqBBone$ &$\Sigmas$& $\{\psi_{=1}\vee((\psi_{\geq \bb(k+2)}\vee\bigvee_{i=2}^{k+2}\psi_{=\bb(i)})\wedge\Forall{x}\neg p(x)) : k\in\mathbb{N}\}$\\
$\TssBB$ &$\Sigmas$& $\{\psiv\}\cup\{(\psi_{\geq k+1}\wedge\psi^{=}_{\geq \qq(k+1)})\vee\bigvee_{i=1}^{k+1}(\psi_{=i}\wedge\psi^{=}_{=\qq(i)}) : k\in\mathbb{N}\}$\\\hline
$\TtwoBB$ &$\Sigmastwo$& $\{\psi^{=}_{\geq k+2}\rightarrow\psi^{\s_{2}}_{\geq \bb(k+2)} : k\in\mathbb{N}\}$\\
$\TBBtwotwo$ &$\Sigmastwo$& $\{(\psi^{\s}_{=1}\wedge\psi^{\s_{2}}_{\geq k})\vee(\psi^{=}_{\geq k+2}\rightarrow\psi^{\s_{2}}_{\geq \bb(k+2)}) : k\in\mathbb{N}\}$\\
$\TstwoBB$ &$\Sigmastwo$& $\{\psiv\}\cup\{(\psi^{\s}_{=1}\wedge\psi^{\s_{2}}_{\geq k})\vee(\psi^{=}_{\geq k+2}\rightarrow\psi^{\s_{2}}_{\geq \bb(k+2)}) : k\in\mathbb{N}\}$\\
$\TsBBtwo$ &$\Sigmastwo$ & $\{(\psi^{\s}_{=1}\wedge\psi^{\s_{2}}_{\geq k})\vee
(\bbdiag{\s}{\s_{2}}{k+2}\wedge\Forall{x}\neg p(x)) : k\in\mathbb{N}\}$\\\hline
$\TBBtwocube$ &$\Sigma_{3}$ & $\{\psi^{\s_{3}}_{=1}\}\cup\{(\psi^{\s}_{=1}\wedge\psi^{\s_{2}}_{\geq k})\vee
\bbdiag{\s}{\s_{2}}{k+2} : k\in\mathbb{N}\}$\\
\end{tabular}
\renewcommand{\arraystretch}{1}
\vspace{2em}
  \captionof{figure}{Busy Beaver Theories for \Cref{tab-summary}. $\bbdiag{\s}{\s_{2}}{k+2}$ stands, for each $k\in\mathbb{N}$, for $(\psi^{\s}_{\geq\bb(k+2)}\wedge\psi^{\s_{2}}_{\geq\bb(k+2)})
\vee\bigvee_{i=2}^{k+2}(\psi^{\s}_{=\bb(i)}\wedge\psi^{\s_{2}}_{=\bb(i)})$, and $p(x)$ for $s(x)=x$; in $\TmnBB$, we assume w.l.g. $m\geq n$.
}
  \label{tab-theories-sigma-s}
\end{figure}

So far we have seen that the theories from \cite{BarTolZoh}, together 
with a small set of simple new theories, 
can already get us quite far in filling \Cref{tab-summary}.
However, for several combinations, it seems that more complex theories are needed.
For this purpose, we utilize the well-known Busy Beaver function,
and define various theories based on it.
In this section, we describe these theories.
First, in \Cref{sec:onbb}, we review the Busy Beaver function, and
explain why it is useful in our context.
Then, in \Cref{sec:bbtheories-sig1,sec:bbtheories-sig2,sec:bbtheories-sigs,sec:bbtheories-sigstwo,sec:bbtheories-sig3}, we describe the theories that make use of 
it, separated according to their signatures.


\subsubsection{On The Busy Beaver Function}
\label{sec:onbb}
The Busy Beaver function, here denoted $\bb$, is an old acquaintance of theoretical computer scientists: essentially, given any $n\in\mathbb{N}$, $\bb(n)$ 
is the maximum number of $1$'s a Turing machine 
with at most $n$ states can write to it's tape
when it halts,
if the tape is initialized to be all $0$'s. Somewhat confusingly, any Turing machine that achieves that number is also called a Busy Beaver. 

It is possible to prove that $\bb(n)\in\mathbb{N}$ for any $n\in\mathbb{N}$ (see \cite{Rado}), and so we may write $\bb:\mathbb{N}\rightarrow\mathbb{N}$; furthermore, $\bb$ is increasing. But the very desirable property of $\bb$ is that it is not only increasing, but actually very rapidly increasing. 

More formally, Rad{\'o} proved, in the seminal paper \cite{Rado}, that $\bb$ grows asymptotically faster than any computable function (being, therefore, non-computable). 
That is, for every computable function $f:\mathbb{N}\rightarrow\mathbb{N}$, there exists $N\in\mathbb{N}$ such that $\bb(n)>f(n)$ for all $n\geq N$. Despite that, the Busy Beaver starts somewhat slowly: $\bb(0)=0$, $\bb(1)=1$, $\bb(2)=4$, $\bb(3)=6$ and $\bb(4)=13$; the exact value of $\bb(5)$ (and actually $\bb(n)$ for any $n\geq 5$) is not known, but is at least $4098$ (\cite{Marxen}).

The fact that $\bb$ grows eventually faster than any computable function is a great property to have 
when constructing
theories that admit the finite model property, while not being finitely witnessable.
Roughly speaking, if the cardinalities of models of a theory are related to $\bb$, this guarantees that it has models of sufficiently large finite size, while not being finitely witnessable since its models grow too fast: by carefully choosing formulas $\phi_{n}$ 
that hold only in the "$n$-th model" of the theory 
(when ordered by cardinality), the number of variables of $\wit(\phi_{n})$ offers an upper bound to $\bb(n)$ and is therefore not computable, leading to a contradiction with the fact that $\wit$ is supposed to be computable. Notice that, despite the dependency of our theories on the Busy Beaver, the function is not actually part of their signatures.

Now we present the theories that are based on $\bb$.
These theories are axiomatized in \Cref{tab-theories-sigma-s}.


\subsubsection{A $\Sigmaone$-Theory}
\label{sec:bbtheories-sig1}
The most basic Busy Beaver theory is $\TBB$. 
This
is the $\Sigmaone$-theory whose models have cardinality $\bb(k)$, for some $k\geq 2$, or are infinite: that is, $\TBB$ has models with $4$ elements, $6$, $13$ and so on.
This theory forms the basis to all other theories of this section, that are designed to admit various properties from \Cref{tab-summary}.

By itself, $\TBB$ has the finite model property while not being (strongly) finitely witnessable. 
It was in fact constructed precisely to exhibit this.
As it turns out, it is also
not smooth, but does satisfy all other properties.
To populate other rows in the table that correspond to theories
with other combinations of properties,
more theories are needed, with richer signatures.


\subsubsection{$\Sigmatwo$-Theories}
\label{sec:bbtheories-sig2}
To fill the rows that correspond to other combinations, 
we introduce several
$\Sigmatwo$ theories.

The $\Sigmatwo$-theory $\TBBtwo$
is more complex. 
It has, essentially, three classes of models: the first is made up of structures $\A$ where $|\s^{\A}|=1$ and $\s_{2}^{\A}$ is infinite; the second, of structures where both $\s^{\A}$ and $\s_{2}^{\A}$ are infinite; and the third, of structures where $|\s^{\A}|=|\s_{2}^{\A}|$ is a finite value that equals $\bb(k)$, for some $k\geq 2$. The formula $\bbdiag{\s}{\s_{2}}{k+2}$, for $k\geq 2$, in the axiomatization equals $(\psi^{\s}_{\geq\bb(k+2)}\wedge\psi^{\s_{2}}_{\geq\bb(k+2)})
\vee\bigvee_{i=2}^{k+2}(\psi^{\s}_{=\bb(i)}\wedge\psi^{\s_{2}}_{=\bb(i)})$ and is similar to $\diag{\s}{\s_{2}}{k+2}$ from $\Tupinfty$, characterizing the models $\A$ where $|\s^{\A}|=|\s_{2}^{\A}|$ and that value, if it is not infinite, equals $\bb(k+2)$. 
%

For each $n>0$, $\TBBn$ has as interpretations those $\A$ with $|\s^{\A}|=n$, and $|\s_{2}^{\A}|$either infinite or equal to $\bb(k)$, for some $k\geq 2$ (so $(|\s^{\A}|, |\s_{2}^{\A}|)$ may equal $(n,4)$, $(n,6)$, $(n,13)$ and so on).

$\TmnBB$ is a $\Sigmatwo$-theory that can be seen as some sort of combination of $\Tmninfty$ and $\TBBn$, dependent on two distinct positive integers $m$ and $n$.
Consider the case where the former is the greater of the two (the other cases are similar).
In this case, we may divide its interpretations $\A$ into three classes: those with $|\s^{\A}|=n$ and $\s_{2}^{\A}$ infinite; those with $|\s^{\A}|=m$ and $\s_{2}^{\A}$ infinite; and those with $|\s^{\A}|=m$ and $|\s_{2}^{\A}|$ equal to some $\bb(k)$, for $k\geq 2$.


\subsubsection{$\Sigmas$-Theories}
\label{sec:bbtheories-sigs}
For some lines of \Cref{tab-summary}, e.g. line $7$,
empty signatures are not enough for presenting examples.
Hence we also introduce $\Sigmas$-theories.

We start with $\TsBB$, which
is, arguably, the most confusing theory we here define: we are forced to appeal not only to the special cardinality formulas found in \Cref{fig-card-s}, but also to the function $\qq$, which is a left inverse of $\bb$. More formally, $\qq:\mathbb{N}\rightarrow\mathbb{N}$ is the only function such that $\qq(k)=\min\{l : \bb(l+1)>\bb(k)\}$: so $\qq(0)=0$, $\qq(1)=\qq(2)=\qq(3)=1$, $\qq(4)=\qq(5)=2$, $\qq(6)=\cdots=\qq(12)=3$, $\qq(13)=\cdots=\qq(4097)=4$, and further values of $\qq$ are currently unknown.
From the definition of $\qq$, we have that $\bb(\qq(k))\leq k$ and $\qq(\bb(k))=k$.
$\qq$ is not computable given that, since $\qq(k)=\min\{l : \bb(l+1)>\bb(k)\}$ by definition, $\qq(k+1)\neq\qq(k)$ iff $k+1$ is a value of $\bb$: so, an algorithm to compute the values of $\bb$ could be obtained by simply computing the values of $\qq$ and checking where there is a change. 

$\TsBB$ is then the $\Sigmas$-theory with models $\A$ with any cardinality $k+1\geq 1$, such that $s^{\A}(a)=a$ holds for precisely $\qq(k+1)$ elements of $\A$, and so $s^{\A}(a)\neq a$ holds for $k+1-\qq(k+1)$ elements, being the function $k\mapsto k+1-\qq(k+1)$ itself non-decreasing, given that $\qq(k+1)$ can equal either $\qq(k)$ or $\qq(k)+1$. 

\begin{example}    
We mention some $\TsBB$-structures as examples:
a structure $\A$ with  $|\s^{\A}|=1$ and $s^{\A}$ the identity; 
a structure $\B$ with $|\s^{\B}|=2$ and $s^{\A}$ a constant function; 
a structure $\C$ with $|\s^{\C}|=3$ (say $\s^{\C}=\{a, b, c\}$) and $s^{\C}$ the identity for only one of these elements (e.g., $s^{\C}$ can be a constant function, but now there are further possibilities such as $s^{\C}(a)=s^{\C}(b)=a$ and $s^{\C}(c)=b$);  and
a structure $\D$ with $|\s^{\D}|=4$ (say $\s^{\D}=\{a, b, c, d\}$) and $s^{\D}$ the identity for only two of these elements (e.g., $s^{\D}(a)=s^{\D}(b)=s^{\D}(c)=a$ and $s^{\D}(d)=d$); 
\end{example}

Next, we continue to describe other $\Sigmas$ theories.

    $\TneqBB$ has essentially two classes of models $\A$: those with $|\s^{\A}|=2$ and $s^{\A}$ never the identity;
    and those with $|\s^{\A}|$ equal to $\bb(k)$ or infinite, for some $k\geq 2$, and $s^{\A}$ the identity.

    $\TneqBBone$ is very similar to $\TneqBB$: the difference lies on where $s$ will be the identity: while in $\TneqBB$ the function $s$ is the identity for all interpretations $\A$ with $|\s^{\A}|>2$, $s$ in $\TneqBBone$ is the identity only for the interpretations $\A$ with $|\s^{\A}|=1$. So, in $\TneqBBone$, we have a model $\A$ with $|\s^{\A}|=1$ and $s^{\A}$ the identity, and then models $\A$ with $|\s^{\A}|=\bb(k)$ for some $k\geq 2$ or infinite, and $s^{\A}(a)$ anything but $a$.

    The $\Sigmas$-theory $\TssBB$ is then just $\TsBB$, satisfying in addition the formula $\psiv$ (see \Cref{fig-card-s}). 
    It has models $\A$ of any finite cardinality $k+1$, as long as $\qq(k+1)$ of these elements $a$ satisfy $s^{\A}(a)=a$,
    or infinite cardinalities, as long as the number of elements $a$ satisfying $s^{\A}(a)=a$ is infinite; additionally, $s^{\A}(s^{\A}(a))$ must always equal either $s^{\A}(a)$ or $a$ itself.


\subsubsection{$\Sigmastwo$-Theories}
\label{sec:bbtheories-sigstwo}
Now for theories in a many-sorted non-empty signature.

    The $\Sigmastwo$-theory $\TtwoBB$ appears simple, but is actually quite tricky: starting by the easy case, if $\s^{\A}$ has infinitely many elements $a$ satisfying $s^{\A}(a)=a$, $\s_{2}^{\A}$ is also infinite. If, however, the number of elements $a\in\s^{\A}$ satisfying $s^{\A}(a)=a$ is finite (notice that, even if this is the case, $\s^{\A}$ may still be infinite) and equal to some $k+2$, then $\s_{2}^{\A}$ has at least $\bb(k+2)$ elements. So, to give a better example, suppose $\s^{\A}$ has $2$ elements satisfying $s^{\A}(a)=a$: then $\s_{2}^{\A}$ has at least $\bb(2)=4$ elements, but may have any cardinality up to, and including, infinite ones; notice that in this example $\s^{\A}$ may be infinite as well, as long as only two of the elements satisfy $s^{\A}(a)=a$.

$\TBBtwotwo$ is the same as $\TtwoBB$, but with extra models $\A$ where $|\s^{\A}|=1$ and $|\s_{2}^{\A}|\geq\omega$ (of course, then we have that $s^{\A}$ is the identity).

$\TstwoBB$ is then the same as $\TBBtwotwo$, with the added validity of the formula $\psiv$; So the models of $\TstwoBB$ are just models of $\TBBtwotwo$ satisfying that $s^{\A}(s^{\A}(a))$ equals either $s^{\A}(a)$ or $a$ itself.

$\TsBBtwo$ is just the $\Sigmatwo$-theory $\TBBtwo$ with the added function $s$ such that, if $|\s^{\A}|=1$, $s^{\A}$ is the identity; and if $|\s^{\A}|>1$, $s^{\A}(a)$ is anything but $a$. 


\subsubsection{A $\Sigma_{3}$-Theory}
\label{sec:bbtheories-sig3}
Finally, $\TBBtwocube$ is obtained by adding a sort with a single element to the $\Sigmatwo$-theory $\TBBtwo$,
similarly to the definition of $\Ttwocube$, that was based on the $\Sigmatwo$-theory $\Ttwo$ (see \Cref{sec:simpletheories}).

\subsection{Theory Operators}
\label{sec:theoryops}
There are two types of theories in \Cref{tab-summary}:
The first consists of base theories, such as $\Tgeqn$,
that 
are axiomatized in \Cref{tab-theories-tab-old,tab-theories-sigma-22,tab-theories-sigma-s}.
The second is obtained from the first, by applying several operators on theories.
For example, the theories $\adds{\Tgeqn}$, $\addf{\Tgeqn}$, $\addf{\adds{\Tgeqn}}$,
are all obtained from the base theory $\Tgeqn$.
So far we have only described the theories of the first type.
In this section we explain the theories of the second type.

The operators that are used in \Cref{tab-summary} were defined in \cite{BarTolZoh}, in order to be able to systematically
generate examples in various signatures.
For example, if $\T$ is a $\Sigma_{1}$-theory, then 
$\Tadds$ is a $\Sigma_{2}$-theory with the same axiomatization as $\T$,
that is, the second sort is completely free and is not axiomatized in any way.
For completeness sake, we include the definitions of these operators here:

\begin{definition}[Theory Operators from \cite{BarTolZoh}]~
\label{def:operators}
\begin{enumerate}
\item 
If $\T$ is a $\Sigmaone$-theory, then $\Tadds$ is the $\Sigma_{2}$-theory axiomatized by $\ax(\T)$.
\item Let $\Sigma_{n}$ be an empty signature with
sorts $S=\{\s_{1}, \ldots, \s_{n}\}$, 
and let $\T$ be a $\Sigma_{n}$-theory. 
The signature $\Sigma^{n}_{s}$ has sorts $S$ and a single unary function symbol
$s$ of arity $\s_{1}\rightarrow\s_{1}$, and
$\Taddf$ is the $\Sigma^{n}_{s}$-theory axiomatized by $\ax(\T)\cup\{\forall\, x.\:[s(x)=x]\}$, where $x$ is a variable of sort $\s_{1}$.
\item Let $\T$ be a theory over an empty signature with sorts 
 $S=\{\s_{1}, \ldots, \s_{n}\}$.
 Then 
 $\Taddnc$ is the
$\Sigma^{n}_{s}$-theory axiomatized by $\ax(\T)\cup\{\psiv\}$ (see \Cref{fig-card-s}.    
\end{enumerate}
\end{definition}

It was proven in \cite{BarTolZoh} that these operators preserve 
the properties $\stainf$, $\smooth$, $\finwit$, $\strfinwit$, $\convex$,
and the lack of them.
Here we prove that the same holds for
$\finmodpro$ and $\stafin$ as well.

\begin{restatable}{theorem}{addsresult}
\label{thm:addsresult}
Let $\T$ be a $\Sigmaone$-theory. Then: $\T$ is $\finmodpro$, or $\stafin$,  
w.r.t. $\{\s\}$ if and only if $\Tadds$ is, respectively, $\finmodpro$, or $\stafin$ w.r.t. $\{\s, \s_{2}\}$.
\end{restatable}

\begin{restatable}{theorem}{addfresult}
\label{thm:addfresult}
If $\T$ is a theory over an empty signature $\Sigma_{n}$ with sorts $S=\{\s_{1}, \ldots, \s_{n}\}$, then: $\T$
is $\finmodpro$, or $\stafin$, w.r.t. $S$ if and only if $\Taddf$ is, respectively, $\finmodpro$, or $\stafin$, w.r.t. $S$.
\end{restatable}

\begin{restatable}{theorem}{addncresult}
\label{thm:addncresult}
If $\T$ is a theory over an empty signature $\Sigma_{n}$ with sorts $S=\{\s_{1}, \ldots, \s_{n}\}$, then: $\T$
is $\finmodpro$, or $\stafin$, w.r.t. $S$ if and only if $\Taddnc$ is, respectively, $\finmodpro$, or $\stafin$, w.r.t. $S$.
\end{restatable}

Thus, in various cases, theories need not
be invented from scratch, but can be generated from other 
theories.
For example, the theory $\Tgeqn$ exhibits all studied properties,
but is defined in a one-sorted signature. 
Using the operators, we obtain variants of this theory
in all signature types, namely
$\adds{\Tgeqn}$ for empty many-sorted signatures,
$\addf{\Tgeqn}$ for non-empty one-sorted signatures, and $\addf{\adds{\Tgeqn}}$ for non-empty many-sorted signatures.
The properties of the theories generated using these operators are guaranteed by \Cref{thm:addsresult,thm:addfresult}, as well as the corresponding results from \cite{BarTolZoh}. 

In two cases of theories defined using the Busy Beaver function, $\TssBB$ and $\TstwoBB$, we cannot obtain them by relying on \Cref{thm:addncresult} from, respectively, $\TsBB$ and $\TtwoBB$, since the signatures of the latter theories are not empty.
Curiously,
adding $\psiv$ to their axiomatizations still has the desirable outcome, but we prove this separately, without relying on \Cref{thm:addncresult}.
Extending \Cref{thm:addncresult} to non-empty signatures is left for future work.

The number of combinations of properties that we consider,
together with the possible types of the signatures,
adds up to $2^{9}=512$.
Our negative results from \Cref{Finitemodelpropertyand stablefiniteness} guarantee that only $\sim$15\% of the actual table can be filled with examples.
The remaining $\sim$85\% are either colored in red or are excluded from the table 
for space considerations.
As for the examples that can be given, 
notice that there are in total an astonishing number of $78$ theories in our table.
But, thanks to the theory operators of \Cref{def:operators}, only $33$ of them ($\sim$42\%) had to be 
concretely axiomatized in \Cref{tab-theories-tab-old,tab-theories-sigma-22,tab-theories-sigma-s}. The remaining
$45$ theories were defined using the operators.

\section{Conclusion}\label{conclusion}

We examined, in addition to all properties considered in \cite{BarTolZoh}, the finite model property, and stable finiteness. 
Interesting restrictions for the combinations involving these properties were established.
We also found interesting theories to fill in our table of combinations, most prominently those involving the Busy Beaver function as well as its inverse. 

One possible direction this research could take is reasonably clear: considering the computability of the $\mc$ function, what will, most probably, double the number of theories to be taken into consideration. Further interesting properties that could be considered include the decidability of the theory's axiomatization, or even its finiteness, and the satisfiability problem of the theory with respect to quantifier-free formulas.

Second, some of the negative results in \cite{BarTolZoh} and in the present paper only hold with respect to the entire set of sorts $\S_{\Sigma}$. We plan to study if they hold also with respect to proper subsets of sorts, and if they do not, to provide counterexamples to those generalizations. 


\printbibliography


\newgeometry{
a4paper,
total={160mm,220mm},
left=25mm,
top=25mm,
}
\appendix

\section{Preliminaries of the Appendix}

The following are basic results in model theory of many-sorted logic that may be found in \cite{Monzano93}.

\begin{lemma}[\cite{Monzano93}]\label{Compactness}
Let $\Sigma$ be a signature; a set of $\Sigma$-formulas $\Gamma$ is satisfiable if and only if each of its finite subsets $\Gamma_{0}$ is satisfiable.
\end{lemma}

\begin{lemma}[\cite{Monzano93}]
\label{LowenheimSkolemDownwards}
Let $\Sigma$ be a signature and $\Gamma$ a set of $\Sigma$-formulas: if $\Gamma$ is satisfiable, there exists an interpretation $\A$ which satisfies $\Gamma$ and where $\s^{\A}$ is countable whenever it is infinite.
\end{lemma}

The two following theorems are results from \cite{BarTolZoh}.

\begin{restatable}{theorem}{sietac}
\label{SI empty theories are convex}
Any theory over an empty signature that is stably infinite w.r.t. the set of all of its sorts is convex w.r.t. any set of sorts.
\end{restatable}

\begin{restatable}{theorem}{osessifw}
\label{OS+ES+-SI=>FW}
Every one-sorted, non-stably-infinite theory $\T$ with an empty signature is finitely witnessable w.r.t. its only sort.
\end{restatable}

We shall also make use of the following fact, found as a corollary of exercise $10.5$ in Section $10$ of \cite{BradleyManna2007}.

\begin{theorem}\label{uninterpretedfunctionsis convex}
    If $\T$ is a theory over a one-sorted signature with only one unary function, and if $\T$ is axiomatized by the empty set (that is, its functions are \emph{uninterpreted}), then $\T$ is convex.
\end{theorem}

Given a function symbol $s$ of arity $\s\rightarrow\s$ and a variable $x$ of sort $\s$, it is useful to define, recursively, the term $s^{k}(x)$ (for $k\in\mathbb{N}$) as follows: $s^{0}(x)=x$, and $s^{k+1}(x)=s(s^{k}(x))$.

\section{Proof of \Cref{FW=>FMP}}

\FWimpliesFMP*

\begin{proof}
    Let $\wit$ be the witness for $\T$ with respect to $S$, $\phi$ be a quantifier-free $\Sigma$-formula, and $\A$ a $\T$-interpretation that satisfies $\phi$:  $\wit$ is a witness, and 
    therefore $\phi$ and $\Exists{\overarrow{x}}\wit(\phi)$ (for $\overarrow{x}=\vars(\wit(\phi))\setminus\vars(\phi)$) must be $\T$-equivalent, and so we must have that $\A$ satisfies $\Exists{\overarrow{x}}\wit(\phi)$; and then there must exist a $\T$-interpretation $\Ap$, different from $\A$ at most on $\overarrow{x}$, that satisfies $\wit(\phi)$. Using again that $\wit$ is a witness, there is a $\T$-interpretation $\B$, with $\s^{\B}=\vars_{\s}(\wit(\phi))^{\B}$ for each $\s\in S$, that satisfies $\phi$. But, because $\vars_{\s}(\wit(\phi))$, and thus $\vars_{\s}(\wit(\phi))^{\B}$, must be finite, we get $|\s^{\B}|$ is finite for each $\s\in S$, proving $\T$ has the finite model property with respect to $S$.
\end{proof}

\section{Proof of \Cref{SFW=>SF}}

\SFWimpliesSF*

\begin{proof}
    Let $\wit$ be the strong witness of $\T$ with respect to $S$, $\phi$ be a quantifier-free $\Sigma$-formula, and $\A$ a $\T$-interpretation that satisfies $\phi$: we have that $\A$ satisfies $\Exists{\overarrow{x}}\wit(\phi)$, for $\overarrow{x}=\vars(\wit(\phi))\setminus\vars(\phi)$; so there must exist a $\T$-interpretation $\Ap$, differing from $\A$ at most on $\overarrow{x}$, that satisfies $\wit(\phi)$. Let $V$ be the set of variables of sort in $S$ on $\wit(\phi)$ or $\phi$, and $\delta_{V}$ be the arrangement on $V$ such that $x$ is related to $y$ iff $x^{\Ap}=y^{\Ap}$, which is clearly satisfied by $\Ap$; by strong finite-witnessability, 
    there must exist a $\T$-interpretation $\B$ that satisfies $\wit(\phi)\wedge\delta_{V}$ with $\s^{\B}=\vars_{\s}(\wit(\phi)\wedge\delta_{V})^{\B}$
    for each $\s\in S$. Because $\B$ satisfies $\wit(\phi)$, it also satisfies $\Exists{\overarrow{x}}\wit(\phi)$ and $\phi$; and $|\s^{\B}|$ is finite for each $\s\in S$, given $\vars_{\s}(\wit(\phi)\wedge\delta_{V})$, and thus $\vars_{\s}(\wit(\phi)\wedge\delta_{V})^{\B}$, are finite; finally, because $\B$ satisfies $x=y$ (for $x$ and $y$ variables of $\wit(\phi)$ or $\phi$) iff $\Ap$ satisfies $x=y$, we have that $|\vars_{\s}(\wit(\phi)\wedge\delta_{V})^{\B}|=|\vars_{\s}(\wit(\phi)\wedge\delta_{V})^{\Ap}|$, and so
    \[|\s^{\B}|=|\vars_{\s}(\wit(\phi)\wedge\delta_{V})^{\B}|=|\vars_{\s}(\wit(\phi)\wedge\delta_{V})^{\Ap}|\leq|\s^{\Ap}|=|\s^{\A}|,\]
    proving $\T$ is stably finite with respect to $S$.    
\end{proof}

\section{Proof of \Cref{OS+FMP=>SF}}

\OSplusFMPimplySF*
\begin{proof}
    Let $\phi$ be a quantifier-free $\Sigma$-formula, and $\A$ a $\T$-interpretation that satisfies $\phi$: because $\T$ has the finite model property, there is a $\T$-interpretation $\B$ that satisfies $\phi$ with $|\s^{\B}|$ finite. There are then two cases to consider: if $|\s^{\B}|\leq|\s^{\A}|$, we already have that $\T$ is stably finite; if $|\s^{\B}|>|\s^{\A}|$, then $\A$ is a $\T$-interpretation that satisfies $\phi$ with $|\s^{\A}|$ finite (and this value is equal to or less than $|\s^{\A}|$), so $\T$ is, again, stably finite.
\end{proof}

\section{Proof of \Cref{OS+ES+SM+FMP=>FW}}

\begin{lemma}\label{verytechnicallemma}
    Let $n$ be a natural number, and consider any subset $A$ of $\mathbb{N}^{n}$ equipped with the order such that $(p_{1},\ldots,p_{n})\leq (q_{1}, \ldots, q_{n})$ iff $p_{i}\leq q_{i}$ for all $1\leq i\leq n$: then $A$ possesses at most a finite number of minimal elements under this order.
\end{lemma}

\begin{proof}
    We will prove the result by induction on $n$, the result being obviously true if $n=0$. So assume the result is true for $n$, and we shall prove it also holds for $n+1$. We start by fixing a minimal element $(p_{1}, \ldots , p_{n}, p_{n+1})$ of $A$ (if there are none, we are done): for any other minimal element $(q_{1}, \ldots , q_{n}, q_{n+1})$, we have that there must exist distinct $1\leq i,j\leq n+1$ such that $p_{i}>q_{i}$, and $p_{j}<q_{j}$; this way, $(q_{1}, \ldots , q_{n}, q_{n+1})$ is a minimal element of
    \[A^{i}_{q_{i}}=\{(r_{1}, \ldots , r_{n}, r_{n+1}) \in A : r_{i}=q_{i}\}\] 
    with the order induced by the order on $A$. By induction hypothesis, given this set may be simply considered as a subset of $\mathbb{N}^{n}$ if we omit the coordinate $i$, the set of minimal elements of $A^{i}_{q_{i}}$ is finite. So every minimal element of $A$ is either $(p_{1}, \ldots, p_{n}, p_{n+1})$, or an element of some $\min(A^{i}_{q_{i}})$, the set of minimal elements of $A^{i}_{q_{i}}$: there are $n+1$ possible values for $i$ (\textit{i.e.} $1$ through $n+1$), and $p_{i}-1$ values for $q_{i}$ (namely $1$ through $p_{i}-1$, given that $q_{i}<p_{i}$); so we can bound the number of minimal elements of $A$ by the sum 
    \[1+\sum_{i=1}^{m}\sum_{q_{i}=1}^{p_{i}-1}|\min(A^{i}_{q_{i}})|,\]
    which is finite, since it is a finite sum of finite numbers $|\min(A^{i}_{q_{i}})|$. The result is then true for $n+1$.
    \end{proof}

\OSESSMFMPimplyFW*

\begin{proof}
Suppose $\S_{\Sigma}=\{\s_{1},\ldots,\s_{n}\}$.
    Consider the set
    \[A=\{(|\s_{1}^{\A}|, \ldots , |\s_{n}^{\A}|) : \text{$\A$ is a $\T$-interpretation}\}\cap\mathbb{N}^{n},\]
    which is not empty because $\T$ has the finite model property: by \Cref{verytechnicallemma} there exist $\T$-interpretations $\A_{1}$ through $\A_{m}$ that correspond to the minimal elements of $A$ (notice we are using that two $\Sigma$-structures $\A$ and $\B$ are isomorphic iff $|\s_{i}^{\A}|=|\s_{i}^{\B}|$ for each $1\leq i\leq n$, given the signature is empty; this is necessary as otherwise one would need to consider $A$ as a multiset to account for non-isomorphic models $\A$ and $\B$ that still have minimal cardinalities $(|\s_{1}^{\A}\, \ldots, |\s_{n}^{\A}|)=(|\s_{1}^{\B}|, \ldots, |\s_{n}^{\B}|)$, and in this case \Cref{verytechnicallemma} does not necessarily hold); so we define $m_{i,j}=|\s_{i}^{\A_{j}}|$ and $m_{i}=\max\{m_{i,j} : 1\leq j\leq m\}$. Once we fix a quantifier-free $\Sigma$-formula $\phi$, take fresh variables $\{x_{1}^{i}, \ldots, x_{m_{i}}^{i}\}$ of sort $\s_{i}$, and consider the witness
    \[\wit(\phi)=\phi\wedge\delta,\quad\text{where}\quad\delta=\bigwedge_{i=1}^{n}\bigwedge_{1\leq j<k\leq m_{i}}(x^{i}_{j}=x^{i}_{k}).\]
    To prove $\phi$ and $\Exists{\overarrow{x}}\wit(\phi)$ are $\T$-equivalent is easy: start with a $\T$-interpretation $\A$ that satisfies $\phi$, and since $\delta$ is satisfied by all $\T$-interpretations by design, $\A$ satisfies $\phi\wedge\delta$, and thus $\Exists{\overarrow{x}}\wit(\phi)$. Conversely, if the $\T$-interpretation $\A$ satisfies $\Exists{\overarrow{x}}\wit(\phi)$, given that none of the variables in $\overarrow{x}$ occur in $\phi$, $\A$ satisfies $\phi\wedge\Exists{\overarrow{x}}\delta$, and thus $\phi$. 

    Now, suppose the $\T$-interpretation $\A$ satisfies $\wit(\phi)$: we then take sets $A_{i}$, for each $1\leq i\leq n$, disjoint from each other and $\s_{i}^{\A}$, with
    \[\max\{0,\quad m_{i}-|\vars_{\s_{i}}(\phi)^{\A}|\}\]
    elements. Finally, we define an interpretation $\B$ by making: $\s_{i}^{\B}=\vars_{\s_{i}}(\phi)^{\A}\cup A_{i}$ (which has at least $m_{i}$ elements for every sort $\s_{i}$, what makes of $\B$ a $\T$-interpretation, given this theory is smooth on an empty signature). And, for the evaluation of the variables, we demand that: $x^{\B}=x^{\A}$ for all variables $x$ in $\phi$; the map taking $x_{j}^{i}$ in $\{x^{i}_{1}, \ldots , x^{i}_{m_{i}}\}$ to $(x^{i}_{j})^{\B}$ in $ A_{i}$ is a surjection for each $1\leq i\leq n$ (what is possible by definition of $A_{i}$); and for other variables $x$, $x^{\B}$ is defined arbitrarily. Then it is true that: $\B$ satisfies $\phi$, since the truth-value of $\phi$, as a formula on the empty signature, is determined by the truth-value of its atomic subformulas $x=y$, and $x^{\B}=y^{\B}$ iff $x^{\A}=y^{\A}$ (since $x^{\B}=x^{\A}$ and $y^{\B}=y^{\A}$); $\B$ satisfies $\wit(\phi)$, since $\delta$ is $\T$-valid and $\B$ is a $\T$-interpretation (given $\T$ is smooth); and $\s^{\B}=\vars_{\s}(\wit(\phi))^{\B}$ for each $\s\in \S_{\Sigma}$, by definition of $\B$.
\end{proof}


\section{Proof of \Cref{OS+ES+SM+SF=>SFW}}

\OSESSMSFimpliesSFW*

\begin{proof}
    Suppose $\S_{\Sigma}=\{\s_{1},\ldots,\s_{n}\}$.
    Consider the set (also used in \Cref{OS+ES+SM+FMP=>FW})
    \[A=\{(|\s_{1}^{\A}|, \ldots , |\s_{n}^{\A}|) : \text{$\A$ is a $\T$-interpretation}\}\cap\mathbb{N}^{n}:\]
    $A$ is not empty because $\T$ is stably finite, and by \Cref{verytechnicallemma} there exist $\T$-interpretations $\A_{1}$ through $\A_{m}$ such that $(|\s_{1}^{\A_{1}}|, \ldots , |\s_{n}^{\A_{n}}|)$ through $(|\s_{1}^{\A_{m}}|, \ldots , |\s_{n}^{\A_{m}}|)$ are the minimal elements of $A$ (under the order described in \Cref{verytechnicallemma}). We then take $m_{i,j}=|\s_{i}^{\A_{j}}|$ and $m_{i}=\max\{|\s_{i}^{\A_{j}}| : 1\leq j\leq m\}$. Given a quantifier-free $\Sigma$-formula $\phi$, we take fresh variables $x_{1}^{i}$ through $x_{m_{i}}^{i}$ of sort $\s_{i}$, and consider the strong witness
    \[\wit(\phi)=\phi\wedge\delta,\quad\text{where}\quad\delta=\bigvee_{j=1}^{m}\bigwedge_{i=1}^{n}\bigwedge_{1\leq k<l\leq m_{i,j}}\neg(x^{i}_{k}=x^{i}_{l}).\]
    If a $\T$-interpretation $\A$ satisfies $\phi$, take the set of variables $V$ of $\phi$, and the arrangement $\delta_{V}$ on $V$ such that $x$ is related to $y$ iff $x^{\A}=y^{\A}$, and it is clear that $\A$ satisfies $\delta_{V}$: we can find a $\T$-interpretation $\B$ that satisfies $\delta_{V}$ with $|\s_{i}^{\B}|$ finite and satisfying $|\s_{i}^{\B}|\leq|\s_{i}^{\A}|$ for each $1\leq i\leq n$. Then, we can construct a $\T$-interpretation $\Bp$, that differs from $\B$ at most on $\overarrow{x}=\vars(\wit(\phi))\setminus\vars(\phi)=\vars(\delta)$, where $\delta$ is satisfied (since we must have $(|\s_{1}^{\Bp}|, \ldots , |\s_{n}^{\Bp}|)\geq (|\s_{1}^{\A_{j}}|, \ldots , |\s_{n}^{\A_{j}}|)$ for some $1\leq j\leq m$), and thus $\Bp$ satisfies $\delta$ and $\delta_{V}$, the last formula given $\B$ satisfies $\delta_{V}$ and $\Bp$ agrees with $\B$ on the values assigned to the variables in $\delta_{V}$. Finally, using the smoothness of $\T$ and the fact that we are dealing with an empty signature, we can find a $\T$-interpretation $\C$ that satisfies $\delta_{V}$ and $\delta$ whose underlying structure is isomorphic to that of $\A$ (that is, $|\s_{i}^{\C}|=|\s_{i}^{\A}|$ for each $1\leq i\leq n$); defining the $\T$-interpretation $\Ap$, that differs from $\A$ at most on the values assigned to $\overarrow{x}$, such that $(x^{i}_{k})^{\Ap}$ equals $(x^{i}_{l})^{\Ap}$ iff $(x^{i}_{k})^{\C}$ equals $(x^{i}_{l})^{\C}$, we get that $\Ap$ satisfies $\phi\wedge\delta$, and so $\A$ satisfies $\Exists{\overarrow{x}}\wit(\phi)$. Conversely, if $\A$ satisfies $\Exists{\overarrow{x}}\wit(\phi)$, again given the fact $\phi$ has none of the variables in $\overarrow{x}$, $\A$ satisfies $\phi\wedge\Exists{\overarrow{x}}\delta$ and thus $\phi$, proving $\phi$ and $\Exists{\overarrow{x}}\wit(\phi)$ are $\T$-equivalent. 

    Now, let $V$ be a set of variables with arrangement $\delta_{V}$, and suppose the $\T$-interpretation $\A$ satisfies $\wit(\phi)\wedge\delta_{V}$. We define a new $\T$-interpretation $\B$ by making: $\s_{i}^{\B}=\vars_{\s_{i}}(\wit(\phi)\wedge\delta_{V})^{\A}$ (which has at least $m_{i,j}$ elements, $(x^{i}_{1})^{\A}$ through $(x^{i}_{m_{i,j}})^{\A}$, for every sort $\s_{i}$ and $j$ an index such that $\A$ satisfies $\bigwedge_{i=1}^{n}\bigwedge_{1\leq k<l\leq m_{i,j}}\neg(x^{i}_{k}=x^{i}_{l})$, what makes of $\B$ indeed a $\T$-interpretation, given this theory is smooth on an empty signature); $x^{\B}=x^{\A}$ for all variables $x$ in $\wit(\phi)\wedge\delta_{V}$; and arbitrarily for other variables. Then it is true that $\B$ satisfies $\wit(\phi)\wedge\delta_{V}$, and has $\s^{\B}=\vars_{\s}(\wit(\phi)\wedge\delta_{V})^{\B}$ for each $\s\in \S_{\Sigma}$.
\end{proof}

\section{Proof of \Cref{threesorted}}

\threesorted*

\begin{proof} 
    If $\T$ is a $\Sigmaone$-theory, we just need to apply \Cref{OS+ES+-SI=>FW}.

    Assume then $\T$ is a $\Sigmatwo$-theory that has the finite model property with respect to $\{\s, \s_{2}\}$, but is neither stably infinite nor stably finite with respect to that set: we will prove 
    that then $\T$ is not convex. Because $\T$ is not stably infinite, \Cref{Compactness} guarantees that there exists an $M\in\mathbb{N}\setminus\{0\}$ such that there are no $\T$-interpretations $\A$ with $|\s^{\A}|, |\s_{2}^{\A}|\geq M$. 
    
     If $\T$ has no interpretations $\A$ with $|\s^{\A}|> M$ (the case with no interpretations where $|\s_{2}^{\A}|> M$ is analogous), 
     because $\T$ is not stably finite there must exist a $\T$-interpretation $\A_{0}$ with $|\s_{2}^{\A_{0}}|\geq\omega$ (and 
     $|\s^{\A_{0}}|\leq M$) such that there are no $\T$-interpretations $\A$ with $|\s^{\A}|\leq |\s^{\A_{0}}|$ and $|\s_{2}^{\A}|$ 
     finite. But, because $\T$ has the finite model property, there must exist a $\T$-interpretation $\A_{1}$ with finite domains that satisfies $\neg(u=v)$, for $u$ and $v$ 
     variables of sort $\s_{2}$ (since this formula is satisfied by 
     $\A_{0}$): of course, then $|\s_{2}^{\A_{1}}|\geq 2$ and $|\s^{\A_{1}}|>|\s^{\A_{0}}|\geq 1$ (by the definition of $\A_{0}$). Therefore, $|\s^{\A_{1}}|\geq 2$, and we can prove $\T$ is not convex:
     take variables $x_{1}$ through $x_{M+1}$ of sort $\s$, and it is clear that $\dash_{\T}\bigvee_{i=1}^{M}\bigvee_{j=i+1}^{M+1}x_{i}=x_{j}$, but we 
     cannot have $\dash_{\T}x_{i}=x_{j}$ for any pair $1\leq i<j\leq M+1$ since there is a $\T$-
     interpretation ($\A_{1}$) with at least two elements in its domain of sort $\s$.

     So, for $\T$ to be convex, we must have two $\T$-interpretations $\A_{0}$ and $\A_{1}$ with $|\s^{\A_{0}}|>M$ and $|\s_{2}^{\A_{1}}|>M$, but again we reach a contradiction, since then 
     $|\s^{\A_{0}}|>1$ and $|\s_{2}^{\A_{1}}|>1$. Indeed, take variables $x_{1}$ through $x_{M+1}$ of sort $\s$, and $u_{1}$ through $u_{M+1}$ of sort $\s_{2}$, and we have that 
     \[\tdash\bigvee_{i=1}^{M}\bigvee_{j=i+1}^{M+1}x_{i}=x_{j}\vee\bigvee_{p=1}^{M}\bigvee_{q=p+1}^{M+1}u_{p}=u_{q},\]
     but neither $\tdash x_{i}=x_{j}$ nor $\tdash u_{p}=u_{q}$ for pairs $1\leq i<j\leq M+1$ or $1\leq p<q\leq M+1$: in the first case, 
     because there exists a $\T$-interpretation with at least two elements in its domain of sort $\s$ ($\A_{0}$); in the second because 
     there exists a $\T$-interpretation with at least two elements in its domain of sort $\s_{2}$ ($\A_{1}$).
\end{proof}

\section{Proof of \Cref{thm:addsresult}}

 \Cref{technical result on adding a sort} is proven in \cite{BarTolZoh} and only restated here, mostly for context. In what follows, given a $\Sigmatwo$-interpretation $\A$, $\subs{\A}$ is the $\Sigmaone$-interpretation with $\s^{\subs{\A}}=\s^{\A}$; and, given a quantifier-free $\Sigmatwo$-formula $\phi$ and a $\Sigmatwo$-interpretation $\A$, $\subsf{\phi}{\A}$ is the $\Sigmaone$-formula where we replace every atomic subformula $u=v$ of $\phi$ with variables of sort $\s_{2}$ for a $\Sigmaone$-tautology (such as $x=x$, for $x$ of sort $\s$) if $\A$ satisfies $u=v$, and otherwise for a $\Sigmaone$-contradiction (such as $\neg(x=x)$).

\begin{lemma}\label{technical result on adding a sort}
Take a $\Sigmatwo$-interpretation $\A$. It is then true that, for any $\Sigmaone$-formula $\varphi$, $\subs{\A}$ satisfies $\varphi$ if, and only if, $\A$ satisfies $\varphi$.
\end{lemma}

\addsresult*

\begin{proof}
    \begin{enumerate}
        \item If $\T$ has the finite model property, given it is one-sorted, we may use \Cref{OS+FMP=>SF} to prove $\T$ is also stably finite. So we may prove that this, in turn, implies $\Tadds$ is stably finite, and therefore has the finite model property.

        \item Suppose now $\T$ is stably finite, let $\phi$ be a quantifier-free $\Sigmatwo$-theory, and let $\A$ be a $\Tadds$-interpretation that satisfies $\phi$. The facts $\subs{\A}$ satisfies the (quantifier-free) formula $\subsf{\phi}{\A}$, and $\T$ is stably finite, imply there exists a $\T$-interpretation $\C$ that satisfies $\subsf{\phi}{\A}$ with $|\s^{\C}|$ finite and satisfying $|\s^{\C}|\leq |\s^{\subs{\A}}|=|\s^{\A}|$. We then define the $\Tadds$-interpretation $\B$ with $\subs{\B}=\C$ and: $\s_{2}^{\B}=\vars_{\s_{2}}(\phi)^{\A}$ (which is a finite set that satisfies $|\s_{2}^{\B}|\leq|\s_{2}^{\A}|$); and $x^{\B}=x^{\A}$ for every variable $x\in\vars_{\s_{2}}(\phi)$ (and arbitrarily otherwise). We then have that, for every atomic subformula $x=y$ of $\phi$ with $x$ and $y$ of sort $\s_{2}$, $\B$ satisfies $x=y$ iff $\A$ satisfies $x=y$, meaning $\subsf{\phi}{\B}=\subsf{\phi}{\A}$ and, therefore, that $\B$ is a $\Tadds$-interpretation that satisfies $\phi$ with $|\s^{\B}|\leq|\s^{\A}|$, $|\s_{2}^{\B}|\leq|\s_{2}^{\A}|$, and both finite.
    \end{enumerate}
    \begin{enumerate}
        \item Assume now $\Tadds$ has the finite model property, let $\phi$ be a quantifier-free $\Sigmaone$-formula, and $\C$ be a $\T$-interpretation that satisfies $\phi$. Then $\A$ such that $\subs{\A}=\C$ and $|\s_{2}^{\A}|=1$ satisfies $\phi$, meaning there exists a $\T$-interpretation $\B$ that satisfies $\phi$ with both $|\s^{\B}|$ and $|\s_{2}^{\B}|$ finite. Then $\subs{\B}$ is a $\T$-interpretation that satisfies $\phi$, with $|\s^{\subs{\B}}|=|\s^{\B}|$ finite.
        
        \item Assume now $\Tadds$ is stably finite, and thus has the finite model property: as we proved above, this implies $\T$ has the finite model property; but, since $\T$ is one sorted, \Cref{OS+FMP=>SF} tells us that $\T$ must also be stably finite.
    \end{enumerate}
\end{proof}

\section{Proof of \Cref{thm:addfresult}}

\Cref{defining Ts from T} was proven in \cite{BarTolZoh}, and we restate it here as it defines some of the symbols we use.

\begin{lemma}\label{defining Ts from T}
Let $\Sigma_{n}$ be the empty signature with $n$ sorts $S=\{\s_{1}, \ldots, \s_{n}\}$, and $\Sigma^{n}_{s}$ the signature with sorts $S$ and only one function symbol $s$ of arity $\s_{1}\rightarrow\s_{1}$. Given a $\Sigma_{n}$-interpretation $\C$, we define a $\Sigma^{n}_{s}$-interpretation $\plusf{\C}$ by making:
\begin{enumerate}
\item $\s^{\plusf{\C}}=\s^{\C}$ for each $\s\in S$;
\item $s^{\plusf{\C}}(a)=a$ for all $a\in \s_{1}^{\C}$; 
\item and $x^{\plusf{\C}}=x^{\C}$ for every variable $x$.
\end{enumerate}
Reciprocally, given any $\Sigma^{n}_{s}$-interpretation $\A$, we may consider the $\Sigma_{n}$-interpretation $\subf{\A}$ with:
\begin{enumerate}
\item $\s^{\subf{\A}}=\s^{\A}$ for each $\s\in S$;
\item and $x^{\subf{\A}}=x^{\A}$, for every variable $x$.
\end{enumerate}
Finally, given a $\Sigma^{n}_{s}$-formula $\varphi$, we repeatedly replace each occurrence of $s(x)$ in $\varphi$ by $x$ until we obtain a $\Sigma_{n}$-formula
$\subff{\varphi}$. Then, it is true that a $\Sigma^{n}_{s}$-interpretation $\A$ that satisfies $\forall\, x.\:[s(x)=x]$ (where $x$ is of sort $\s_{1}$) also satisfies $\varphi$ iff $\subf{\A}$ satisfies $\subff{\varphi}$; of course, given that for any $\Sigma_{n}$-interpretation $\C$, $\subf{\plusf{\C}}=\C$, $\C$ satisfies a $\Sigma_{n}$-formula $\varphi$ iff $\plusf{\C}$ satisfies $\varphi$ (since $\subff{\varphi}=\varphi$).
\end{lemma}

\addfresult*

\begin{proof}
    \begin{enumerate}
        \item If $\T$ has the finite model property, let us take a quantifier-free $\Sigma^{s}_{n}$-formula $\phi$, and a $\Taddf$-interpretation $\A$ that satisfies $\phi$. We then have that $\subf{\A}$ satisfies $\subff{\phi}$, and so there exists a $\T$-interpretation $\C$, with $|\s^{\A}|$ finite for each $\s\in S$, that satisfies $\subff{\phi}$. Of course, then $\plusf{\A}$ is a $\Taddf$-interpretation, finite in every one of its domains, that satisfies $\phi$.

        \item Suppose $\T$ is stably finite, let $\phi$ be a quantifier-free $\Sigma^{s}_{n}$-formula, and $\A$ a $\Taddf$-interpretation that satisfies $\phi$, meaning $\subf{\A}$ satisfies $\subff{\phi}$. Because $\subff{\phi}$ is quantifier-free, and $\T$ is stably finite, there exists a $\T$-interpretation $\C$ that satisfies $\subff{\phi}$ with $|\s^{\C}|\leq |\s^{\subf{\A}}|=|\s^{\A}|$ finite for each $\s\in S$. And then $\plusf{\C}$ is a $\Taddf$-interpretation that satisfies $\phi$ (since $\subf{\plusf{\C}}=\C$), with $|\s^{\plusf{\C}}|=|\s^{\C}|\leq |\s^{\A}|$ finite for each $\s\in S$.
    \end{enumerate}
    \begin{enumerate}
        \item Suppose $\Taddf$ has the finite model property, let $\phi$ be a quantifier-free $\Sigma_{n}$-formula, and $\C$ a $\T$-interpretation that satisfies $\phi$. Then $\plusf{\C}$ satisfies $\phi$, and since $\Taddf$ has the finite model property there exists a $\Taddf$-interpretation $\A$ that satisfies $\phi$ with $|\s^{\A}|$ finite for each $\s\in S$. Of course, $\subf{\A}$ is then a $\T$-interpretation that satisfies $\phi$ with $|\s^{\subf{\A}}|=|\s^{\A}|$ finite for each $\s\in S$.
        
        \item Assume now $\Taddf$ is stably finite, let $\phi$ be a quantifier-free $\Sigma_{n}$-formula, and $\C$ be a $\T$-interpretation that satisfies $\phi$: $\plusf{\C}$ is then a $\Taddf$-interpretation that satisfies $\phi$, and there must then exist a $\Taddf$-interpretation $\A$, with $|\s^{\A}|$ finite and $|\s^{\A}|\leq |\s^{\plusf{\C}}|$ for each $\s\in S$, that satisfies $\phi$; of course, $\subf{\A}$ is then a $\T$-interpretation that satisfies $\phi$ with $|\s^{\subf{\A}}|$ finite and $|\s^{\subf{\A}}|\leq |\s^{\C}|$ for each $\s\in S$.
    \end{enumerate}
\end{proof}

\section{Proof of \Cref{thm:addncresult}}

We will use some of the symbols found in \Cref{defining Ts from T}. Moreover, given a $\Sigma_{s}^{n}$ formula $\phi$ and a $\Sigma_{s}^{n}$-interpretation $\A$ that satisfies $\phi$, we define the auxiliary formula $\subncf{\phi}{\A}$. To do that, let: $\vars_{\s_{1}}(\phi)=\{z_{1}, \ldots, z_{m}\}$; $M_{i}$ be the maximum of $j$ such that $s^{j}(z_{i})$ appears in $\phi$; and $V=\{y_{i,j} : 1\leq i\leq m, 0\leq j\leq M_{i}+2\}$ be a set of fresh variables of sort $\s_{1}$.
\begin{enumerate}
\item $\dagg{\phi}$ is defined by replacing each atomic subformula $s^{j}(z_{i})=s^{q}(z_{p})$ of $\phi$ by $y_{i,j}=y_{p,q}$;
\item $\A^{\prime}$ is the $\T$-interpretation that differs from $\C$ at most on $V$, where $y^{\A^{\prime}}_{i,j}=(s^{\A})^{j}(z_{i}^{\A})$;
\item $\delta_{V}$ is the arrangement induced on $V$ by making $x$ related to $y$ iff $x^{\A^{\prime}}=y^{\A^{\prime}}$.
\end{enumerate}
We then make $\subncf{\phi}{\A}=\dagg{\phi}\wedge\delta_{V}$.

\addncresult*

\begin{proof} 

\begin{enumerate}
    \item Suppose $\T$ has the finite model property, let $\phi$ be a quantifier-free $\Sigma_{s}^{n}$-formula, and $\A$ be a $\Taddnc$-interpretation that satisfies $\phi$: there is then an interpretation $\A^{\prime}$ differing from $\A$ at most on the value assigned to variables (thus making of $\A^{\prime}$ a $\Taddnc$-interpretation) that satisfies the $\Sigma_{n}$ formula $\subncf{\phi}{\A}$. Then $\subf{\A^{\prime}}$ is a $\T$-interpretation that satisfies $\subncf{\phi}{\A}$ and, by assumption on $\T$, there exists a $\T$-interpretation $\C$ that satisfies $\subncf{\phi}{\A}$ with $\s_{k}^{\C}$ finite for each $1\leq k\leq n$.

    Now, we define an interpretation $\B$. First, we make $\s_{k}^{\B}=\s_{k}^{\C}$ for all $1\leq k\leq n$, so $\B$ has all domains finite. Then, $s^{\B}(y_{i,j}^{\C})=y_{i,j+1}^{\C}$ for each $1\leq i\leq m$ and $0\leq j\leq M_{i}+1$ (notice that that also defines $s^{\B}(y_{i,M_{i}+2}^{\C})$, since $y_{i,M_{i}+2}$ must equal either $y_{i,M_{i}}$ or $y_{i,M_{i}+1}$ in $\C$), and $s^{\B}(a)=a$ for each $a\in\s_{1}^{\C}\setminus\vars_{\s_{1}}(\subncf{\phi}{\A})^{\C}$ (so that $\B$ always satisfies $\psiv$, given $\A$ satisfies that formula and therefore $y_{i,j+2}$ equals either $y_{i,j}$ or $y_{i,j+1}$ in $\A^{\prime}$, and thus in $\C$): this way, $\B$ is then a $\T$-interpretation that satisfies $\psiv$, or in other words, a $\Taddnc$-interpretation. Finally, we make $z_{i}=y_{i,0}^{\C}$, and arbitrary for other variables: this way $\B$ satisfies $\phi$, and we have proven $\Taddnc$ has the finite model property.
    
    \item The proof is the same as the one above if $\T$ is stably finite: let $\phi$ be a quantifier-free $\Sigma_{s}^{n}$-formula, $\A$ be a $\Taddnc$-interpretation that satisfies $\phi$, and $\A^{\prime}$ be a $\Taddnc$-interpretation differing from $\A$ at most on the value assigned to variables  that satisfies the $\Sigma_{n}$ formula $\subncf{\phi}{\A}$. 
    
    $\subf{\A^{\prime}}$ is then a $\T$-interpretation that satisfies $\subncf{\phi}{\A}$ and, by the fact $\T$ is assumed to be stably finite, there exists a $\T$-interpretation $\C$ that satisfies $\subncf{\phi}{\A}$ with all domains finite, and $|\s_{k}^{\C}|\leq |\s_{k}^{\subf{\A^{\prime}}}|=|\s_{k}^{\A}|$ for each $1\leq k\leq n$. We then carefully define an interpretation $\B$ with the same domains as $\C$ that turns out to be a $\Taddnc$-interpretation that satisfies $\phi$, and since $|\s_{k}^{\B}|=|\s_{k}^{\C}|\leq |\s_{k}^{\A}|$ we are done.
\end{enumerate}

Now we deal with the cases where $\Taddnc$ has the desired properties, which is much easier.

\begin{enumerate}
    \item If $\Taddnc$ has the finite model property, let $\phi$ be a $\Sigma_{n}$-formula and $\C$ be a $\T$-interpretation that satisfies $\phi$: $\plusf{\C}$ is then a $\Taddnc$-interpretation that satisfies $\phi$ by \Cref{defining Ts from T}; so there must exist a $\Taddnc$-interpretation $\A$, with all domains finite, that satisfies $\phi$. Again by \Cref{defining Ts from T}, we have $\subf{\A}$ is a $\T$-interpretation that satisfies $\phi$, and its domains are finite, so we are done.
    \item Finally, suppose $\Taddnc$ is stably finite, and let $\phi$ be a $\Sigma_{n}$-formula and $\C$ be a $\T$-interpretation that satisfies $\phi$. By \Cref{defining Ts from T}, $\plusf{\C}$ is a $\Taddnc$-interpretation that satisfies $\phi$, and given $\Taddnc$ is stably finite, there must exist a $\Taddnc$-interpretation $\A$, with all domains finite, that satisfies $\phi$ and $|\s_{k}^{\A}|\leq |\s_{k}^{\plusf{\C}}|=|\s_{k}^{\C}|$ for each $1\leq k\leq n$. Using \Cref{defining Ts from T}, $\subf{\A}$ is a $\T$-interpretation that satisfies $\phi$, its domains are finite, and $|\s_{k}^{\subf{\A}}|=|\s_{k}^{\A}|\leq|\s_{k}^{\C}|$ for each $1\leq k\leq n$, meaning the proof is done.
\end{enumerate}

\end{proof}

\section{Proofs for Theories from \Cref{tab-theories-tab-old}}

For the theories of \Cref{tab-theories-tab-old},
most properties were proven in \cite{BarTolZoh}.
Here we only consider the new properties:
stable finiteness and the finite model property.

\subsection{$\Tgeqn$}

\begin{lemma}
    The $\Sigmaone$-theory $\Tgeqn$ is stably finite, and thus has the finite model property, with respect to its only sort, for every $n\in\mathbb{N}\setminus\{0\}$.
\end{lemma}

\begin{proof}
From \cite{BarTolZoh}, $\Tgeqn$ is strongly finitely witnessable, and from \Cref{SFW=>SF} we get that $\Tgeqn$ is stably finite, and thus has the finite model property by \Cref{SF=>FMP}.
\end{proof}

\subsection{$\Teven$}

\begin{lemma}
    The $\Sigmaone$-theory $\Teven$ is stably finite, and thus has the finite model property, with respect to its only sort.
\end{lemma}

\begin{proof}
From \cite{BarTolZoh}, $\Teven$ is finitely witnessable, and from \Cref{FW=>FMP} we get that $\Teven$ has the finite model property; by \Cref{OS+FMP=>SF}, and the fact $\Teven$ is a $\Sigmaone$-theory, we get it is also stably finite.
\end{proof}

\subsection{$\Tinfty$}

\begin{lemma}
    The $\Sigmaone$-theory $\Tinfty$ does not have the finite model property, and thus is not stably finite, with respect to $\{\s\}$.
\end{lemma}

\begin{proof}
    Obvious, since $\Tinfty$ has models, but none of them are finite.
\end{proof}

\subsection{$\Tninfty$}

\begin{lemma}
    The $\Sigmaone$-theory $\Tninfty$ does not have the finite model property, and thus is not stably finite, with respect to $\{\s\}$.
\end{lemma}

\begin{proof}
    Consider the formula
    $\bigwedge_{1\leq i<j\leq n+1}\neg(x_{i}=x_{j})$.
    it is satisfied by any infinite model of $\Tninfty$, but it cannot be satisfied by any finite $\Tninfty$-interpretation $\A$, since necessarily $|\s^{\A}|=n$.
\end{proof}

\subsection{$\Tleqn$}

\begin{lemma}
    The $\Sigmaone$-theory $\Tleqn$ is stably finite, and thus has the finite model property, with respect to $\{\s\}$.
\end{lemma}

\begin{proof}
    Let $\phi$ be a quantifier-free $\Sigmaone$-formula, and $\A$ a $\Tleqn$-interpretation that satisfies $\phi$: then $\A$ is already a finite $\Tleqn$-interpretation that satisfies $\phi$ with $|\s^{\A}|\leq|\s^{\A}|$.
\end{proof}

\subsection{$\Tmn$}

\begin{lemma}
    The $\Sigmaone$-theory $\Tmn$ is stably finite, and thus has the finite model property, with respect to $\{\s\}$.
\end{lemma}

\begin{proof}
    Let $\phi$ be a quantifier-free $\Sigmaone$-formula, and $\A$ a $\Tmn$-interpretation that satisfies $\phi$: of course $\A$ is already a finite $\Tmn$-interpretation that satisfies $\phi$.
\end{proof}

\subsection{$\Ttwo$}

\begin{lemma}
    The $\Sigmatwo$-theory $\Ttwo$ has the finite model property with respect to $\{\s\}$.
\end{lemma}

\begin{proof}
From \cite{SZRRBT-21}, $\Ttwo$ is finitely witnessable; and from \Cref{FW=>FMP}, we get that it has also the finite model property.

\end{proof}

\begin{lemma}
    The $\Sigmatwo$-theory $\Ttwo$ is not stably finite with respect to $\{\s\}$.
\end{lemma}

\begin{proof}
    From \cite{SZRRBT-21}, $\Ttwo$ is smooth but not strongly finitely witnessable with respect to $\{\s, \s_{2}\}$: since the signature is empty, by \Cref{OS+ES+SM+SF=>SFW}, we conclude it cannot be stably finite.
\end{proof}

\subsection{$\Ttwotwo$}

\begin{lemma}
    The $\Sigmatwo$-theory $\Ttwotwo$ does not have the finite model property, and thus is not stably finite, with respect to $\{\s,\s_{2}\}$.
\end{lemma}

\begin{proof}
    Obvious, given $\Ttwotwo$ is not contradictory, but also has no interpretations $\A$ with $|\s_{2}^{\A}|$ finite.
\end{proof}

\subsection{$\Toddtwo$}

\begin{lemma}
    The $\Sigmatwo$-theory $\Toddtwo$ is stably finite, and thus has the finite model property, with respect to $\{\s,\s_{2}\}$.
\end{lemma}

\begin{proof}
    Take a quantifier-free $\Sigmatwo$-formula $\phi$ and a $\Toddtwo$-interpretation $\A$ that satisfies $\phi$: then, let $n=|\vars_{\s_{2}}(\phi)^{\A}|$, and $N=n$ if $n$ is 
    odd, and $N=n+1$ if $n$ is even. Take a set $A$ with $N-n$ elements disjoint from the domains of $\A$ (notice $A$ may be empty), and define the $\Toddtwo$-interpretation $\B$ such that: $\s^{\B}=\s^{\A}$ (and so 
    $|\s^{\B}|=1$); $\s^{\B}=\vars_{\s_{2}}(\phi)^{\A}\cup A$ (which has necessarily $N$ elements, what makes 
    of $\B$ indeed a $\Toddtwo$-interpretation); $x^{\B}=x^{\A}$ for all variables $x$ in $\phi$; and arbitrarily otherwise. Obviously $\B$ satisfies $\phi$, is finite in both 
    domains and has the property that $|\s^{\B}|\leq |\s^{\A}|$ and $|\s_{2}^{\B}|\leq |\s_{2}^{\A}|$ (notice that, even if $N=n+1$, since
    $|\s_{2}^{\A}|$ must be an odd, or infinite, number we have $|\s_{2}^{\A}|\geq N$).
\end{proof}

\subsection{$\Tonetwo$}

\begin{lemma}
    The $\Sigmatwo$-theory $\Tonetwo$ does not have the finite model property, and thus is not stably finite, with respect to $\{\s,\s_{2}\}$.
\end{lemma}

\begin{proof}
    Obvious, since $\Tonetwo$ has models, and yet has no models $\A$ where $\s_{2}^{\A}$ is finite.
\end{proof}

\subsection{$\TsM$}

\begin{lemma}
    The $\Sigmas$-theory $\TsM$ is stably finite, and thus has the finite model property, with respect to $\{\s\}$.
\end{lemma}

\begin{proof}
    $\TsM$ is finitely witnessable with respect to $\{\s\}$ (proven in \cite{BarTolZoh}), and thus by \Cref{FW=>FMP} $\TsM$ has the finite model property; and, since $\Sigmas$ is a one-sorted signature, by \Cref{OS+FMP=>SF} we obtain $\TsM$ is also stably finite.
\end{proof}

\subsection{$\TSM$}

\begin{lemma}
    The $\Sigmas$-theory $\Tneqodd$ is stably finite, and thus has the finite model property, with respect to $\{\s\}$.
\end{lemma}

\begin{proof}
    Because of \Cref{OS+FMP=>SF}, we just need to prove $\TSM$ has the finite model property; and because of \Cref{FW=>FMP}, we just use the fact that $\TSM$ is finitely witnessable, as proved in \cite{BarTolZoh}.
\end{proof}

\subsection{$\Tneqtwoinfty$}

\begin{lemma}
    The $\Sigmas$-theory $\Tneqtwoinfty$ does not have the finite model property, and thus is not stably finite, with respect to $\{\s\}$.
\end{lemma}

\begin{proof}
    There is a $\Tneqtwoinfty$-interpretation $\A$ with $|\s^{\A}|=\omega$ (write, for simplicity, $\s^{\A}=\{a_{i} \mid i\in\mathbb{N}\}$), and $s^{\A}(a_{i})=a_{i+1}$ for all $i\in\mathbb{N}$ (and any assignment of the variables) that satisfies $\neg(s(x)=x)$. However, the only finite $\Tneqtwoinfty$-interpretations satisfy $\Forall{x}(s(x)=x)$, and thus cannot satisfy $\neg(s(x)=x)$.
\end{proof}

\subsection{$\Tneqodd$}

\begin{lemma}
    The $\Sigmas$-theory $\Tneqodd$ is stably finite, and thus has the finite model property, with respect to $\{\s\}$.
\end{lemma}

\begin{proof}
    Because the signature is one-sorted, according to \Cref{OS+FMP=>SF} we only need to show the theory has the finite model property. So take a 
    quantifier-free $\Sigmas$-formula $\phi$ and a $\Tneqodd$-interpretation $\A$ that satisfies $\phi$:
    we can assume $\A$ has at least $3$ elements since, otherwise, $\phi$ would already be known to be satisfied by a finite $\Tneqodd$-interpretation. Take as well the variables $x_{1}$ through $x_{n}$ in $\phi$, 
    and the maximum $M_{i}$, for each $i$,
    of the values $j$ such that $s^{j}(x_{i})$ appears in $\phi$, and define 
    \[m=|\{s^{j}(x_{i}) : 1\leq i\leq n, 0\leq j\leq M_{i}+1\}^{\A}|,\]
    and there are two cases to consider: if $m$ is odd, make $M=m+2$, and if $m$ is even, make $M=m+1$; take as well a set $A$ with $M-m$ elements disjoint from $\s^{\A}$. We then define an interpretation $\B$ by making: 
    \[\s^{\B}=\{s^{j}(x_{i}) : 1\leq i\leq n, 0\leq j\leq M_{i}+1\}^{\A}\cup A;\]
    $s^{\B}(a)=s^{\A}(a)$ for all $a\in \{s^{j}(x_{i}) : 1\leq i\leq n, 0\leq j\leq M_{i}\}^{\A}$, and $s^{\B}(a)$ any value in $A$ different from 
    $a$ (of which there is one, since $A$ has at least $2$ elements) for each 
    $a\in A$ or $a=(s^{\A})^{M_{i}+1}(x_{i}^{\A})$ not in $\{s^{j}(x_{i}) : 1\leq i\leq n, 0\leq j\leq M_{i}\}^{\A}$; $x^{\B}=x^{\A}$ for each variable $x$ of $\phi$, and arbitrarily otherwise. $\B$ has then an odd number of elements (thus 
    finite) greater or equal than $3$, and satisfies $\neg(s(x)=x)$ for all $x$, meaning it is a $\Tneqodd$-interpretation; furthermore, it satisfies $\phi$, and so we are done.
\end{proof}

\subsection{$\Tneqoneinfty$}

\begin{lemma}
    The $\Sigmas$-theory $\Tneqoneinfty$ does not have the finite model property, and thus is not stably finite, with respect to $\{\s\}$.
\end{lemma}

\begin{proof}
    Take a variable $x$ and consider the quantifier-free formula $\neg(s(x)=x)$, satisfied by the infinite $\Tneqoneinfty$-interpretations: because the only 
    finite $\Tneqoneinfty$-interpretation has only one element, it must satisfy $s(x)=x$ instead, and so this theory cannot have the finite model property.
\end{proof}

\section{Proofs for Theories from \Cref{tab-theories-sigma-22}}

\subsection{$\Tupinfty$}

\begin{lemma}\label{Tupinfty is SI}
    The $\Sigmatwo$-theory $\Tupinfty$ is stably infinite with respect to $\{\s, \s_{2}\}$.
\end{lemma}

\begin{proof}
    Let $\phi$ be a quantifier-free $\Sigmatwo$-formula, and $\A$ a $\Tupinfty$-interpretation that satisfies $\phi$. Given two infinite sets $A$ and $B$, disjoint from each other and $\s^{\A}$ and $\s_{2}^{\A}$, we define a $\Tupinfty$-interpretation $\B$ by making: $\s^{\B}=\s^{\A}\cup A$; $\s_{2}^{\B}=\s_{2}^{\A}\cup B$; $x^{\B}=x^{\A}$ for every variable $x$ of sort $\s$; and $u^{\B}=u^{\A}$ for every variable $u$ of sort $\s_{2}$. Then it is true that $\B$ satisfies $\phi$, and both $\s^{\B}$ and $\s_{2}^{\B}$ are infinite.
\end{proof}

\begin{lemma}\label{Tupinfty is not SM}
    The $\Sigmatwo$-theory $\Tupinfty$ is not smooth with respect to $\{\s, \s_{2}\}$.
\end{lemma}

\begin{proof}
    Let $x$ be a variable of sort $\s$, and it is clear that $x=x$ is satisfied by a $\Tupinfty$-interpretation $\A$ with $|\s^{\A}|=|\s_{2}^{\A}|=2$. But there are no $\Tupinfty$-interpretations $\B$ with $|\s^{\B}|=3$ and $|\s_{2}^{\B}|=2$, so $\Tupinfty$ cannot be smooth.
\end{proof}

\begin{lemma}\label{Tupinfty is not SF}
    The $\Sigmatwo$-theory $\Tupinfty$ is not stably finite with respect to $\{\s, \s_{2}\}$.
\end{lemma}

\begin{proof}
    Let $x$ be a variable of sort $\s$, and $\A$ be a $\Tupinfty$-interpretation with $|\s^{\A}|=1$ and $|\s_{2}^{\A}|=\omega$: it is clear that $\A$ satisfies $x=x$. But there are no $\Tupinfty$-interpretations $\B$ with $|\s^{\B}|=1$ and $|\s_{2}^{\B}|$ finite, so $\Tupinfty$ cannot be stably finite.
\end{proof}

\begin{lemma}\label{Tupinfty is FW}
    The $\Sigmatwo$-theory $\Tupinfty$ is finitely witnessable with respect to $\{\s, \s_{2}\}$.
\end{lemma}

\begin{proof}
    Given a quantifier-free $\Sigmatwo$-formula $\phi$, consider the witness
    \[\wit(\phi)=\phi\wedge\bigwedge_{i=1}^{m}x_{i}=x_{i}\wedge\bigwedge_{i=1}^{m}u_{i}=u_{i},\]
    where
    \[m=\max\{|\vars_{\s}(\phi)^{\A}|,|\vars_{\s_{2}}(\phi)^{\A}|, 2\},\]
    $x_{1}$ through $x_{m}$ are fresh variables of sort $\s$, and $u_{1}$ through $u_{m}$ are fresh variables of sort $\s_{2}$. Given $\bigwedge_{i=1}^{m}x_{i}=x_{i}\wedge\bigwedge_{i=1}^{m}u_{i}=u_{i}$ is a tautology, and $\phi$ possesses none of the variables in $\overarrow{x}=\vars(\wit(\phi))\setminus\vars(\phi)$, we have that $\phi$ and $\Exists{\overarrow{x}}\wit(\phi)$ are $\Tupinfty$-equivalent.

    Now, assume $\A$ is a $\Tupinfty$-interpretation that satisfies $\wit(\phi)$. Take sets $A$ and $B$, disjoint from each other and $\s^{\A}$ and $\s_{2}^{\A}$, with $m-|\vars_{\s}(\phi)^{\A}|$ and $m-|\vars_{\s_{2}}(\phi)^{\A}|$ elements, respectively. We define the $\Tupinfty$-interpretation $\B$ by making: $\s^{\B}=\vars_{\s}(\phi)^{\A}\cup A$; $\s_{2}^{\B}=\vars_{\s_{2}}(\phi)^{\A}\cup B$; $x^{\B}=x^{\A}$ for $x\in\vars_{\s}(\phi)$, and $x_{i}\in\{x_{i} : 1\leq i\leq m\}\mapsto x_{i}^{\B}\in\s^{\B}$ a surjection; and $u^{\B}=u^{\A}$ for $u\in\vars_{\s_{2}}(\phi)$, and $u_{i}\in\{u_{i} : 1\leq i\leq m\}\mapsto u_{i}^{\B}\in\s_{2}^{\B}$ also a surjection (and arbitrary for other variables). This way, $\B$ satisfies $\phi$, and in addition has $\s^{\B}=\vars_{\s}(\wit(\phi))^{\B}$ and $\s_{2}^{\B}=\vars_{\s_{2}}(\wit(\phi))^{\B}$.
\end{proof}

\begin{lemma}\label{Tupinfty has FM}
    The $\Sigmatwo$-theory $\Tupinfty$ has the finite model property with respect to $\{\s, \s_{2}\}$.
\end{lemma}

\begin{proof}
Follows from \Cref{FW=>FMP} and \Cref{Tupinfty is FW}.
\end{proof}

\begin{lemma}\label{Tupinfty is not SW}
    The $\Sigmatwo$-theory $\Tupinfty$ is not strongly finitely witnessable with respect to $\{\s, \s_{2}\}$.
\end{lemma}

\begin{proof}
Follows from \Cref{SFW=>SF} and \Cref{Tupinfty is not SF}.
\end{proof}

\begin{lemma}\label{Tupinfty is CV}
    The $\Sigmatwo$-theory $\Tupinfty$ is convex with respect to $\{\s, \s_{2}\}$.
\end{lemma}

\begin{proof}
    Since $\Tupinfty$ is stably infinite, and $\Sigmatwo$ is an empty signature, we simply apply \Cref{SI empty theories are convex}.
\end{proof}

\subsection{$\Tmninfty$}

\begin{lemma}
    The $\Sigmatwo$-theory $\Tmninfty$ is not stably infinite, and thus not smooth, with respect to $\{\s, \s_{2}\}$ for any $m,n\in\mathbb{N}$.
\end{lemma}

\begin{proof}
    Obvious, since $\Tmninfty$ is not contradictory, but it does not have any interpretations $\A$ with $|\s^{\A}|$ infinite.
\end{proof}

\begin{lemma}\label{Tmninfty is not SF}
    The $\Sigmatwo$-theory $\Tmninfty$ is not stably finite with respect to $\{\s, \s_{2}\}$ for any $m,n\in\mathbb{N}$ such that $m\neq n$.
\end{lemma}

\begin{proof}
    Take a variable $x$ of sort $\s$, and we know $x=x$ is satisfied by any $\Tmninfty$-interpretation $\A$ with $|\s^{\A}|=\min\{m,n\}$ (and, necessarily, $|\s_{2}^{\A}|\geq\omega$): but, since there are no $\Tmninfty$-interpretations $\B$ with $|\s^{\B}|\leq \min\{m,n\}$ and $|\s_{2}^{\B}|$ finite, $\Tmninfty$ cannot be stably finite.
\end{proof}

\begin{lemma}\label{Tmninfty is FW}
    The $\Sigmatwo$-theory $\Tmninfty$ is finitely witnessable with respect to $\{\s, \s_{2}\}$  for any $m,n\in\mathbb{N}$.
\end{lemma}

\begin{proof}
    Let $\phi$ be a quantifier-free $\Sigmatwo$-formula, and consider the witness
    \[\wit(\phi)=\phi\wedge\bigwedge_{i=1}^{\max\{m,n\}}x_{i}=x_{i},\]
    where $x_{1}$ through $x_{\max\{m,n\}}$ are fresh variables of sort $\s$; this set of variables will be denoted by $\overarrow{x}$, and it equals $\vars(\wit(\phi))\setminus\vars(\phi)$. Suppose first that the $\Tmninfty$-interpretation $\A$ satisfies $\phi$: because $\bigwedge_{i=1}^{\max\{m,n\}}x_{i}=x_{i}$ is a tautology, $\A$ also satisfies that formula, and so $\A$ satisfies $\wit(\phi)$ and thus $\Exists{\overarrow{x}}\wit(\phi)$. The reciprocal is also true because $\phi$ has none of the variables in $\overarrow{x}$, and so $\phi$ and $\Exists{\overarrow{x}}\wit(\phi)$ are $\Tmninfty$-equivalent.

    Now, assume $\A$ is a $\Tmninfty$-interpretation that satisfies $\wit(\phi)$; take a set $A$ with $\max\{m,n\}-|\s^{\A}|$ elements, and disjoint from $\s^{\A}$ and $\s_{2}^{\A}$. We define a $\Tmninfty$-interpretation $\B$ by making: $\s^{\B}=\s^{\A}\cup A$; $\s_{2}^{\B}=\vars_{\s_{2}}(\phi)^{\A}$; $x^{\B}=x^{\A}$ and $u^{\B}=u^{\A}$ for all variables $x$ and $u$, of sorts $\s$ and $\s_{2}$ respectively, in $\phi$; $x_{i}\mapsto x_{i}^{\B}$ any surjective map, for $x_{i}$ in $\overarrow{x}$ (and arbitrarily otherwise). Then it is true that $\B$ satisfies $\wit(\phi)$ (since it satisfies $\phi$), and $\s^{\B}=\vars_{\s}(\wit(\phi))^{\B}$ and $\s_{2}^{\B}=\vars_{\s_{2}}(\wit(\phi))^{\B}$.
\end{proof}

\begin{lemma}
    The $\Sigmatwo$-theory $\Tmninfty$ has the finite model property with respect to $\{\s, \s_{2}\}$ for any $m,n\in\mathbb{N}$.
\end{lemma}

\begin{proof}
Follows from \Cref{FW=>FMP} and \Cref{Tmninfty is FW}.
\end{proof}

\begin{lemma}
    The $\Sigmatwo$-theory $\Tmninfty$ is not strongly finitely witnessable with respect to $\{\s, \s_{2}\}$ for any $m,n\in\mathbb{N}$ such that $m\neq n$.
\end{lemma}

\begin{proof}
Follows from \Cref{SFW=>SF} and \Cref{Tmninfty is not SF}.
\end{proof}

\begin{lemma}
    The $\Sigmatwo$-theory $\Tmninfty$ is not convex with respect to $\{\s, \s_{2}\}$ if $\max\{m,n\}>1$.
\end{lemma}

\begin{proof}
    Let $M=\max\{m,n\}>1$ and take $M+1$ variables $x_{1}$ trough $x_{M+1}$, of sort $\s$: it is then true that $\dash_{\Tmninfty}\bigvee_{i=1}^{M}\bigvee_{j=i+1}^{M+1}x_{i}=x_{j}$, by the pigeonhole principle, but we cannot get $\dash_{\Tmninfty}x_{i}=x_{j}$ for any $1\leq i<j\leq M+1$ since there is a $\Tmninfty$-interpretation $\A$ with $|\s^{\A}|=M>1$.
\end{proof}

\subsection{$\Tsupinfty$}

\begin{lemma}
    The $\Sigmastwo$-theory $\Tsupinfty$ is not stably infinite, and thus not smooth, with respect to $\{\s, \s_{2}\}$.
\end{lemma}

\begin{proof}
    $s(x)=x$ is satisfied by the $\Tsupinfty$-interpretation $\A$ with $|\s^{\A}|=1$ and $|\s_{2}^{\A}|=\omega$, but all $\Tsupinfty$-interpretations $\B$ with $|\s^{\A}|\geq 2$ satisfy instead $\Forall{x}\neg(s(x)=x)$.
\end{proof}

\begin{lemma}\label{Tsupinfty is FW}
    The $\Sigmastwo$-theory $\Tsupinfty$ is finitely witnessable with respect to $\{\s, \s_{2}\}$.
\end{lemma}

\begin{proof}
    The proof is similar to the one of \Cref{Tupinfty is FW}, some care being necessary to give witnesses to the function $s$; see the proof that $\Tneqodd$ is finitely witnessable in \cite{BarTolZoh} for how this can be done.
\end{proof}

\begin{lemma}\label{Tsupinfty is not SF}
    The $\Sigmastwo$-theory $\Tsupinfty$ is not stably finite with respect to $\{\s, \s_{2}\}$.
\end{lemma}

\begin{proof}
    Obvious, since $\Tsupinfty$ has an interpretation $\A$ with $|\s^{\A}|=1$ and $|\s_{2}^{\A}|=\omega$, while there are no $\Tsupinfty$-interpretations $\B$ with $|\s^{\B}|=1$ and $\s_{2}^{\B}$ finite.
\end{proof}

\begin{lemma}
    The $\Sigmastwo$-theory $\Tsupinfty$ has the finite model property with respect to $\{\s, \s_{2}\}$.
\end{lemma}

\begin{proof}
    Follows from \Cref{Tsupinfty is FW} and \Cref{FW=>FMP}.
\end{proof}

\begin{lemma}
    The $\Sigmastwo$-theory $\Tsupinfty$ is not strongly finitely witnessable with respect to $\{\s, \s_{2}\}$.
\end{lemma}

\begin{proof}
Follows from \Cref{SFW=>SF} and \Cref{Tsupinfty is not SF}.
\end{proof}

\begin{lemma}\label{Tneg is CV}
    Let $\T$ be a stably infinite  $\Sigma_{n}$-theory (the sorts of $\Sigma_{n}$ are $S=\{\s_{1}, \ldots, \s_{n}\}$) with no interpretations $\A$ such that $|\s_{1}^{\A}|=1$: $\T$ is necessarily convex because of \Cref{SI empty theories are convex}.
    
    Then the $\Sigma^{n}_{s}$-theory $\Tneg$, where $\Sigma^{n}_{s}$ is the signature with sorts $S$ and one function symbol $s:\s_{1}\rightarrow\s_{1}$, and $\Tneg$ is axiomatized by $\textit{Ax}(\T)\cup\{\Forall{x}\neg(s(x)=x)\}$, is convex with respect to its only sort.
\end{lemma}

\begin{proof}
    Given a $\Sigma_{s}^{n}$-cube $\phi$, we change the name of every variable $x$ of sort $\s_{1}$ in $\phi$ to $x_{0}$, and replace each occurrence of 
    $s^{i+1}(x)$ in $\phi$ by a new variable $x_{i+1}$, adding the literal $\neg(x_{i+1}=x_{i})$ to $\phi$: the result 
    is the $\Sigma_{n}$-cube $\phi^{\prime}$. 
    
    It is clear that, if $\phi$ is $\Tneg$-satisfiable, $\phi^{\prime}$ is $\T$-satisfiable: indeed, if $\A$ satisfies 
    $\phi$, take the $\Sigma_{n}$-interpretation $\subf{\A}$ with $\s_{i}^{\subf{\A}}=\s_{i}^{\A}$ for each $1\leq i\leq n$; by changing the 
    value of $x_{i}^{\subf{\A}}$ (for $x$ of sort $\s_{1}$) so that it equals $(s^{\A})^{i}(x^{\A})$, the resulting 
    $\T$-interpretation $\C$ then satisfies $\phi^{\prime}$. 

    Conversely, if the $\T$-interpretation $\C$ satisfies $\phi^{\prime}$, we define a $\Tneg$-interpretation $\A$ by making: 
    $\s_{i}^{\A}=\s_{i}^{\C}$ for each $1\leq i\leq n$; $s^{\A}(x_{i}^{\C})=x_{i+1}^{\C}$ if $x_{i+1}$ occurs in $\phi^{\prime}$, and otherwise 
    any value different from $x_{i}^{\C}$ 
    itself; and $x^{\A}=x_{0}^{\C}$ for any variable $x$ in $\phi$, and $x^{\A}=x^{\C}$ otherwise. It is then easy to show $\A$ satisfies $\phi$.

    So assume that $\dash_{\Tneg}\phi\rightarrow\bigvee_{j=1}^{m}x_{j}=y_{j}$: then $\dash_{\T}\phi^{\prime}\rightarrow\bigvee_{j=1}^{m}x_{j}=y_{j}$, since otherwise 
    there exists a $\T$-interpretation $\C$ that satisfies $\phi^{\prime}$ but not $\bigvee_{j=1}^{m}x_{j}=y_{j}$, meaning 
    there is a $\Tneg$-interpretation that satisfies $\phi$ but not $\bigvee_{j=1}^{m}x_{j}=y_{j}$ and thus leading to a contradiction. Given $\T$ is convex, 
    $\dash_{\T}\phi^{\prime}\rightarrow x_{j}=y_{j}$ for some $1\leq j\leq m$, and hence 
    $\phi^{\prime}\wedge\neg(x_{j}=y_{j})$ is not $\T$-satisfiable. It follows that $\phi\wedge\neg(x_{j}=y_{j})$ is not $\Tneg$-satisfiable, and so 
    $\dash_{\Tneg}\phi\rightarrow x_{j}=y_{j}$, showing this theory is in fact convex.
\end{proof}

\begin{lemma}\label{Tsupinfty is CV}
    The $\Sigmastwo$-theory $\Tsupinfty$ is convex with respect to $\{\s, \s_{2}\}$.
\end{lemma}

\begin{proof}


Suppose that $\phi$ is a $\Sigmastwo$-formula that is also a conjunction of literals, and that $\dash_{\Tsupinfty}\phi\rightarrow\bigvee_{i=1}^{n}x_{i}=y_{i}$. If $\phi$ is only true in
$\Tsupinfty$-interpretations $\A$ where $|\s^{\A}|=1$ we are done: indeed, in that case we may assume none of the pairs $(x_{i}, y_{i})$ are of sort $\s$ (otherwise we would simply have 
$\dash_{\Tsupinfty}\phi\rightarrow x_{i}=y_{i}$), take the conjunction $\phi^{\prime}$ of literals in $\phi$ with variables of sort $\s_{2}$ and it is true that 
$\dash_{\Tinfty}\phi^{\prime}\rightarrow\bigvee_{i=1}^{n}x_{i}=y_{i}$, if one sees $\Tinfty$ as a theory with the sort $\s_{2}$; since $\Tinfty$ is convex, it follows that 
$\dash_{\Tinfty}\phi^{\prime}\rightarrow x_{i}=y_{i}$ for some $1\leq i\leq n$, and thus $\dash_{\Tsupinfty}\phi\rightarrow x_{i}=y_{i}$.

So we can assume that $\phi$ is true in some $\Tsupinfty$-interpretation $\A_{0}$ with $|\s^{\A_{0}}|\geq 2$: by \Cref{LowenheimSkolemDownwards} and the fact the theory $\T$, whose interpretations are all $\Tsupinfty$-interpretations $\A$ with $|\s^{\A}|\geq 2$, axiomatized by 
\[\{\big((\psi^{\s}_{k+2}\wedge\psi^{\s_{2}}_{\geq k+2})\vee\bigvee_{i=2}^{k+2}(\psi^{\s}_{=i}\wedge\psi^{\s_{2}}_{=i})\big)\wedge \Forall{x}\neg(s(x)=x) : k\in\mathbb{N}\},\]
is stably infinite, we may assume that $|\s^{\A_{0}}|=|\s_{2}^{\A_{0}}|=\omega$. Now, given that $\dash_{\Tsupinfty}\phi\rightarrow\bigvee_{i=1}^{n}x_{i}=y_{i}$, we have that 
$\dash_{\T}\phi\rightarrow\bigvee_{i=1}^{n}x_{i}=y_{i}$, and since the theory $\T$ is convex by \Cref{Tneg is CV}, we have $\dash_{\T}\phi\rightarrow x_{i}=y_{i}$ for some $1\leq i\leq n$, and of course every 
$\Tsupinfty$-interpretation $\A$ with $|\s^{\A}|\geq 2$ then satisfies $\phi\rightarrow x_{i}=y_{i}$.

Suppose, however, there is a $\Tsupinfty$-interpretation $\A_{1}$ with $|\s^{\A_{1}}|=1$ that does not satisfy $\phi\rightarrow x_{i}=y_{i}$, and thus 
satisfies $\phi$ but not $x_{i}=y_{i}$ (of course, for that to happen, $x_{i}$ and $y_{i}$ must be variables of sort $\s_{2}$): again by \Cref{LowenheimSkolemDownwards}, we
may assume $|\s_{2}^{\A_{1}}|=\omega$. We then define an interpretation $\B$ such that: $\s^{\B}=\s^{\A_{0}}$; $\s_{2}^{\B}=\s_{2}^{\A_{1}}$ (so $|\s^{\B}|=|\s_{2}^{\B}|=\omega$); 
$s^{\B}=s^{\A_{0}}$ (so $\B$ satisfies $\Forall{x}\neg(s(x)=x)$, and is thus a $\Tsupinfty$-interpretation); 
$x^{\B}=x^{\A_{0}}$ for all variables $x$ of sort $\s$; and $u^{\B}=u^{\A_{1}}$ for all variables $u$ of sort $\s_{2}$. Then $\B$ satisfies $\phi$: any literal with variables of sort $\s$ in $\phi$ is satisfied by 
$\A_{0}$ and, by construction, by $\B$; and any literals with variables of sort $\s_{2}$ in $\phi$ is satisfied by $\A_{1}$ and, again by construction, by $\B$; and since $\phi$ is a conjunction of literals, it is then 
satisfied by $\B$. But $\B$ does not satisfy $x_{i}=y_{i}$ while still being a $\T$-interpretation, what is a contradiction.
\end{proof}

\subsection{$\Ttwocube$}

\begin{lemma}
    The $\Sigma_{3}$-theory $\Ttwocube$ is not stably infinite, and thus not smooth, with respect to $\{\s, \s_{2}, \s_{3}\}$.
\end{lemma}

\begin{proof}
    Obvious, since it is not contradictory, but has no models $\A$ with $|\s_{3}^{\A}|$ infinite.
\end{proof}

\begin{lemma}
    The $\Sigma_{3}$-theory $\Ttwocube$ is not stably finite with respect to $\{\s, \s_{2}, \s_{3}\}$.
\end{lemma}

\begin{proof}
    Consider any quantifier-free $\Sigma_{3}$-formula $\phi$ that is satisfied by the $\Ttwocube$-interpretation $\A$ with $|\s^{\A}|=2$, $|\s_{2}^{\A}|=\omega$ and $|\s_{3}^{\A}|=1$: there are no 
    $\Ttwocube$-interpretations $\B$ such $|\s^{\B}|\leq |\s^{\A}|$ and $|\s_{2}^{\B}|\leq |\s_{2}^{\A}|$ and all two of these are finite, so $\Ttwocube$ cannot be stably finite.
\end{proof}

\begin{lemma}\label{Ttwocube is FW}
    The $\Sigma_{3}$-theory $\Ttwocube$ is finitely witnessable with respect to $\{\s, \s_{2}, \s_{3}\}$.
\end{lemma}

\begin{proof}
    Take as witness, given a quantifier-free $\Sigma_{3}$-formula $\phi$, 
    \[\wit(\phi)=\phi\wedge\delta,\quad\text{where}\quad \delta=\bigwedge_{i=1}^{3}x_{i}=x_{i}\wedge\bigwedge_{j=1}^{3}u_{j}=u_{j}\wedge(t=t),\]
    $x_{i}$ are fresh variables of sort $\s$, $u_{j}$ of sort $\s_{2}$, and $t$ of sort $\s_{3}$. Of course, if $\overarrow{x}=\vars(\wit(\phi))\setminus\vars(\phi)$, we have that $\phi$ and $\Exists{\overarrow{x}}\wit(\phi)=\phi\wedge\Exists{\overarrow{x}}\delta$ are $\Ttwocube$-equivalent since $\delta$ is a tautology.

    Now, suppose $\A$ is a $\Ttwocube$-interpretation that satisfies $\wit(\phi)$: take 
    $m=|\vars_{\s}(\phi)^{\A}|$, $n=|\vars_{\s_{2}}(\phi)^{\A}|$, $M=\max\{0, 3-m\}$, $N=\max\{0, 3-n\}$, and sets 
    $A$ and $B$ with cardinalities, respectively, $M$ and $N$, disjoint from each other and $\s^{\A}$. We define an interpretation $\B$ such that: $\s^{\B}=\vars_{\s}(\phi)^{\A}\cup A$, 
    $\s_{2}^{\B}=\vars_{\s_{2}}(\phi)^{\A}\cup B$, and 
    $\s_{3}^{\B}=\s_{3}^{\A}$; $x^{\B}=x^{\A}$ for any variable $x$ in $\vars(\phi)$; $x_{i}\in\{x_{1}, x_{2}, x_{3}\}\mapsto x_{i}^{\B}\in A$ a surjection, and $y_{j}\in\{y_{1}, y_{2}, y_{3}\}\mapsto y_{j}^{\B}\in B$ also a surjection; and arbitrarily for other variables.
    Then we have that $\B$ is a $\Ttwocube$-interpretation that satisfies $\wit(\phi)$, since it satisfies $\phi$, with $\t^{\B}=\vars_{\t}(\wit(\phi))^{\B}$ for each $\t\in\{\s, \s_{2}, \s_{3}\}$.
\end{proof}

\begin{lemma}\label{FMP of Ttwocube}
    The $\Sigma_{3}$-theory $\Ttwocube$ has the finite model property with respect to $\{\s, \s_{2}, \s_{3}\}$.
\end{lemma}

\begin{proof}
Follows from \Cref{FW=>FMP} and \Cref{Ttwocube is FW}.
\end{proof}

\begin{lemma}
    The $\Sigma_{3}$-theory $\Ttwocube$ is not strongly finitely witnessable with respect to $\{\s, \s_{2}, \s_{3}\}$.
\end{lemma}

\begin{proof}
    Suppose we have a strong witness $\wit$: for any quantifier-free formula $\phi$ that is satisfied by the $\Ttwocube$-interpretation $\A$ 
    with $|\s^{\A}|=2$ and $|\s_{2}^{\A}|=\omega$, there is a $\Ttwocube$-interpretation
    $\A^{\prime}$ differing from $\A$ at most on the value assigned to $\overarrow{x}=\vars(\wit(\phi))\setminus\vars(\phi)$ that satisfies $\wit(\phi)$. Of course, if 
    $V=\vars(\wit(\phi))$, and 
    $\delta_{V}$ is the arrangement such that $x$ is related to $y$ iff $x^{\A^{\prime}}=y^{\A^{\prime}}$, then $\A^{\prime}$ satisfies $\wit(\phi)\wedge\delta_{V}$. Of course, since $\wit$ is a strong witness, there is a $\T$-interpretation $\B$ that satisfies
    $\wit(\phi)\wedge\delta_{V}$ with $\s^{\B}=\vars_{\s}(\wit(\phi)\wedge\delta_{V})^{\B}$ and $\s_{2}^{\B}=\vars_{\s_{2}}(\wit(\phi)\wedge\delta_{V})^{\B}$ (and also $\s_{3}^{\B}=\vars_{\s_{3}}(\wit(\phi)\wedge\delta_{V})^{\B}$, but that is irrelevant right now): 
    but if $\vars_{\s}(\wit(\phi)\wedge\delta_{V})^{\B}$ has only one element, we have a contradiction since no $\Ttwocube$-interpretation has only one element in its domain of sort $\s$; and if 
    $\vars_{\s}(\wit(\phi)\wedge\delta_{V})^{\B}$ has two elements, $\B$ is a $\Ttwocube$-interpretation with 
    $|\s^{\B}|=2$ and $|\s_{2}^{\B}|$ finite, which is again a contradiction, and we are done.
\end{proof}

\begin{lemma}
    The $\Sigma_{3}$-theory $\Ttwocube$ is convex with respect to $\{\s, \s_{2}, \s_{3}\}$.
\end{lemma}

\begin{proof}
    We can safely ignore variables of sort $\s_{3}$, since all $\Ttwocube$-interpretations have domains of sort $\s_{3}$ with cardinality $1$: what is left is essentially $\Ttwo$. Using this theory is convex, one can with some effort prove the same holds for $\Ttwocube$.
\end{proof}

\section{Proofs for Theories in \Cref{tab-theories-sigma-s}}

First we prove the following lemma:

\begin{restatable}{lemma}{bblemma}
\label{propertiesofbb}
The following holds for $\bb$:
    \begin{enumerate}
    \item $\bb$ is increasing.
        \item For all $n\in\mathbb{N}\setminus\{0,1\}$, $\bb(n)>n$.
        \item For any computable function $f:\mathbb{N}\rightarrow\mathbb{N}$, there exists $N\in\mathbb{N}$ such that $\bb(n)>f(n)$ for all $n\geq N$.
    \end{enumerate}
\end{restatable}

\begin{proof}
    \begin{enumerate}
    \item Trivial, as one can take a $n+1$-states Turing machine whose first $n$ states emulate the Turing machine that prints $\bb(n)$ $1$'s, and make the new state print a single new $1$.
        \item Trivial, from the last item and the fact that $\bb(2)>2$.
        \item See \cite{Rado}.
    \end{enumerate}
\end{proof}

\subsection{$\TBB$}

\begin{lemma}
    The $\Sigmaone$-theory $\TBB$ is stably infinite with respect to its only sort.
\end{lemma}

\begin{proof}
    It is a theory over an empty signature, and there is a $\TBB$-interpretation $\A$ with $|\s^{\A}|=\omega$, to which any finite $\TBB$-interpretation can be extended.
\end{proof}

\begin{lemma}
    The $\Sigmaone$-theory $\TBB$ is not smooth with respect to its only sort.
\end{lemma}

\begin{proof}
    There is a $\TBB$-interpretation $\A$ with $|\s^{\A}|=4$, but no $\TBBtwo$-interpretations $\B$ with $|\s^{\B}|=5$.
\end{proof}

\begin{lemma}\label{Arbitrarymodels=>FMP}
    If $\T$ is a $\Sigmaone$-theory such that, for every $n\in\mathbb{N}$, there exists a $\T$-interpretation $\A$ with $n<|\s^{\A}|<\omega$, then $\T$ has the finite model property with respect to $\s$.
\end{lemma}

\begin{proof}
    Let $\phi$ be a quantifier-free $\Sigmaone$-formula that is satisfied by some $\T$-interpretation $\A$, and let $n$ be $|\vars(\phi)^{\A}|$: by hypothesis, there exists a $\T$-interpretation $\B$ with $n<|\s^{\B}|<\omega$. Since $n<|\s^{\B}|$, there exists an injective function $f:\vars(\phi)^{\A}\rightarrow\s^{\B}$; so we define an interpretation $\Bp$, different from $\B$ only on $\vars(\phi)$, where $x^{\Bp}=f(x^{\A})$ (for each $x\in\vars(\phi)$).

    Now, every atomic subformula of $\phi$ is of the form $x=y$: and $\A$ satisfies $x=y$ iff $x^{\A}=y^{\A}$, what happens iff $f(x^{\A})=f(y^{\A})$ (since $f$ is injective), what in turns happens iff $x^{\Bp}=y^{\Bp}$. So $\Bp$ satisfies an atomic subformula of $\phi$ iff $\A$ satisfies the same subformula, and since $\phi$ is a quantifier-free formula in an empty signature (and therefore its truth value is entirely determined by the truth values of its atomic subformulas), $\phi$ is satisfied by $\Bp$. Given $|\s^{\Bp}|=|\s^{\B}|<\omega$, we have finished the proof.
\end{proof}

\begin{lemma}
    The $\Sigmaone$-theory $\TBB$ is stably finite, and thus has the finite model property, with respect to its only sort.
\end{lemma}

\begin{proof}
    Because of \Cref{propertiesofbb}, for any $n\in\mathbb{N}$, there exists an $m\in\mathbb{N}\setminus\{0,1\}$ such that $\bb(m)>n$, and so the $\TBB$-interpretation $\A$ with $\bb(m)$ elements satisfies $n<|\s^{\A}|<\omega$. We can then apply \Cref{Arbitrarymodels=>FMP} to obtain that $\TBB$ has the finite model property, and then use \Cref{OS+FMP=>SF} to obtain $\TBB$ is also stably finite.
\end{proof}

\begin{lemma}\label{TBB is not FW}
    The $\Sigmaone$-theory $\TBB$ is not finitely witnessable, and thus not strongly finitely witnessable, with respect to its only sort.
\end{lemma}

\begin{proof}
    Suppose, by contradiction, that $\TBB$ has a witness $\wit$, which is a computable function. We then define a function $f:\mathbb{N}\rightarrow\mathbb{N}$ by making $f(0)=0$ and
    $f(n+1)=|\vars_{\s}(\wit(\psi_{\geq f(n)+1}))|$.
    Clearly, $f$ is computable.

    But we also show by induction that, for all $n\in\mathbb{N}$, $f(n)\geq\bb(n)$, what proves $f$ is non computable given \Cref{propertiesofbb}, thus leading to a contradiction. It is clear that $f(0)\geq\bb(0)=0$; assuming we have successfully proved $f(n)\geq\bb(n)$, we know that $\psi_{\geq f(n)+1}$ is $\TBB$-satisfiable, since this theory has an infinite model, and therefore so is $\wit(\psi_{\geq f(n)+1})$. Given $\wit$ is a witness, there must exist a $\TBB$-interpretation $\A$ that satisfies $\wit(\psi_{\geq f(n)+1})$ with $\s^{\A}=\vars_{\s}(\wit(\psi_{\geq f(n)+1}))^{\A}$. But, by the definition of $f$, 
    \[f(n+1)=|\vars_{\s}(\wit(\psi_{\geq f(n)+1}))|\geq |\vars_{\s}(\wit(\psi_{\geq f(n)+1}))^{\A}|=|\s^{\A}|;\]
    and since $\A$ satisfies $\wit(\psi_{\geq f(n)+1})$, it must also satisfy $\Exists{\overarrow{x}}\wit(\psi_{\geq f(n)+1})$, for $\overarrow{x}=\vars(\wit(\psi_{\geq f(n)+1}))\setminus\vars(\psi_{\geq f(n)+1})$, and thus $\psi_{\geq f(n)+1}$, meaning $|\s^{\A}|\geq f(n)+1$. But $\A$, being a finite $\TBB$-interpretation, must satisfy $|\s^{\A}|=\bb(m)$ for some $m\in\mathbb{N}$; since 
    \[\bb(m)=|\s^{\A}|>f(n)\geq\bb(n),\]
    $m>n$ according to 
    \Cref{propertiesofbb}, and therefore $m\geq n+1$, proving $f(n+1)\geq |\s^{\A}|\geq \bb(n+1)$.
\end{proof}

\begin{lemma}
    The $\Sigmaone$-theory $\TBB$ is convex with respect to its only sort.
\end{lemma}

\begin{proof}
    $\TBB$ is stably infinite, and is defined on an empty signature: according to \Cref{SI empty theories are convex}, it is then convex.
\end{proof}

\subsection{$\TBBtwo$}\label{proofs of TBBtwo}

\begin{lemma}
    The $\Sigmatwo$-theory $\TBBtwo$ is stably infinite with respect to $\{\s, \s_{2}\}$.
\end{lemma}

\begin{proof}
    Obvious, since it is a theory over an empty signature, and there is a $\TBBtwo$-interpretation $\A$ with $|\s^{\A}|=|\s_{2}^{\A}|=\omega$.
\end{proof}

\begin{lemma}
    The $\Sigmatwo$-theory $\TBBtwo$ is not smooth with respect to $\{\s, \s_{2}\}$.
\end{lemma}

\begin{proof}
    Obvious, since there is a $\TBBtwo$-interpretation $\A$ with $|\s^{\A}|=|\s_{2}^{\A}|=4$, but no $\TBBtwo$-interpretations $\B$ with $|\s^{\B}|=5$ and $|\s^{\B}_{2}|=4$.
\end{proof}

\begin{lemma}\label{FMP of TBBtwo}
    The $\Sigmatwo$-theory $\TBBtwo$ has the finite model property with respect to $\{\s, \s_{2}\}$.
\end{lemma}

\begin{proof}
    Let $\phi$ be a quantifier-free $\Sigmatwo$-formula, and $\A$ a $\TBBtwo$-interpretation that satisfies $\phi$. Take then an $n\in\mathbb{N}\setminus\{0,1\}$ such that
    \[n\geq\max\{|\vars_{\s}(\phi)^{\A}|, |\vars_{\s_{2}}(\phi)^{\A}|\},\]
    which exists because both $\vars_{\s}(\phi)$ and $\vars_{\s_{2}}(\phi)$ are finite, and we know that $\bb(n)>n$ from \Cref{propertiesofbb} (since $n\in\mathbb{N}\setminus\{0,1\}$). Take then the $\TBBtwo$-interpretation $\B$ with $|\s^{\B}|=|\s_{2}^{\B}|=\bb(n)$ and, for injective functions $f:\vars_{\s}(\phi)^{\A}\rightarrow\s^{\B}$ and $f_{2}:\vars_{\s_{2}}(\phi)^{\A}\rightarrow\B$, $x^{\B}=f(x^{\A})$ and $u^{\B}=f_{2}(u^{\A})$ for $x\in\vars_{\s}(\phi)$ and $u\in\vars_{\s_{2}}(\phi)$ (and arbitrary for other variables). Then, for any atomic subformula $x=y$ or $u=v$ of $\phi$, where $x$ and $y$ are of sort $\s$, and $u$ and $v$ are of sort $\s_{2}$, we have that $\B$ satisfies $x=y$ (respectively $u=v$) iff $\A$ satisfies $x=y$ ($u=v$), meaning $\B$ satisfies $\phi$. Of course, both $|\s^{\B}|$ and $|\s_{2}^{\B}|$ equal $\bb(n)$, and are thus finite, finishing the proof.
\end{proof}

\begin{lemma}
    The $\Sigmatwo$-theory $\TBBtwo$ is not stably finite with respect to $\{\s, \s_{2}\}$.
\end{lemma}

\begin{proof}
    Consider the $\Sigmatwo$-formula $x=x$, where $x$ is a variable of sort $\s$, that is satisfied by the $\TBBtwo$-interpretation $\A$ with $|\s^{\A}|=1$ and $|\s_{2}^{\A}|=\omega$. But there is no $\TBBtwo$-interpretation $\B$ that satisfies $x=x$ with $|\s^{\B}|\leq|\s^{\A}|$, $|\s_{2}^{\B}|\leq|\s_{2}^{\A}|$ and both finite, since there are no $\TBBtwo$-interpretations $\B$ with $|\s^{\B}|=1$ and $|\s_{2}^{\B}|$ finite, given that $\bb(k+2)\geq 4$ for $k\in\mathbb{N}$.
\end{proof}

\begin{lemma}\label{not FW of TBBtwo}
    The $\Sigmatwo$-theory $\TBBtwo$ is not finitely witnessable, and thus not strongly finitely witnessable, with respect to $\{\s, \s_{2}\}$.
\end{lemma}

\begin{proof}
    We proceed as in the proof of \Cref{TBB is not FW}, supposing, by contradiction, that $\TBBtwo$ has a witness $\wit$: we then define a function $f:\mathbb{N}\rightarrow{N}$ by making $f(0)=0$, $f(1)=1$ and
    $f(n+1)=|\vars_{\s}(\wit(\psi^{\s}_{\geq f(n)+1}))|$.
\end{proof}

\begin{lemma}\label{TBBtwo is CV}
    The $\Sigmatwo$-theory $\TBBtwo$ is convex with respect to $\{\s, \s_{2}\}$.
\end{lemma}

\begin{proof}
    Since $\TBBtwo$ is stably infinite and $\Sigmatwo$ is empty, \Cref{SI empty theories are convex} guarantees it is convex.
\end{proof}

\subsection{$\TBBn$}\label{TBBn}

\begin{lemma}
    The $\Sigmatwo$-theory $\TBBn$ is not stably infinite, and thus not smooth, with respect to $\{\s, \s_{2}\}$ for any $n\geq 1$.
\end{lemma}

\begin{proof}
    Obvious, since $\TBBn$ is not contradictory, but there are no $\TBBn$-interpretations $\A$ such that $|\s^{\A}|$ and $|\s_{2}^{\A}|$ are both infinite.
\end{proof}

\begin{lemma}
    The $\Sigmatwo$-theory $\TBBn$ is stably finite, and thus has the finite model property, with respect to $\{\s, \s_{2}\}$ for any $n\geq 1$.
\end{lemma}

\begin{proof}
    Let $\phi$ be a quantifier-free $\Sigmatwo$-formula, and $\A$ a $\TBBn$-interpretation that satisfies $\phi$. Take then $\bb(m)$ as the minimum value of $\bb(k)$ such that
    \[\bb(k)\geq\min\{2,\quad |\vars_{\s_{2}}(\phi)^{\A}|\}.\]
    Take then the $\TBBn$-interpretation $\B$ with $\s^{\B}=\s^{\A}$, $|\s_{2}^{\B}|=\bb(m)$ and, for an injective function  $f_{2}:\vars_{\s_{2}}(\phi)^{\A}\rightarrow\B$, $x^{\B}=x^{\A}$ and $u^{\B}=f_{2}(u)$ for $x\in\vars_{\s}(\phi)$ and $u\in\vars_{\s_{2}}(\phi)$ (and arbitrary for other variables). This way, $\B$ satisfies $\phi$, and of course, $|\s^{\B}|=|\s^{\A}|$ and $|\s_{2}^{\B}|\leq|\s_{2}^{\A}|$ are both finite, finishing the proof.
\end{proof}

\begin{lemma}
    The $\Sigmatwo$-theory $\TBBn$ is not finitely witnessable, and thus not strongly finitely witnessable, with respect to $\{\s, \s_{2}\}$ for any $n\geq 1$.
\end{lemma}

\begin{proof}
    We proceed as in the proof of \Cref{TBB is not FW}, supposing, by contradiction, that $\TBBn$ has a witness $\wit$: we then define a function $f:\mathbb{N}\rightarrow{N}$ by making $f(0)=0$, $f(1)=1$ and
    $f(n+1)=|\vars_{\s_{2}}(\wit(\psi^{\s_{2}}_{\geq f(n)+1}))|$.
\end{proof}

\begin{lemma}
    The $\Sigmatwo$-theory $\TBBn$ is not convex with respect to $\{\s, \s_{2}\}$ for $n>1$.
\end{lemma}

\begin{proof}
    Take $n+1$ variables $x_{1}$ trough $x_{n+1}$ of sort $\s$, and it is clear that $\dash_{\TBBn}\bigvee_{i=1}^{n}\bigvee_{j=i+1}^{n+1}x_{i}=x_{j}$; we cannot, however, derive $\dash_{\TBBn}x_{i}=x_{j}$ for any $1\leq i<j\leq n+1$ since there is a $\T$-interpretation $\A$ with $|\s^{\A}|=n\geq 2$, implying $\TBBn$ is not convex.
\end{proof}

\subsubsection{$\TBBone$}

As proved in \Cref{TBBn}, $\TBBone$ is neither stably infinite, smooth, finitely witnessable nor strongly finitely witnessable, but it is stably finite and has the finite model property with respect to $\{\s, \s_{2}\}$.

\begin{lemma}
    The $\Sigmatwo$-theory $\TBBone$ is convex with respect to $\{\s, \s_{2}\}$.
\end{lemma}

\begin{proof}
    If we disregard the second sort of $\TBBone$, it becomes simply $\TBB$: given this theory is stably-infinite in an empty signature, from \Cref{SI empty theories are convex} we obtain it is convex. 
\end{proof}

\subsection{$\TmnBB$}

\begin{lemma}
    The $\Sigmatwo$-theory $\TmnBB$ is not stably infinite, and thus not smooth, with respect to $\{\s,\s_{2}\}$, for all $m,n\in\mathbb{N}$.
\end{lemma}

\begin{proof}
    Obvious, given $\TmnBB$ is not contradictory, but has no interpretations $\A$ where both $|\s^{\A}|$ and $|\s_{2}^{\A}|$ are infinite.
\end{proof}

\begin{lemma}
    The $\Sigmatwo$-theory $\TmnBB$ has the finite model property with respect to $\{\s,\s_{2}\}$ for all $m,n\in\mathbb{N}$.
\end{lemma}

\begin{proof}
    Let $\phi$ be a quantifier-free $\Sigmatwo$-formula, and $\A$ a $\TmnBB$-interpretation that satisfies $\phi$. Take sets $A$ and $B$, disjoint from each other and $\s^{\A}$ and $\s_{2}^{\A}$, with $\max\{m,n\}-|\vars_{\s}(\phi)^{\A}|$ and $\bb(k)-|\vars_{\s_{2}}(\phi)^{\A}|$ elements, respectively, where $k$ is the minimum natural number (greater than $1$) such that $\bb(k)\geq |\vars_{\s_{2}}(\phi)^{\A}|$. We then define a $\TmnBB$-interpretation $\B$ by making: $\s^{\B}=\vars_{\s}(\phi)^{\A}\cup A$ (which has $\max\{m,n\}$ elements); $\s_{2}^{\B}=\vars_{\s_{2}}(\phi)^{\A}\cup B$ (that has, then, $\bb(k)$ elements); $x^{\B}=x^{\A}$ and $u^{\B}=u^{\A}$ for all variables $x$ and $u$ in $\phi$ of sorts, respectively, $\s$ and $\s_{2}$; and arbitrarily for other variables. It is then clear that $\B$ is a $\TmnBB$-interpretation, with both $|\s^{\B}|$ and $|\s_{2}^{\B}|$ finite, that satisfies $\phi$.
\end{proof}

\begin{lemma}
    The $\Sigmatwo$-theory $\TmnBB$ is not stably finite with respect to $\{\s,\s_{2}\}$ for all $m,n\in\mathbb{N}$ such that $m\neq n$.
\end{lemma}

\begin{proof}
    Obvious, since $x=x$, for $x$ a variable of sort $\s$, is satisfied by the $\TmnBB$-interpretation $\A$ with $|\s^{\A}|=\min\{m,n\}$ and $|\s_{2}^{\A}|=\omega$, but there are no $\TmnBB$-interpretations with $|\s^{\B}|$ and $|\s_{2}^{\B}|$ finite, and $|\s^{\B}|\leq|\s^{\A}|$ and $|\s_{2}^{\B}|\leq|\s_{2}^{\B}|$. 
\end{proof}

\begin{lemma}
    The $\Sigmatwo$-theory $\TmnBB$ is not finitely witnessable, and thus not strongly finitely witnessable, with respect to $\{\s,\s_{2}\}$ for all $m,n\in\mathbb{N}$.
\end{lemma}

\begin{proof}
    Suppose $\TmnBB$ has a (necessarily computable) witness $\wit$, and consider the function $f:\mathbb{N}\rightarrow\mathbb{N}$ such that $f(0)=0$, $f(1)=1$ and, for $k>1$,
    \[f(k)=|\vars(\wit(\delta_{k}))|,\quad\text{where}\quad\delta_{k}=\bigwedge_{1\leq i<j\leq \max\{m,n\}}\neg(x_{i}=x_{j})\wedge\bigwedge_{1\leq i<j\leq f(k)+1}\neg(u_{i}=u_{j}),\]
    where $x_{1}$ through $x_{\max\{m,n\}}$ are variables of sort $\s$, and the $u_{i}$ are variables of sort $\s_{2}$: we prove by induction that $f(k)\geq\bb(k)$, and thus $f$ is not computable, leading to a contradiction with its definition that uses only computable notions. This is obvious for $k=0$ and $k=1$, so assume it holds for an arbitrary $k$. Since $\delta_{k+1}$ holds (for all $k\in\mathbb{N}$) in some interpretation on the $\TsBB$-model with $\max\{m,n\}$ elements in the domain of sort $\s$ and $\omega$ elements in the domain of sort $\s_{2}$, so does $\wit(\delta_{k+1})$: hence there must exist a $\TsBB$-interpretation $\A$ that satisfies $\wit(\delta_{k+1})$ (and thus $\delta_{k+1}$) with $\s^{\A}=\vars(\wit(\delta_{k+1}))^{\A}$ and $\s_{2}^{\A}=\vars_{\s_{2}}(\wit(\delta_{k+1}))^{\A}$. So $|\s^{\A}|\geq\max\{m,n\}$ (and thus $|\s^{\A}|=\max\{m,n\}$), and $|\s_{2}^{\A}|\geq f(k)+1\geq \bb(k)+1$, forcing $f(k+1)$ to be equal to or greater than $\bb(k+1)$, as we wished to prove.
\end{proof}

\begin{lemma}
    The $\Sigmatwo$-theory $\TmnBB$ is not convex with respect to $\{\s, \s_{2}\}$ if $\max\{m,n\}>1$.
\end{lemma}

\begin{proof}
    Take $\max\{m,n\}+1$ variables $x_{1}$ trough $x_{\max\{m,n\}+1}$ of sort $\s$: it is then true that 
    \[\dash_{\TmnBB}\bigvee_{i=1}^{\max\{m,n\}-1}\bigvee_{j=i+1}^{\max\{m,n\}}x_{i}=x_{j},\]
    but we cannot get $\dash_{\TmnBB}x_{i}=x_{j}$ for any $1\leq i<j\leq \max\{m,n\}+1$ since there is a $\TmnBB$-interpretation $\A$ with $|\s^{\A}|=\max\{m,n\}>1$ and $|\s_{2}^{\A}|=\omega$.
\end{proof}

\subsection{$\TsBB$}\label{proof of TsBB}

\begin{lemma}\label{TsBB is SI}
    The $\Sigmas$-theory $\TsBB$ is smooth, and thus stably infinite, with respect to $\{\s\}$.
\end{lemma}

\begin{proof}
    Let $\phi$ be a quantifier-free $\Sigmas$-formula, $\A$ a $\TsBB$-interpretation that satisfies $\phi$, and $\kappa(\s)$ a cardinal equal to or grater than $|\s^{\A}|$. There are then two cases to consider.

    \begin{enumerate}
        \item If $\kappa(\s)$ is infinite, consider a set $A$ with cardinality $\kappa(\s)$ and disjoint from $\s^{\A}$: we define a $\TsBB$-interpretation $\B$ by making $\s^{\B}=\s^{\A}\cup A$  (which has cardinality $\kappa(\s)$); $s^{\B}(a)=s^{\A}(a)$ for all $a\in\s^{\A}$, and $s^{\B}(a)=a$ for all $a\in A$ (meaning $\B$ has infinite elements $a$ satisfying $s^{\B}(a)=a$, and thus is indeed a $\TsBB$-interpretation); and $x^{\B}=x^{\A}$ for all variables $x$ (meaning $\B$ satisfies $\phi$).
        \item If $\kappa(\s)$ is finite, so is $|\s^{\A}|$: take then sets $A$ and $B$, disjoint from each other and $\s^{\A}$, with $\qq(\kappa(\s))-\qq(|\s^{\A}|)$ and $\kappa(s)-(|\s^{\A}|+|A|)$ elements, respectively. We then define an interpretation $\B$ by making:
        \begin{enumerate}
            \item $\s^{\B}=\s^{\A}\cup A\cup B$ (which has $\kappa(\s)$ elements); 
            \item $s^{\B}(a)=s^{\A}(a)$ for all $a\in\s^{\A}$, $s^{\B}(a)=a$ for all $a\in A$, and $s^{\B}(a)$ any element of $\s^{\B}\setminus\{a\}$ for $a\in B$ (notice that $\B$ possesses $\qq(|\s^{\A}|)$ elements in $\s^{\A}$ where $s^{\B}$ is the identity, and $\qq(\kappa(\s))-\qq(|\s^{\A}|)$ more on $A$, to a total of $\qq(\kappa(\s))$, hat makes of $\B$ a $\TsBB$-interpretation);
            \item and $x^{\B}=x^{\A}$ for all variables $x$ (meaning $\B$ satisfies $\phi$).
            \end{enumerate}
    \end{enumerate}
\end{proof}

\begin{lemma}\label{TsBB has FMP}
    The $\Sigmas$-theory $\TsBB$ is stably finite, and thus has the finite model property, with respect to $\{\s\}$.
\end{lemma}

\begin{proof}
    Because of \Cref{OS+FMP=>SF}, we only need to prove that $\TsBB$ has the finite model property. So take a quantifier-free $\Sigmas$-formula $\phi$ and an infinite $\TsBB$-interpretation $\A$ that satisfies $\phi$. Let $\vars(\phi)=\{z_{1}, \ldots , z_{n}\}$, and for each $z_{i}$, let $M_{i}$ be the maximum of $j$ such that the term $s^{j}(z_{i})$ occurs in $\phi$. Take then
    \[\alpha(\phi)^{\A}=\{(s^{\A})^{j}(z_{i}^{\A}): 1\leq i\leq n, 0\leq j\leq M_{i}\},\quad Id^{\A}(\phi)=\{a\in \alpha(\phi)^{\A} : s^{\A}(a)=a\},\]
    \[\alpha(\phi)_{+1}^{\A}=\{(s^{\A})^{j}(z_{i}^{\A}): 1\leq i\leq n, 0\leq j\leq M_{i}+1\},\]
    $p=|\alpha(\phi)^{\A}|$, $q=|\alpha(\phi)_{+1}^{\A}|$, $r=|Id^{\A}(\phi)|$, and $m=\max\{\bb(-p+q+r), q\}$. We finally take sets $A$ and $B$, disjoint from each other and $\s^{\A}$, with, respectively, $\qq(m)+p-q-r$ elements and $m-p+r-\qq(m)$ elements. We then define an interpretation $\B$ with: 
    \begin{enumerate}
        \item $\s^{\B}=\alpha(\phi)_{+1}^{\A}\cup A\cup B$ (which has $m$ elements); 
        \item $s^{\B}(a)=\s^{\A}(a)$ for all $a\in \alpha(\phi)^{\A}$ (which has $r$ elements that satisfy $s^{\B}(a)=a$, namely those in $Id^{\A}(\phi)$), $\s^{\A}(a)=a$ for all $a\in A$ or $a\in \alpha(\phi)_{+1}^{\A}\setminus\alpha(\phi)^{\A}$ (there are, respectively, $\qq(m)+p-q-r$ and $q-p$, leading to a total of $\qq(m)$ elements in $\B$ where $\s^{\B}$ is the identity, meaning $\B$ is indeed a $\TsBB$-interpretation), and $s^{\B}(a)$ any value in $\s^{\B}\setminus\{a\}$ for $a\in B$;
        \item $x^{\B}=x^{\A}$ for all variables $x$ in $\phi$, and otherwise arbitrary.
        \end{enumerate}
        It is obvious that $\B$ is finite, and furthermore one can show that it additionally satisfies $\phi$.
\end{proof}

\begin{lemma}\label{TsBB is not FW}
    The $\Sigmas$-theory $\TsBB$ is not finitely witnessable, and thus not strongly finitely witnessable, with respect to $\{\s\}$.
\end{lemma}

\begin{proof}
    Suppose $\TsBB$ indeed has a witness $\wit$ and consider the function $f:\mathbb{N}\rightarrow\mathbb{N}$ such that $f(0)=0$ and, for $n>0$,
    \[f(n)=|\vars(\wit(\delta_{n}))|,\quad\text{where}\quad\delta_{n}=\bigwedge_{1\leq i<j\leq n}\neg(x_{i}=x_{j})\wedge\bigwedge_{i=1}^{n}(s(x_{i})=x_{i}):\]
    since $\delta_{n}$ holds in some infinite $\TsBB$-interpretations, so does $\wit(\phi)$, and thus there must exist a $\TsBB$-interpretation $\A$ that satisfies $\wit(\delta_{n})$ (and thus $\delta_{n}$) with $\s^{\A}=\vars(\wit(\delta_{n}))^{\A}$. Then, there i a finite number $m\geq n$ of elements in $\A$ satisfying $s^{\A}(a)=a$, meaning that $\A$ satisfies $\psi_{=m}^{=}$ and thus $\psi_{\geq\bb(m)}$. In other words, $f(n)\geq\bb(m)\geq \bb(n)$, and so $f$ cannot be computable, contradicting its definition.
\end{proof}

\begin{lemma}\label{satisfiability in TsBB}
    Let $\T$ be the $\Sigmas$-theory with all $\Sigmas$-structures as models: if the quantifier-free $\Sigmas$-formula $\phi$ is $\T$-satisfiable, it is $\TsBB$-satisfiable.
\end{lemma}

\begin{proof}
    Let $\phi$ be a quantifier-free, $\T$-satisfiable formula, and let $\A$ be a $\T$-interpretation that satisfies $\phi$; because $\T$ is stably infinite, we can obtain a $\T$-interpretation $\B$ that satisfies $\phi$ with $|\s^{\B}|\geq \omega$. Then let $m$ be the number of elements $a$ in $\s^{\B}$ such that $s^{\B}(a)=a$, and $n$ be the number of elements $a$ such that $s^{\B}(a)\neq a$.

    Because $|\s^{\B}|\geq\omega$, either $m$ or $n$ is infinite, and we define $\C$ by making $\s^{\C}=\s^{\B}\cup\{a_{n}:n\in\mathbb{N}\}\cup\{b_{n}:n\in\mathbb{N}\}$, where the latter sets are disjoint from $\s^{\B}$ and each other; $s^{\C}(a)=s^{\B}(a)$ for all $a\in \s^{\B}$, $s^{\C}(a_{n})=a_{n}$ and $s^{\C}(b_{n})=a_{n}$; and $x^{\C}=x^{\B}$ for all variables $x$. Then $\C$ not only satisfies $\phi$, but is also a $\TsBB$-interpretation, meaning $\phi$ is also $\TsBB$-satisfiable.
\end{proof}

\begin{lemma}\label{TsBB is CV}
    The $\Sigmas$-theory $\TsBB$ is convex with respect to $\{\s\}$.
\end{lemma}

\begin{proof}
    For a proof by contradiction, take a conjunction of $\Sigmas$-literals $\phi$ and assume that, although $\dash_{\TsBB}\phi\rightarrow\bigvee_{i=1}^{n}x_{i}=y_{i}$, we have $\not\dash_{\TsBB}\phi\rightarrow x_{i}=y_{i}$ for each $1\leq i\leq n$. So we can find $\TsBB$-interpretations $\A_{i}$ that satisfy, simultaneously, $\phi$ and $\neg(x_{i}=y_{i})$.

    Now, we state that $\tdash\phi\rightarrow\bigvee_{i=1}^{n}x_{i}=y_{i}$, for $\T$ the $\Sigmas$-theory of the uninterpreted function $s$: if that were not true, we should be able to find a $\T$-interpretation $\A$ that satisfies $\phi$ but not $\bigvee_{i=1}^{n}x_{i}=y_{i}$, meaning $\phi\wedge\neg\bigvee_{i=1}^{n}x_{i}=y_{i}$ is $\T$-satisfiable, and by \Cref{satisfiability in TsBB} that would imply this formula is $\TsBB$-satisfiable, contradicting our assumptions. And, because all $\TsBB$-interpretations are $\T$-interpretations, we get that $\phi\wedge\neg(x_{i}=y_{i})$ is $\T$-satisfiable for each $1\leq i\leq n$. 

    Given \Cref{uninterpretedfunctionsis convex} and that $\tdash\phi\rightarrow\bigvee_{i=1}^{n}x_{i}=y_{i}$, there must exist $1\leq i\leq n$ such that $\tdash\phi\rightarrow x_{i}=y_{i}$, contradicting that $\phi\wedge\neg(x_{i}=y_{i})$ is $\T$-satisfiable. The conclusion must be that $\TsBB$ is convex.
\end{proof}

\subsection{$\TneqBB$}

\begin{lemma}
    The $\Sigmas$-theory $\TneqBB$ is not stably infinite, and thus not smooth, with respect to $\{\s\}$.
\end{lemma}

\begin{proof}
 Any $\TneqBB$-interpretation $\A$ with $|\s^{\A}|=2$ satisfies $\neg(s(x)=x)$, while all infinite $\TneqBB$-interpretations satisfy instead $s(x)=x$.
 \end{proof}

 \begin{lemma}
    The $\Sigmas$-theory $\TneqBB$ is stably finite, and thus has the finite model property, with respect to $\{\s\}$.
\end{lemma}

\begin{proof}
    Let $\phi$ be a quantifier-free $\Sigmas$-formula, and $\A$ a $\TneqBB$-interpretation that satisfies $\phi$: we may assume $\A$ is infinite, since otherwise there would be nothing to prove, and so $\A$ satisfies $\Forall{x}(s(x)=x)$. Take then the smallest $m$ such that $\bb(m)\geq\max\{|\vars(\phi)^{\A}|, 3\}$, and a set $A$ disjoint from $\s^{\A}$ with $\bb(m)-|\vars(\phi)^{\A}|$ elements. We define the interpretation $\B$ by: $\s^{\B}=\vars(\phi)^{\A}\cup A$ (so $\B$ has at least $\bb(m)\geq 3$ elements, but $|\s^{\B}|\leq|\s^{\A}|$); $s^{\B}(a)=a$ for all $a\in\s^{\B}$ (so $\B$ satisfies $\Forall{x}(s(x)=x)$ and is, therefore, a $\TneqBB$-interpretation); $x^{\B}=x^{\A}$ for all variables $x$ in $\phi$ (so $\B$ satisfies $\phi$); and arbitrarily for other variables.
\end{proof}

\begin{lemma}
    The $\Sigmas$-theory $\TneqBB$ is not finitely witnessable, and thus not strongly finitely witnessable, with respect to $\{\s\}$.
\end{lemma}

\begin{proof}
The proof is very similar to that of \Cref{TBB is not FW}: if $\TneqBB$ indeed has a witness $\wit$, we define a function $f:\mathbb{N}\rightarrow{N}$ by making $f(0)=0$, $f(1)=1$, and for $n\geq 1$
$f(n+1)=|\vars_{\s}(\wit(\psi_{\geq f(n)+1}))|$.
Because $\wit$ is computable, so must be $f$, but one can prove that $f$ grows at least as fast as $\bb$, so we reach a contradiction.
\end{proof}

\begin{lemma}
    The $\Sigmas$-theory $\TneqBB$ is not convex with respect to $\{\s\}$.
\end{lemma}

\begin{proof}
    Take variables $x$, $y$, $w$ and $z$, and consider the conjunction of literals $\phi$ equal to $\neg(s(x)=x)\wedge\neg(y=w)$: it is then true that $\dash_{\TneqBB}\phi\rightarrow(y=z)\vee(w=z)$, since $\neg(s(x)=x)$ is only satisfied in a $\TneqBB$-interpretation $\A$ with $|\s^{\A}|=2$, and by the pigeonhole principle, if $\neg(y=w)$, then either $y=z$ or $w=z$. But we cannot get either $\dash_{\TneqBB}\phi\rightarrow y=z$ nor $\dash_{\TneqBB}\phi\rightarrow w=z$, since $\A$ has $|\s^{\A}|=2$, and we can change the value assigned to $z$ while keeping $\phi$ satisfied.
\end{proof}

\subsection{$\TneqBBone$}

\begin{lemma}
    The $\Sigmas$-theory $\TneqBBone$ is not stably infinite, and thus is not smooth, with respect to $\{\s\}$.
\end{lemma}

\begin{proof}
    Obvious, since the $\TneqBBone$-interpretation with one element in its domain satisfies $s(x)=x$, while all infinite $\TneqBBone$-interpretation satisfy instead $\neg(s(x)=x)$.
\end{proof}

\begin{lemma}
    The $\Sigmas$-theory $\TneqBBone$ is stably finite, and thus has the finite model property, with respect to $\{\s\}$.
\end{lemma}

\begin{proof}
    Because of \Cref{OS+FMP=>SF}, we just need to prove $\TneqBBone$ has the finite model property. We then take a 
    quantifier-free $\Sigmas$-formula $\phi$, and a $\TneqBBone$-interpretation $\A$ that satisfies $\phi$, and without loss of generality we assume that $\A$ has at least $4$ elements, and thus $s^{\A}$ is never the identity. If $x_{1}$ through $x_{n}$ are the variables in $\phi$, and $M_{i}$ is the maximum of the values $j$ such that $s^{j}(x_{i})$ appears in $\phi$, we define 
    \[A=\{s^{j}(x_{i}) : 1\leq i\leq n, 0\leq j\leq M_{i}\}^{\A},\quad A^{+}=\{s^{j}(x_{i}) : 1\leq i\leq n, 0\leq j\leq M_{i}+1\}^{\A}\]
    and $m=|A^{+}|$: then we take the least $M$ on the image of $\bb$ such that $M\geq m$; we also take a set $B$ with $M-m$ elements disjoint from $\s^{\A}$. We then define an interpretation $\B$ by making: $\s^{\B}=A^{+}\cup B$; $s^{\B}(a)=s^{\A}(a)$ for all $a\in A$, and $s^{\B}(a)$ any value different from 
    $a$, for each 
    $a\in B$ or $a\in A^{+}\setminus A$; $x^{\B}=x^{\A}$ for each variable $x$ of $\phi$, and arbitrarily otherwise. $\B$ has then $M$ (finite) number of elements, that is a value of $\bb$, and satisfies $\neg(s(x)=x)$ for all $x$, meaning it is a $\TneqBBone$-interpretation; furthermore, it satisfies $\phi$, and so we are done.
\end{proof}

\begin{lemma}
    The $\Sigmas$-theory $\TneqBBone$ is not finitely witnessable, and thus is not strongly finitely witnessable, with respect to $\{\s\}$.
\end{lemma}

\begin{proof}
    This follows \Cref{TBB is not FW}: suppose we have a witness $\wit$, necessarily computable, and define a function $f:\mathbb{N}\rightarrow\mathbb{N}$ such that $f(0)=0$, $f(1)=1$ and, assuming $f(n)$ defined, 
    \[f(n+1)=|\vars_{\s}(\wit(\bigwedge_{i=1}^{f(n)}\bigwedge_{j=i+1}^{f(n)+1}\neg(x_{i}=x_{j})))|,\]
    which is clearly computable. The contradiction is that $f$ grows at least as fast as $\bb$, and by \Cref{propertiesofbb} cannot be computable.
\end{proof}

\begin{lemma}
    The $\Sigmas$-theory $\TneqBBone$ is convex with respect to $\{\s\}$.
\end{lemma}

\begin{proof}
    This is a very long and dull proof: however,  an essentially equal proof may be found in \cite{BarTolZoh} when we prove that $\Tneqodd$ is convex.
\end{proof}

\subsection{$\TssBB$}

The proofs that $\TneqBBone$ is smooth (and thus stably infinite) and stably finite (and that it therefore has the finite model property), but not finitely witnessable (and thus is not strongly finitely witnessable) are very similar to those of the correspondent properties of $\TsBB$, in \Cref{proof of TsBB}.

\begin{lemma}\label{TssBB is not CV}
    The $\Sigmas$-theory $\TssBB$ is not convex with respect to $\{\s\}$.
\end{lemma}

\begin{proof}
    Because $\TssBB$ satisfies $\psiv$, given the cube $\phi$ equal to $(y=s(x))\wedge(z=s(y))$, we have that $\dash_{\TssBB}\phi\rightarrow (x=z)\vee(y=z)$. Take then the $\TssBB$-interpretations $\A$ and $\B$ with: $\s^{\A}=\{a, b, c, d\}$; $s^{\A}(a)=s^{\B}(a)=a$, $s^{\A}(b)=s^{\B}(b)=b$, $s^{\A}(c)=a$, $s^{\A}(d)=b$, $s^{\B}(c)=d$ and $s^{\B}(d)=c$; $x^{\A}=c$, $y^{\A}=z^{\A}=a$, $x^{\B}=z^{\B}=c$ and $y^{\B}=d$.

    Then both $\A$ and $\B$ satisfy $\phi$, but $\A$ does not satisfy $x=z$, while $\B$ does not satisfy $y=z$.
\end{proof}

\subsection{$\TtwoBB$}\label{proofs of TtwoBB}

\begin{lemma}\label{TtwoBB is SM}
    The $\Sigmastwo$-theory $\TtwoBB$ is smooth, and thus stably infinite, with respect to $\{\s, \s_{2}\}$.
\end{lemma}

\begin{proof}
    Take a quantifier-free $\Sigmastwo$-formula $\phi$, a $\TtwoBB$-interpretation $\A$ that satisfies $\phi$, and cardinals $\kappa(\s)\geq|\s^{\A}|$ and $\kappa(\s_{2})\geq |\s_{2}^{\A}|$; take as well sets $A$ and $B$ with, respectively, $\kappa(\s)-|\s^{\A}|$ and $\kappa(\s_{2})-|\s_{2}^{\A}|$ elements, disjoint from each other and the domains of $\A$.

    We the define the interpretation $\B$ as follows: $\s^{\B}=\s^{\A}\cup A$ and $\s_{2}^{\B}=\s_{2}^{\A}\cup B$ (so $|\s^{\B}|\geq |\s^{\A}|$); $s^{\B}(a)=s^{\A}(a)$ for each $a\in \s^{\A}$, and $\s^{\B}(a)$ is anything but $a$ for $a\in A$ (this way, the number $k+1$ of elements in $\s^{\B}$ satisfying $s^{\B}(a)=a$ is the same as the number of elements in $\s^{\A}$ satisfying $s^{\A}(a)=a$; since $|\s^{\A}|\geq \bb(k+1)$, $|\s^{\B}|\geq \bb(k+1)$, and thus $\B$ is a $\TtwoBB$-interpretation); and $x^{\B}=x^{\A}$ for every variable $x$. This way $\B$ satisfies $\phi$, $|\s^{\B}|=\kappa(\s)$ and $|\s_{2}^{\B}|=\kappa(\s_{2})$.
\end{proof}

\begin{lemma}\label{TtwoBB is not FW}
    The $\Sigmastwo$-theory $\TtwoBB$ is not finitely witnessable, and thus not strongly finitely witnessable, with respect to $\{\s, \s_{2}\}$.
\end{lemma}

\begin{proof}
    Suppose we have a witness $\wit$. We define a function $f:\mathbb{N}\rightarrow\mathbb{N}$ by making $f(0)=0$, $f(1)=1$, and assuming $f(n)$ defined, 
    \[f(n+1)=|\vars_{\s_{2}}(\wit(\bigwedge_{i=1}^{n+1}(s(x_{i})=x_{i})\wedge\bigwedge_{i=1}^{n}\bigwedge_{j=i+1}^{n+1}\neg(x_{i}=x_{j})))|,\]
    which is computable: after all, producing the formula inside the witness is easily seen to be computable, as well as is producing its
    witness (given $\wit$ is computable), its set of variables and the cardinality of that set. So $f$ is a computable function.

    However, since $\delta_{n+1}=\bigwedge_{i=1}^{n+1}(s(x_{i})=x_{i})\wedge\bigwedge_{i=1}^{n}\bigwedge_{j=i+1}^{n+1}\neg(x_{i}=x_{j}$ is only satisfied in a $\TtwoBB$-interpretation $\A$ that satisfies 
    $\psi^{=}_{\geq n+1}$, $\wit(\delta_{n+1})$ is only satisfied in a $\TtwoBB$-interpretation $\A^{\prime}$ with $|\s^{\A^{\prime}}|\geq\bb(n+1)$. Because there must exist a $\TtwoBB$-interpretation $\B$ that satisfies 
    $\wit(\delta_{n+1})$ with $\s^{\B}=\vars_{\s}(\wit(\delta_{n+1}))^{\B}$, we obtain that $f(n+1)=|\vars_{\s}(\wit(\delta_{n+1}))|\geq |\s^{\B}|\geq\bb(n+1)$, and by 
    \Cref{propertiesofbb} $f$ cannot be computable, leading to a contradiction.
\end{proof}

\begin{lemma}\label{satisfiability in TtwoBB}
    Let $\T$ be the $\Sigmastwo$-theory of uninterpreted functions: then a quantifier-free $\Sigmastwo$-formula $\phi$ is $\TtwoBB$-satisfiable if it is $\T$-satisfiable.
\end{lemma}

\begin{proof}
    See \Cref{satisfiability in TsBB} for a similar proof.
\end{proof}

\begin{lemma}\label{TtwoBB is CV}
    The $\Sigmastwo$-theory $\TtwoBB$ is convex with respect to $\{\s, \s_{2}\}$.
\end{lemma}

\begin{proof}
    Given \Cref{uninterpretedfunctionsis convex}, it is easy to prove that the $\Sigmastwo$-theory of uninterpreted functions is also convex; by applying \Cref{satisfiability in TtwoBB}, we get that $\TtwoBB$ is convex much like in the proof of \Cref{TsBB is CV}.
\end{proof}

\begin{lemma}\label{TtwoBB is SF}
    The $\Sigmastwo$-theory $\TtwoBB$ is stably finite, and thus has the finite model property, with respect to $\{\s, \s_{2}\}$.
\end{lemma}

\begin{proof}
    Take a quantifier-free $\Sigmastwo$-formula $\phi$, and a $\TtwoBB$-interpretation $\A$ that satisfies
    $\phi$. Let $\vars_{\s}(\phi)=\{x_{1}, \ldots , x_{n}\}$ and, for each $i\in[1,n]$, let $M_{i}$ be the maximum of $j$ such that $s^{j}(x_{i})$ appears in $\phi$. We also define
    \[A=\{s^{j}(x_{i}) : 1\leq i\leq n, 0\leq j\leq M_{i}\}^{\A}\quad\text{and}\quad A^{+}=\{s^{j}(x_{i}) : 1\leq i\leq n, 0\leq j\leq M_{i}+1\}^{\A}.\]
    Now, let $m$ be the maximum between $2$ and the number of elements $a$ in $A$ that satisfy $s^{\A}(a)=a$: notice it is equal to or less than the number of 
    elements $a$ in $\s^{\A}$ that satisfy $s^{\A}(a)=a$. Take as well a set $B$ with $\max\{0, \bb(m)-|\vars_{\s}(\phi)^{\A}|\}$ elements.

    We define an interpretation $\B$ such that: $\s^{\B}=\vars_{\s}(\phi)^{\A}\cup B$ and $\s_{2}^{\B}=A^{+}$ (so $|\s^{\B}|=\bb(m)$); $s^{\B}(a)=s^{\A}(a)$ for each $a\in A$, and $s^{\B}(a)$
    is anything but $a$ for $a\in A^{+}\setminus A$ (so the number of elements $a\in \s^{\B}$ satisfying 
    $s^{\B}(a)=a$ is $m$, making of $\B$ a $\TtwoBB$-interpretation); and $x^{\B}=x^{\A}$ for any variable $x$ in $\phi$, and arbitrary otherwise. 
    
    This way $\B$ satisfies $\phi$, and has both $\s^{\B}$ and $\s_{2}^{\B}$ finite. Furthermore, $\s_{2}^{\B}\subseteq \s_{2}^{\A}$, so $|\s_{2}^{\B}|\leq|\s_{2}^{\A}|$, and since the number of elements 
    $a\in\s^{\B}$ satisfying $s^{\B}(a)=a$ is $m$, equal to or less than the number $m^{\prime}$ of elements 
    $a\in\s^{\A}$ satisfying $s^{\A}(a)=a$, we have that $|\s^{\B}|=\bb(m)\leq \bb(m^{\prime})\leq |\s^{\A}|$, finishing the proof.
    
\end{proof}

\subsection{$\TBBtwotwo$}

All proofs for $\TBBtwotwo$ are very similar to those of $\TBBtwo$ in \Cref{proofs of TBBtwo}, with the exception of the proof that $\TBBtwotwo$ is not stably finite.

\begin{lemma}\label{TBBtwotwo is not SF}
    The $\Sigmastwo$-theory $\TBBtwotwo$ is not stably finite with respect to $\{\s, \s_{2}\}$.
\end{lemma}

\begin{proof}
    Obvious, since $\TBBtwotwo$ has interpretations $\A$ with $|\s^{\A}|=1$ and $|\s_{2}^{\A}|=\omega$, but no interpretations $\B$ with $|\s^{\B}|=1$ and $\s_{2}^{\B}$ finite.
\end{proof}

\subsection{$\TstwoBB$}

The proofs that $\TstwoBB$ is smooth (and therefore stably infinite), not finitely witnessable (and thus not strongly finitely witnessable) and has the finite model property are very similar to the corresponding proofs of $\TtwoBB$ in \Cref{proofs of TtwoBB}. The proof that $\TstwoBB$ is not convex follows the proof that $\TssBB$ is not convex, in \Cref{TssBB is not CV}; and, finally, the proof that $\TstwoBB$ is not stably finite, is almost the same as that of \Cref{TBBtwotwo is not SF}.

\subsection{$\TsBBtwo$}

\begin{lemma}
    The $\Sigmastwo$-theory $\TsBBtwo$ is not stably infinite, and thus not smooth, with respect to $\{\s, \s_{2}\}$.
\end{lemma}

\begin{proof}
    Obvious: while $s(x)=x$ is satisfied by any $\TsBBtwo$-interpretation $\A$ with $|\s^{\A}|=1$ and $|\s_{2}^{\A}|=\omega$, all $\TsBBtwo$-interpretations $\B$ with $|\s^{\A}|$ greater than one satisfy instead $\Forall{x}\neg(s(x)=x)$.
\end{proof}

\begin{lemma}
    The $\Sigmastwo$-theory $\TsBBtwo$ is not finitely witnessable, and thus not strongly finitely witnessable, with respect to $\{\s, \s_{2}\}$.
\end{lemma}

\begin{proof}
    The proof is similar to that of \Cref{TBB is not FW}: if there is a witness $\wit$, define a function $f:\mathbb{N}\rightarrow\mathbb{N}$ by making $f(0)=0$, $f(1)=1$, and assuming $f(n)$ defined, 
    \[f(n+1)=|\vars_{\s}(\wit(\bigwedge_{i=1}^{f(n)}\bigwedge_{j=i+1}^{f(n)+1}\neg(x_{i}=x_{j})))|,\]
    where $x_{1}$ through $x_{f(n)+1}$ are variables of sort $\s$. With this, $f$ is supposed to be computable, but grows faster than $\bb$, contradicting \Cref{propertiesofbb}.
\end{proof}

\begin{lemma}
    The $\Sigmastwo$-theory $\TsBBtwo$ is convex with respect to $\{\s, \s_{2}\}$.
\end{lemma}

\begin{proof}
    This proof is very similar to the one of \Cref{Tsupinfty is CV}. We start by 
    taking a $\Sigmastwo$-cube $\phi$, and assuming that $\dash_{\TsBBtwo}\phi\rightarrow\bigvee_{i=1}^{n}x_{i}=y_{i}$. We can also assume 
    that $\phi$ is true in some $\TsBBtwo$-interpretation $\A_{0}$ with 
    $|\s^{\A_{0}}|>1$: otherwise, given the theory whose interpretations are all 
    $\TsBBtwo$-interpretations $\A$ with $|\s^{\A}|=1$ is essentially $\Tinfty$, which is convex, we would be done; we 
    can take $|\s^{\A_{0}}|=\omega$ by \Cref{LowenheimSkolemDownwards}.

    Now, consider the theory whose interpretations are all $\TsBBtwo$-interpretations $\A$ with $|\s^{\A}|>1$, and since it is stably infinite, and 
    convex by \Cref{Tneg is CV}, we have that there exists $1\leq i\leq n$ such that $\phi\rightarrow x_{i}=y_{i}$ all $\TsBBtwo$-interpretations $\A$ with 
    $|\s^{\A}|>1$; in particular, by $\A_{1}$. But suppose there is a $\TsBBtwo$-interpretation $\A_{1}$ with $|\s^{\A_{1}}|=1$, and without loss of generality 
    $|\s_{2}^{\A_{1}}|=\omega$, that satisfies $\phi$ and not $x_{i}=y_{i}$: we then build a $\TsBBtwo$-interpretation $\B$ with $\s^{\B}=\s^{\A_{0}}$, 
    $\s_{2}^{\B}=\s_{2}^{\A_{1}}$ and $s^{\B}=s^{\A_{0}}$ that satisfies $\phi$ and not $x_{i}=y_{i}$, leading to a contradiction.
\end{proof}

\begin{lemma}
    The $\Sigmastwo$-theory $\TsBBtwo$ has the finite model property with respect to $\{\s, \s_{2}\}$.
\end{lemma}

\begin{proof}
    Follows the steps of \Cref{FMP of TBBtwo}.
\end{proof}

\begin{lemma}
    The $\Sigmastwo$-theory $\TsBBtwo$ is not stably finite with respect to $\{\s, \s_{2}\}$.
\end{lemma}

\begin{proof}
    Obvious: while there is a $\TsBBtwo$-interpretation $\A$ with $|\s^{\A}|=1$ and $|\s_{2}^{\A}|=\omega$, there are no $\Tsupinfty$-interpretations $\B$ with $|\s^{\B}|=1$ and $|\s_{2}^{\B}|<\omega$.
\end{proof}

\subsection{$\TBBtwocube$}

\begin{lemma}
    The $\Sigma_{3}$-theory $\TBBtwocube$ is not stably infinite, and thus not smooth, with respect to $\{\s, \s_{2}, \s_{3}\}$.
\end{lemma}

\begin{proof}
    Obvious, since it is not contradictory, but has no models $\A$ with $|\s_{3}^{\A}|$ infinite.
\end{proof}

\begin{lemma}\label{FMP of TBBtwocube}
    The $\Sigma_{3}$-theory $\TBBtwocube$ has the finite model property with respect to $\{\s, \s_{2}, \s_{3}\}$.
\end{lemma}

\begin{proof}
    Essentially the same proof as that of \Cref{FMP of TBBtwo}.
\end{proof}

\begin{lemma}
    The $\Sigma_{3}$-theory $\TBBtwocube$ is not stably infinite with respect to $\{\s, \s_{2}, \s_{3}\}$.
\end{lemma}

\begin{proof}
    Consider any quantifier-free $\Sigma_{3}$-formula $\phi$ that is satisfied by the $\TBBtwocube$-interpretation $\A$ with $|\s^{\A}|=|\s_{3}^{\A}|=1$ and $|\s_{2}^{\A}|=\omega$: since there are no 
    $\TBBtwocube$-interpretations $\B$ such that $|\s^{\B}|=1$ and $|\s_{2}^{\B}|$ is finite, $\TBBtwocube$ cannot be stably finite.
\end{proof}

\begin{lemma}
    The $\Sigma_{3}$-theory $\TBBtwocube$ is not finitely witnessable, and thus not strongly finitely witnessable, with respect to $\{\s, \s_{2}, \s_{3}\}$.
\end{lemma}

\begin{proof}
    See \Cref{not FW of TBBtwo} for an essentially equal proof.
\end{proof}

\begin{lemma}
    The $\Sigma_{3}$-theory $\TBBtwocube$ is convex with respect to $\{\s, \s_{2}, \s_{3}\}$.
\end{lemma}

\begin{proof} 
    Take a conjunction of literals $\phi$, and suppose $\dash_{\TBBtwocube}\phi\rightarrow\bigvee_{i=1}^{n}x_{i}=y_{i}$: if any pair $(x_{i}, y_{i})$ is of sort $\s_{3}$ we 
    are done, since then $\dash_{\TBBtwocube}\phi\rightarrow x_{i}=y_{i}$. So we may assume that some 
    all $x_{i}$ and $y_{i}$ are of sorts in $\{\s, \s_{2}\}$. We can also assume no literal of $\phi$ has variables of sort $\s_{3}$, since then $\phi$ is either a 
    tautology or a contradiction, cases very easy to consider.

    So $\phi$, under these assumptions, is actually a $\Sigma_{2}$-formula, as well as $\bigvee_{i=1}^{n}x_{i}=y_{i}$: not 
    only that, but $\dash_{\TBBtwo}\phi\rightarrow\bigvee_{i=1}^{n}x_{i}=y_{i}$; indeed, if there is 
    a $\TBBtwo$-interpretation $\A$ that satisfies $\phi$ but not 
    $\bigvee_{i=1}^{n}x_{i}=y_{i}$, the $\TBBtwocube$-interpretation $\B$ with $\s^{\B}=\s^{\A}$ and 
    $\s_{2}^{\B}=\s_{2}^{\A}$ would satisfy $\phi$ but not $\bigvee_{i=1}^{n}x_{i}=y_{i}$. But $\TBBtwo$ is convex, as proven in 
    \Cref{TBBtwo is CV}, and so $\dash_{\TBBtwo}\phi\rightarrow x_{i}=y_{i}$ for some $1\leq i\leq n$. 
    Of course it follows that $\dash_{\TBBtwocube}\phi\rightarrow x_{i}=y_{i}$, and so $\TBBtwocube$ is convex.
\end{proof}

\end{document}